\documentclass[11pt]{elsarticle}
\usepackage[caption = false]{subfig}
\usepackage[multiple]{footmisc}
\usepackage{eurosym}

\usepackage[final]{changes}

\usepackage{amsmath,amssymb,amsthm,mathrsfs,amsfonts,dsfont}

\usepackage{array}
\usepackage{epstopdf}
\usepackage{setspace}
\usepackage{algorithm,algorithmic}
\usepackage{natbib}
\usepackage{color}
\usepackage{rotating}
\usepackage{tikz}
\usepackage{subfig}
\usepackage[toc,page]{appendix}
\usepackage[font=small,skip=0pt]{caption}
\usepackage{eurosym}
\usepackage{hyperref}

\newcommand*\colvec[3][]{
	\begin{pmatrix}\ifx\relax#1\relax\else#1\\\fi#2\\#3\end{pmatrix}
}
\newcommand*\colvecfive[5][]{
	\begin{pmatrix}\ifx\relax#1\relax\else#1\\\fi#2\\#3 \\#4\\#5\end{pmatrix}
}

\newtheorem{remark}{Remark}
\theoremstyle{remark}
\newtheorem{definition}{Definition}
\theoremstyle{definition}
\newtheorem{proposition}{Proposition}

\newtheorem{example}{Example}

\numberwithin{table}{section}
\numberwithin{figure}{section}

\graphicspath{{figures/}}  

\journal{}
\begin{document}
	
\begin{frontmatter}
\title{Global Balance and Systemic Risk in Financial Correlation Networks}
		
\author[Unimi]{Paolo Bartesaghi}
\author[Ifisc]{Fernando Diaz-Diaz}
\author[Unimib]{Rosanna Grassi} 
\author[Unimib]{Pierpaolo Uberti}

\address[Unimi]{University of Milano, Via Conservatorio 7, 20122 Milano, Italy}
\address[Ifisc]{Institute of Cross-Disciplinary Physics and Complex Systems, IFISC (UIB-CSIC), 07122 Palma de Mallorca, Spain}
\address[Unimib]{University of Milano - Bicocca, Via Bicocca degli Arcimboldi 8, 20126 Milano, Italy}
		
\begin{abstract}
\added{The global balance index is used in the network literature to quantify how balanced a signed network is. In this paper we show that the global balance index of financial correlation networks can be used as a systemic risk measure. We define the global balance index of a network starting from a diffusive process that describes how the information spreads across nodes in a network, providing an alternative derivation to the usual combinatorial one.} The steady state of this process is the solution of a linear system governed by the exponential of the replication matrix of the process. We provide a bridge between the numerical stability of this linear system, measured by the condition number in an opportune norm, and the structural predictability of the underlying signed network. The link between the condition number and related systemic risk measures, such as the market rank indicators, allows the global balance index to be interpreted as a new systemic risk measure. A comprehensive empirical application to real financial data finally confirms that the global balance index of financial correlation networks represents a valuable and effective systemic risk indicator.
\end{abstract}
		
		\begin{keyword}
			Networks, Signed networks, Global balance, Systemic risk measures 
		\end{keyword}
	\end{frontmatter}

\section{Introduction}

The financial market is a striking example of a complex system that exhibits rich cooperative dynamics and collective behaviors. The correlation structure between the returns of different assets can change dramatically during major events such as crashes and bubbles. For example, during a crash, all stocks behave similarly, implying that the entire market acts as a single, highly synchronized community. On the contrary, during a bubble, a particular sector may over-perform, widening the gap between sectors or communities. 

Starting from the 2008 global financial crisis, frequent and severe crises affected the economic system. This fact renewed the interest of practitioners and \added{academics} for systemic risk measures. A unique definition of systemic risk is missing in the literature, probably for its complexity and for the multitude of potential triggering events: financial crises, sovereign debt crises, pandemics, wars and so on. Therefore, systemic risk has been studied from many different and alternative \added{points of view}: for example, \cite{31.giglio} proposes a macroeconomic approach; \cite{26.diebold2014network} and \cite{23.demirer} applies a network connectedness analysis; \cite{adrian2011covar}, \cite{36.laeven2016bank}, \cite{9.Bernal} and \cite{37.lopez2012CoVaR}, among others, focus on conditional value at risk as a measure of systemic risk; and \cite{aldasoro2018multiplex} investigates the inter-bank financial network as a source of systemic risk. 
An extensive review of the literature on the topic is \added{beyond} the scope of the present research and for more details we refer, for instance, to \cite{14.Bisias}, \cite{45.rodriguez2013systemic} and \cite{47.silva2017analysis}.

In \cite{Billio2012}, the authors introduce two econometric measures of systemic risk that capture the interconnectedness among the monthly returns of hedge funds, banks, brokers, and insurance companies. They implement principal components analysis and Granger causality tests for the purpose. Similar methodologies have been explored by \cite{53.zheng2012changes} and \cite{52.zhang2020global}. Our contribution is framed in this context since we propose to interpret as a systemic risk measure an index computed from the matrix of returns correlations and based on the spectral properties of the corresponding network. The fact that our index depends on the eigenvalues of the correlation matrix overcomes the principal issue of the correlation, i.e. its intrinsic pairwise structure. This issue is underlined in \cite{26.diebold2014network}, “Correlation-based measures remain widespread, yet they measure only pairwise association and are largely wed to linear, Gaussian thinking, making them of limited value in financial-market contexts.”


\added{In the literature on financial networks, correlation coefficients are often mapped onto the Euclidean space by introducing suitable distances. The most widely used methodology is based on the nonlinear transformation $d(i,j)=\sqrt{2(1-c_{ij})}$, introduced in \cite{Mantegna1999}, where $c_{ij}$ is the Pearson correlation between $i$ and $j$. Other transformations have been proposed to overcome some drawbacks of the previous one, such as the nonmetric distance $d(i,j)=-\log c_{ij}^{2}$ in \cite{Hasse2020}. While they both allow the construction of weighted undirected networks and the application of classical network theory tools, the former forbids indirect links to be shorter than direct links, and the latter fails to exploit the information contained in the correlation sign.}

From a network theory perspective, if we do not employ the transformations mentioned above, the assets in a market (or in a portfolio) can be modeled as \textit{nodes}, while pairs of assets are connected via \textit{edges} with an associated weight $w_{ij}\in [-1,1]$, which corresponds to the Pearson correlation between them. The result is a signed, weighted, undirected network with self-loops. Importantly, the presence of negative weights creates novel phenomena that are absent in unsigned networks, like structural balance \cite{Harary1953, Cartwright1956}. A network is said to be structurally balanced if the nodes can be split into two subsets, such that connections between nodes of the same subset are positive and connections between nodes of different subsets are negative \cite{Harary1953}. While perfect balance is a mathematical idealization, most empirical networks display statistical properties reminiscent of balanced networks \cite{Kirkley2019,Facchetti2011}, like a high proportion of balanced triads \cite{Gallo2023}. Because of this, several balance indices have been proposed to measure how close a given network is to a structurally balanced one \cite{Kunegis2010, Estrada2014, Facchetti2011, Aref2018, Kirkley2019, Talaga2023}. The network balance significantly influences its dynamics across various models, including linear consensus \cite{Altafini2013}, the voter model \cite{Li2015}, synchronization \cite{Schaub2016}, and contagion models \cite{Fan2012,Lee2023,DiazDiaz2025}. It also plays a key role in the emergence of polarization in social systems \cite{Neal2020,Talaga2023,Fraxanet2023,Estrada2025}, compatibility of social relationships \cite{Ruiz-Garcia2023} and conflict in history \cite{bartesaghi2024d}.

The theory of signed networks \cite{Zaslavsky1982} and their balance can provide interesting insight into the subtle relationship between predictability and overall systemic risk associated with groups of stocks or the entire network. A key insight, due to Harary, is that the level of structural balance impacts also on the performance of an investment portfolio.
In fact, to the best of our knowledge, the first paper that establishes a connection between signed networks and risk in financial networks is the one by Harary \textit{et al.} \cite{Harary2002}. In this paper, the authors address the question of how to design an effective hedging strategy for portfolio selection. They analyze the network of correlations between securities and conclude that a hedging strategy is effective, and thus the portfolio less risky, if there is at least one negative link and if the signed network is 
structurally balanced. However, the authors do not deal with systemic risk measures in the proper sense and do not study the relationship between these measures and structural balance indices. Indeed, the global balance index has never been interpreted as an indicator of systemic risk.

In this research, we show how to move from these structural indicators for signed networks to systemic risk measures as usually interpreted in the literature.
The logical path of the paper can be described as follows. We start from a non-conservative diffusive process on networks governed by the adjacency matrix, interpreted as a replication matrix of the information transmitted by the nodes. We interpret the local version of the structural balance as a measure of the amount of information that leaves one node and returns to the same node when the diffusion process has reached the steady state. From the local measure we move to the global balance index, which is known and widely used in the literature to quantify structural balance in signed networks. Then, we show that this global indicator is naturally related to a ratio of the condition numbers of the two linear systems defining the steady state of the diffusive process on the signed network and on the corresponding unsigned network, that is the network with all positive edges. Thus, a natural bridge is created between the structural balance in correlation networks and a measure of the sensitivity of the network to an external perturbation of the state of the nodes. The link to known algebraic measures of conditioning provides the further connection to systemic risk measures. The condition number of the correlation matrix and some suitable generalizations called market rank indicators (MRI) have already been proposed in the literature as systemic risk measures (see \cite{Uberti2020a}). We then compare the global balance of the signed correlation network with these systemic risk measures, which are part of the broader class of so-called proper measures of connectedness. This comparison leads us to explore the ability of the global balance to discriminate systemic events. Finally, we look at how such a global indicator performs as an effective measure of systemic risk, testing its performance on real financial data.

The paper is structured as follows. In Section \ref{Sect. Preliminaries} we recall the mathematical definitions. In Section \ref{Sect. Balance}
we show how the structural balance can be obtained  from a discrete diffusive process on the network. In Section \ref{Set. Risk measures} we outline a conceptual path from the balance indices to systemic risk measures. In Section \ref{Empirical analysis} we perform an empirical analysis to prove the effectiveness of the balance index as systemic risk measure. In Section \ref{Sect. Conclusion} we draw conclusions and outlines for future research.

\section{Preliminaries}
\label{Sect. Preliminaries}
\subsection{Signed networks and structural balance}
\added{A weighted undirected network is represented by a graph $G = (V,E,\mathbf{W})$, where $V$ is the set of $N$ nodes, and $E$ is the set of undirected edges. The weighted adjacency matrix $\mathbf{W}=[w_{ij}]$ of order $N \times N$, with $w_{ij}\in \mathbb{R}^{+}$, encodes all the adjacency relationships between nodes. Specifically, $w_{ij}>0$ if there is an edge between nodes $i$ and $j$, $w_{ij}=0$ otherwise. A weighted undirected network is said to be \textit{signed} if the edge weights can be positive or negative, that is $w_{ij}\in \mathbb{R}$. The edge signs depend on the nature of the relationship between the pair of nodes.}
It may relate to activating or inhibiting functions in biological systems \cite{McDonald2008}, trusting or distrusting links in social or political networks \cite{Galam1996}, cooperative or antagonistic relationships in economic settings \cite{Pandey2021}, correlated and anticorrelated behavior in the time series of asset returns in the financial context \cite{Estrada2021}. 

A binary adjacency matrix $\mathbf{A}=[a_{ij}]$ can be associated with the signed network $G$ by setting all non-null weights equal to $1$ or $-1$ according to the sign of the entries $w_{ij}$.
In the paper we will also need to refer to the underlying unsigned network, that is the graph obtained from $G$ by neglecting the edges sign. The matrices associated with this network are $|{\bf W}|$ and $|{\bf A}|$. The strength of a node $i$, in the signed network, is then defined as $s_{i}=\sum_{j}|w_{ij}|$ and the degree as $d_{i}=\sum_{j} |a_{ij}|$.

A signed graph $G$ is \textit{structurally balanced} if there are no negative cycles, that is, cycles with an odd number of negative edges (see \cite{Harary1953,Cartwright1956}). A balanced network is characterized by the balance theorem as a graph that admits a bipartition of the vertex set $V$ such that every edge between the two subsets is negative, while every edge within each subset is positive (see \cite{Harary1953}). 

A signed graph $G$ is called \textit{antibalanced} if $-G$ is balanced, that is, if the graph with opposite edge signs with respect to $G$ is balanced (see \cite{Harary1957}). An antibalanced graph does not contain cycles with an odd number of positive edges, that is all even cycles are positive and all odd cycles are negative. Equivalently, there is a bipartition of $V$ into two subsets, such that an edge is positive if and only if it has its endpoints in different subsets (see \cite{Harary1957,Zaslavsky2013}).

Finally, a signed graph $G$ is \textit{strictly unbalanced} if $G$ is neither balanced nor antibalanced. Structural balance can also be framed in terms of closed walks instead of cycles. In a network, a walk is a sequence of vertices where each adjacent pair is connected by an edge, and a closed walk is a walk that starts and ends at the same vertex.  The sign of a walk is the product of the signs of its edges. Unlike cycles, closed walks can revisit vertices and edges, making them simpler to compute. Indeed, while enumerating cycles is an NP-complete problem, the enumeration of walks can be solved in a polynomial time by computing the powers of the adjacency matrix. Thus, by utilizing closed walks, the analysis of structural balance becomes more computationally efficient. The use of walks also provides connections to the theory of random walks and Markov chains, as we prove in the next sections.

\subsection{Correlation networks from asset returns}

Let us consider the time series of the log-returns of $N$ assets represented by the random variables $X_{it}=\log \frac{P_{i,t+1}}{P_{i,t}}$,
where $P_{i,t}$ is the price at time $t$, and the corresponding standardized variables $\tilde{X}_{it}$, with $i=1,...,N$ and $t=1,...,T$.  Given the matrix of the standardized returns $\mathbf{\tilde{X}}\in {\mathbb R}^{N\times T}$, the correlation matrix is
$\mathbf{C}=\frac{1}{T} \mathbf{\tilde{X} }\mathbf{\tilde{X}}^T.$
We denote by $(\lambda_{i},{\bf {\phi}}_{i})$, $i=1 \dots N$, the eigenpairs of the matrix $\mathbf{C}$. Recall that $\sqrt{T\cdot \lambda_{i}}$ are the singular values of the matrix $ \mathbf{\tilde{X}}$, and
the number of nonzero singular values of $\mathbf{\tilde{X}}$ defines its rank.
Since correlations are inherently signed, neglecting the diagonal elements -- that are all equal to $1$ -- the matrix $\bf C$ can be viewed as the weighted adjacency matrix of a signed network $G_C = (V,E,\mathbf{\mathbf{C}})$, where the nodes are the assets and the signed weighted edges are the correlations between them. We will call $G_C$ the corresponding correlation network.

\section{From dynamics to balance indices}
\label{Sect. Balance}

Our aim is to quantify the level of reliability and structural predictability of the network and derive the local and global balance indices in an original way. We start from a discrete non-Markovian non-conservative diffusive process on the network. This process is then applied to the correlation network where, thanks to the features of the correlation matrix, it is possible to prove some important properties on the stability of the related dynamical problem.

\subsection{A non-conservative information diffusion process on signed networks}
\label{Model}

In \cite{Lambiotte2024} the authors propose a dynamical process based on the stochastic transition matrix for two walkers of opposite sign and nature. Conversely, we propose  a non-Markovian non-conservative diffusion process based on the adjacency matrix ${\bf A}$. We note that signed networks are useful for the study of non-conservative processes, as the presence of negative connections typically causes standard dynamics like linear consensus to become non-conservative (see \cite{Altafini2013}).

We assume that the state of a node in the network is described by a variable that can take positive or negative values. A high absolute value of this variable reflects a high information content. A positive (negative) value on a node is interpreted as a favorable (unfavorable) information content for that node. An edge transfers information between the two incident nodes. Negative edges flip the sign of walkers crossing them, while their sign remains unchanged when crossing positive edges. In other words, a positive walker becomes negative after crossing a negative edge, and is received by the destination node as the opposite of that in the source node, while it preserves its sign when crossing a positive edge. For instance, in a correlation network, if a piece of information is favorable to a given security, it remains favorable to all the positive correlated securities and it becomes unfavorable to the securities negatively correlated with the former. On the opposite, if a piece of information is unfavorable to a given security, it remains unfavorable to all the positively correlated securities, while it becomes favorable to the securities negatively correlated. An overall neutral information content on a node is represented by a state equal to zero. In general, although we refer to a discrete time process, the state variable on each node is assumed to be continuous.

Given these premises, two possible non-conservative diffusion processes can be designed. In the first process, a node replicates and transfers its entire information content to all its neighbors. In the second one, the node does not transfer its entire information content, but only the information received in the previous step of the process.

Let $x_{i}(t)$ be the information content of node $i$ at a discrete time $t\geq 0$.
In the first case, when the node transmits all its content to its neighbors, the process is described by the recursive relation $x_{i}(t)=x_{i}(t-1)+\sum_{j=1}^{N}A_{ij}x_{j}(t-1)$, and the state vector ${\bf x}(t)$ evolves according to ${\bf x}(t)={\bf x}(t-1)+{\bf A} {\bf x}(t-1)$. The adjacency matrix $\bf A$, called  \textit{replication} matrix in the framework of diffusive processes, captures how the original information is copied and transferred from node $i$ to node $j$. In a discrete setting, this process produces a new configuration of the whole network at the time step $t$, given by ${\bf x}(t)=({\bf I}+{\bf A})^{t}{\bf x}_{0}$, where ${\bf x}_{0}$ is the initial state vector.

In the process described above, a node retains its own information and diffuses all its content to its neighbors. 
More interestingly, let us suppose that, at time $t$, 
a node retains all what it has received up to time $t-1$
from its neighbors and it does not transmit its entire information content but only the portion \textit{just received} from its neighbors in the previous step $t-1$. In words, the node does not share all of its information content but only the part that \textit{it has just learned} from its neighbors. This is because what it knew before has already been sent in the previous steps of the diffusion process. Moreover, to ensure a greater generality, we assume that each node transmits only a fraction $\alpha(t)$ ($0<\alpha(t)\leq1$) of \added{the received information}. In summary, the dynamics is expressed by the following iterative relation, for $t\geq 0$,
\begin{equation}
	x_{i}(t+1)=x_{i}(t)+\alpha(t)\sum_{j=1}^{N}a_{ij}\left[x_{j}(t)-x_{j}(t-1) \right]
	\label{diffusion1}
\end{equation}
where $x_{i}(-1)\equiv 0$ and $\alpha(t)$ is the penalization factor. Setting $\Delta {\bf x}(t)={\bf x}(t)-{\bf x}(t-1)$, and rewriting in matrix form, it equals
\begin{equation}
	\Delta 	{\bf x}(t+1)={\alpha}(t){\bf A} \Delta {\bf x}(t)
	\label{diffusion2}
\end{equation}
where $\Delta {\bf x}(0)={\bf x}(0)= {\bf x}_{0}$. Simple computations show that Eq. \eqref{diffusion2} is equivalent to ${\bf x}(t)={\bf x}(t-1)+\left( \prod_{k=0}^{t-1}\alpha(k) \right) {\bf A}^{t} {\bf x}_{0}$. This allows to express the state vector ${\bf x}(t)$ at time $t$ in terms of the initial state vector ${\bf x}_{0}$ as
\begin{equation}
	{\bf x}(t)={\bf G}(t) {\bf x}_{0}
\end{equation}
where the matrix function ${\bf G}(t)$ is
\begin{equation*}
	{\bf G}(t)={\bf I}+\sum_{\tau=1}^{t}\left( \prod_{k=0}^{\tau-1}\alpha(k) \right) {\bf A}^{\tau}.
\end{equation*}
In particular, by setting $\alpha(k)=\frac{1}{k+1}$, for $k\geq 0$, the matrix function becomes ${\bf G}(t)=\sum_{\tau=0}^{t} \frac{1}{\tau!} {\bf A}^{\tau}$, so that
\begin{equation}
	{\bf x}(t)=\left[ \sum_{\tau=0}^{t} \frac{1}{\tau!} {\bf A}^{\tau}\right] {\bf x}_{0}.
	\label{state}
\end{equation}
Since $\lim_{t \to \infty}{\bf G}(t)=e^{\bf A}$, the asymptotic state is
\begin{equation}
	{\bf x}_{\infty}=e^{\bf A} {\bf x}_{0}.
	\label{asymptoticstate}
\end{equation}
\begin{remark}
The above derivation provides an alternative interpretation of the exponential operator. In the literature (see \cite{bartesaghi2024d}), the exponential operator usually measures the information diffusion in a Markovian context at $t=1$. We have shown that the same operator can also be interpreted as the equilibrium distribution of a non-Markovian process. This overcomes the issue of arbitrary tuning the time parameter at $t=1$. \added{In the Supplementary Material, Sec. 1, we show another possible choice for the coefficient $\alpha(k)$ and the resulting matrix function. However, for the remainder of the paper, we will always refer to the choice in Eq. \ref{state} for two main reasons: first, it allows for analytical derivations; second, the resulting operator, the matrix exponential, has been thoroughly studied in the mathematical literature, and many optimized algorithms exist to compute it efficiently.}
\end{remark}

\subsection{Balance indices}
Let us consider the process in Eq. \eqref{state}. We assume that the diffusion starts from a single node in the network. Let us set ${\bf x}_{0}={\bf e}_{i}$, where ${\bf e}_{i}$ is the $i-th$ vector of the standard basis in $\mathbb{R}^{N}$. 
Then the asymptotic state is 
\begin{equation}
	{\bf x}_{\infty}=[e^{\bf A}]_{i},
\end{equation}
where $[e^{\bf A}]_{i}$ denotes the $i-$th column of the exponential matrix. The amount of information originating from node $i$ that returns to the same node after the diffusion process reaches a steady state is then given by the diagonal element
\begin{equation}
	{x}_{\infty,i}=[e^{\bf A}]_{ii}.
\end{equation}
This description does not depend on the sign of the weights and can be repeated for both the matrices $\mathbf{A}$ and $|{\mathbf A}|$.
Due to the presence of signed edges, the amount of information lost in a non-conservative diffusion process through the network  can be measured by
\begin{equation}
	\kappa_{i}(G):=\frac{[e^{\bf A}]_{ii}}{[e^{|{\bf A}|}]_{ii}}.
	\label{Localbalance}
\end{equation}
This quantity, called \textit{local balance} of the node $i$, has been introduced in \cite{bartesaghi2024d} by both a combinatoric perspective and a diffusive approach based on the Lerman-Gosh Laplacian \cite{Lerman2024}. In this context, it compares the fraction of information originating from $i$ that comes back to $i$ in the actual network and in the balanced versions of the same network. The numerator in Eq. \eqref{Localbalance} represents the \textit{polarization} (see \cite{Lambiotte2024}) at each node, defined as the difference between the number of positive and negative walkers; the denominator is the total number of walkers coming back to that node.\footnote{Similarly, the term $(e^A)_{ij}$ captures the net influence that an initial shock in node $i$ exerts on node $j$. This object, known as \textit{signed communicability}, has recently been proposed in \cite{Estrada2025} and could be used to define an edge balance index: $\kappa_{ij} := \frac{(e^A)_{ij}}{(e^{|A|})_{ij}}$. Such an index captures the effect of the signed edges and the unbalance on the transmission of information between nodes.}

\begin{definition}
	The \textit{global balance} index of the signed network $G$ is defined as:
	\begin{equation}
		\kappa(G)
		=\frac{\sum_{i=1}^{N}[e^{\bf A}]_{ii}}{\sum_{i=1}^{N}[e^{|{\bf A}|}]_{ii}}
		=\frac{{\rm tr}[e^{\bf A}]}{{\rm tr}[e^{|{\bf A}|}]}
		=\frac{\sum_{i=1}^{N}e^{\lambda_{i}}}{\sum_{i=1}^{N}e^{\overline{\lambda}_{i}}}
		\label{Globalbalance}
	\end{equation}
	where $\lambda_{i}$ and $\overline{\lambda}_{i}$ are respectively the eigenvalues of ${\bf A}$ and $|{\bf A}|$.
\end{definition}

Let us recall the Acharya's theorem (see \cite{Acharya1980}), which proves that ${\bf A}$ and $|{\bf A}|$ have the same spectrum if and only if the network is balanced. Hence, a network is balanced if and only if $\kappa_{i}(G)=1, \forall i,$ or, equivalently, $\kappa(G)=1$ (see \cite{bartesaghi2024d}).

Finally, we observe that the above result still holds if we replace the binary matrix $\bf A$ with the weighted matrix $\bf W$. Moreover, the presence of identical loops in the network does not affect the computation of the balance indices. In other words, if the diagonal elements of the matrix are all equal, they do not contribute to the global balance.

More specifically, let $\bf A$ be a square matrix with real, possibly negative, entries and diagonal elements equal to $\chi \in {\mathbb R}^{+}$. Let $\tilde{\bf A}={\bf A}-\chi {\bf I}$, where ${\bf I}$ is the identity matrix, and $G$ and $\tilde G$ the corresponding graphs. Then $\kappa_{v}(G)=\kappa_{v}({\tilde G})$. In fact,

\begin{equation*}
		\kappa_{v}(G):=\frac{[e^{\bf A}]_{vv}}{[e^{|{\bf A}|}]_{vv}}
		=\frac{[e^{\tilde{\bf A}+\chi{\bf I}}]_{vv}}{[e^{|{\tilde{\bf A}}|+\chi{\bf I}}]_{vv}}
		=\frac{[e^{\tilde{\bf A}}\cdot e^{\chi{\bf I}}]_{vv}}{[e^{|{\tilde{\bf A}}|}\cdot e^{\chi{\bf I}}]_{vv}}
		=\frac{[e^{\tilde{\bf A}}]_{vv}}{[e^{|{\tilde{\bf A}}|}]_{vv}}=\kappa_{v}({\tilde G})
\end{equation*}

This result allows us to apply the described process to a correlation network $G_c$, using the correlation matrix $\bf {C}$, with diagonal elements $\chi=1$, as the replication matrix. 

\section{From balance indices to systemic risk measures}
\label{Set. Risk measures}
We now present the conceptual path from the illustrated balance indices to systemic risk measures on correlation networks. The first step is an assessment on the structural predictability of the diffusion model as a function of the signed network on which it operates. The second step is an analysis of the conditioning of the linear problem associated with the asymptotic state. This will naturally unveil a connection between the balance indices and a class of systemic risk measures, the market rank indicators, which fall within the broader set of so-called proper measures of connectedness.

\subsection{Interpretation of the balance indices in terms of network structural reliability/predictability}
\label{predictability}

We aim to highlight a crucial aspect of network balance, which is related to the degree of structural predictability of the dynamic processes occurring in the network and to the reliability of the latter in transmitting information. It is widely accepted that structurally balanced networks make the behavior of the dynamics largely predictable (see \cite{Lambiotte2024}). In particular, we discuss how the notion of \textit{structural predictability} should be interpreted. 
We begin by considering the toy examples of three and four nodes, whose modeling signed graphs can be exhaustively listed (see Fig. \ref{fig1}).

\begin{figure}[H]
	\centering
	\captionsetup[subfloat]{labelformat=empty}
	\subfloat[$K_3$]{\includegraphics[width=0.50\textwidth]{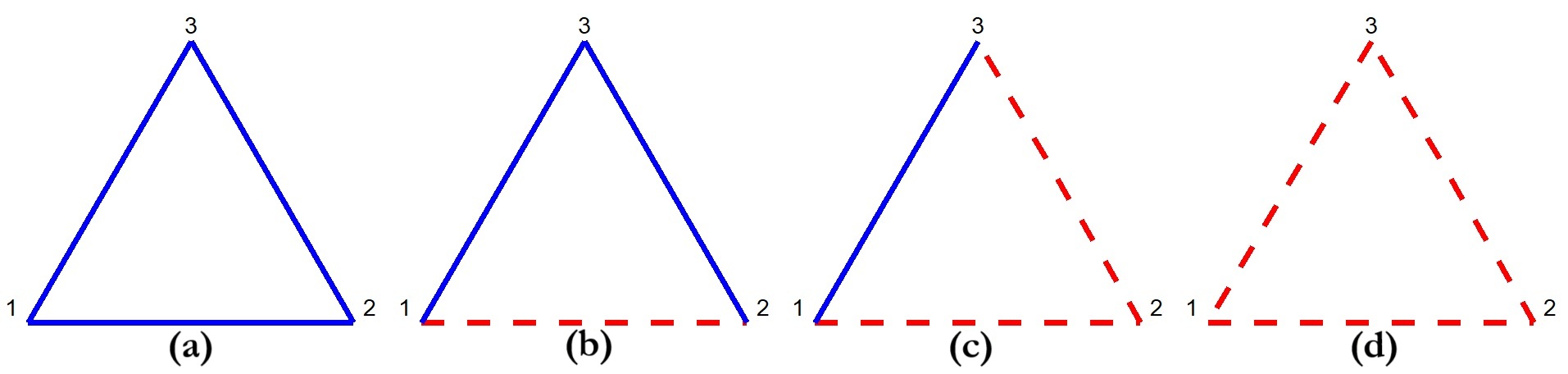}}\\
	\subfloat[$K_4$ balanced]{\includegraphics[width=0.50\textwidth]{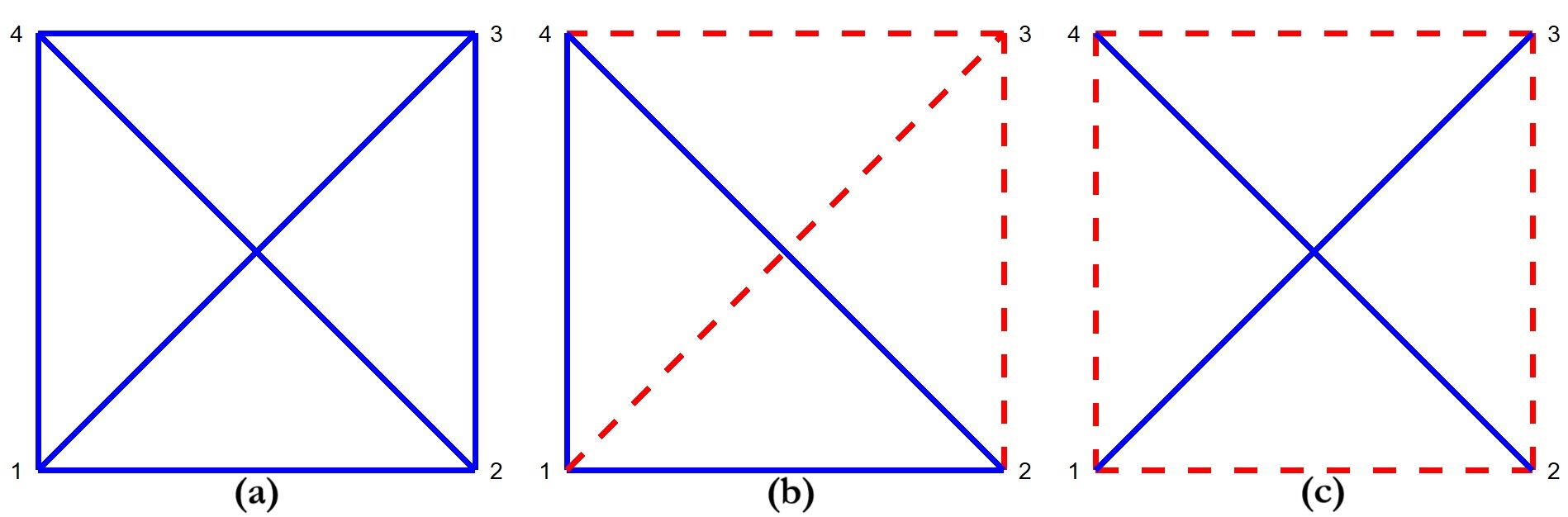}}\\
	\subfloat[$K_4$ unbalanced]{\includegraphics[width=0.50\textwidth]{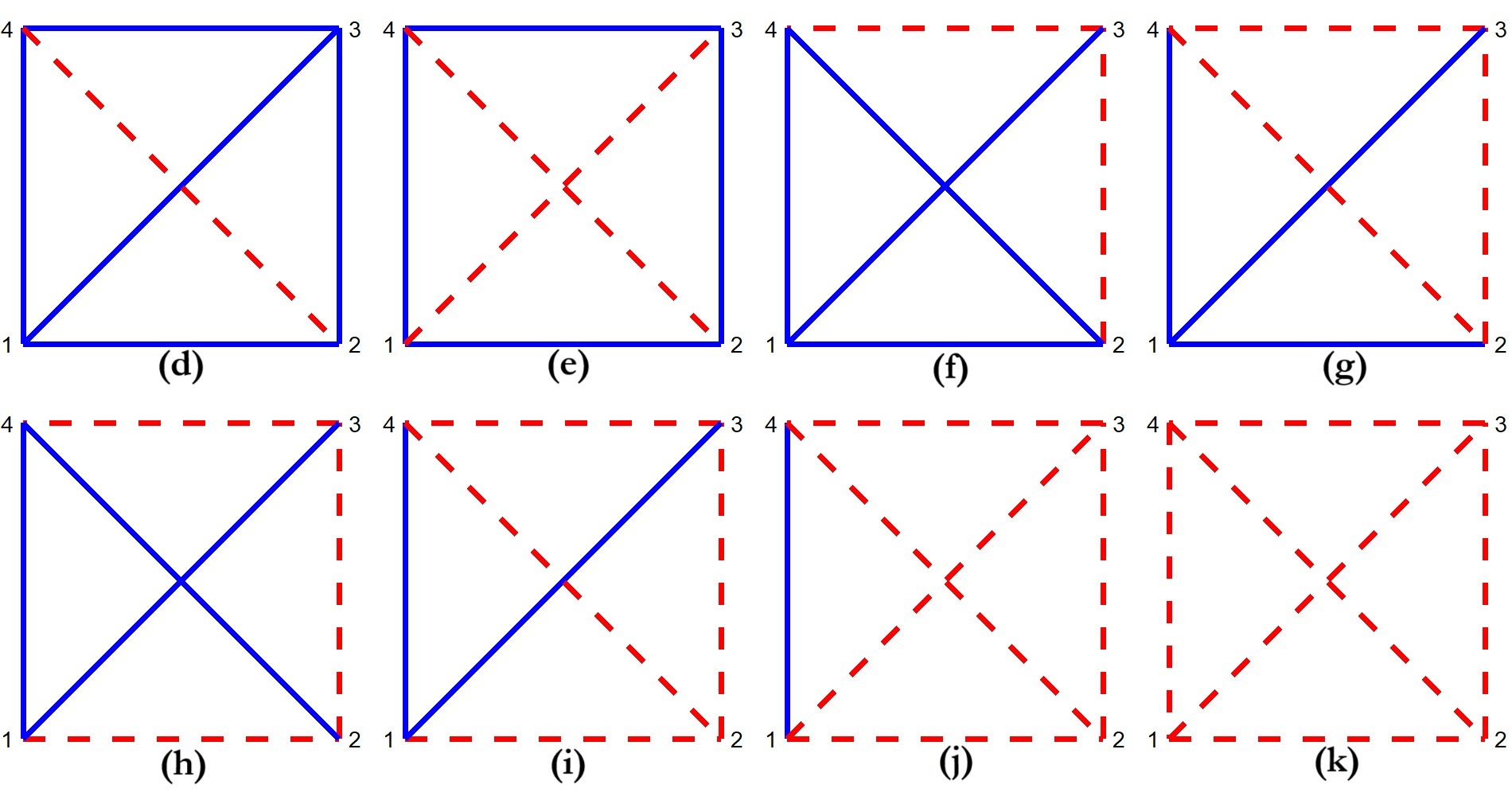}}
	\caption{All possible $K_3$ and $K_4$ signed graphs, up to isomorphism. In the first line, balanced triangles, (a) and (c), and unbalanced triangles, (b) and (d). In the second line, the three balanced $K_4$ graphs. In the third and fourth line, the eight unbalanced $K_4$ graphs, among which (e), (g) and (k) are antibalanced and the others are strictly unbalanced.}
	\label{fig1}
\end{figure}

With three nodes, four complete signed graphs $K_{3}$ are possible. The monotone answer of one node can be uniquely decided from that of the other two in the balanced triangles (a) and (c) in Fig. \ref{fig1}, first line, while this is not possible in the unbalanced triangles (b) and (d). In fact, in the last two cases, starting from the behavior of a node, however the triangle is traversed, it leads to the opposite behavior at the same node, showing the inherent instability and unpredictability of the structure.

To reinforce this idea, let us consider the case of a complete signed graph of four nodes. Three of them are balanced, namely (a), (b) and (c), see Fig. \ref{fig1}, second line. The other eight are unbalanced, see the third and the fourth line in Fig. \ref{fig1}; more precisely, three are antibalanced, (e), (g) and (k), and the remaining are strictly unbalanced.

For each graph, we apply the non-conservative model in Eq. \eqref{state}. We refer, for instance, to the four nodes graphs (b), (f), (k)  and we compare their behavior with the fully positive graph (a).

\begin{example}
	Suppose each node is in an initial state represented by the value of its attribute or endowment at time $t=0$. Graph $K_4$ (b) is balanced, so that $\kappa(G)=1$. If node $1$ increases its attribute, nodes $2$ and $4$ increase their own attributes as well, and the attribute of node $3$ can only decrease. For instance, starting from an initial uniform state ${\bf x}_{0}=[1,1,1,1]^{T}$, the asymptotic state is ${\bf x}_{\infty}=[10.227,10.227, -9.491,10.227]^{T}$. If we increase by $1$ unit the state of node $1$, that is ${\bf x}_{0}=[2,1,1,1]^{T}$, we get ${\bf x}_{\infty}=[15.524,15.156,-14.420,15.156]^{T}$. Note that, in this case, the local balance of each node is $\kappa_{i}(G)=1, \ \forall i=1,2,3,4$. 
\end{example}

\begin{example}
	We now focus on graph $K_4$, (f) that is strictly unbalanced with $\kappa(G)=0.592$. If node $1$ increases its attribute, nodes $2$ and $4$ should increase their attribute. Now, however, node $3$ should both increase its state due to its connection with node $1$ and decrease its state due to the effect of its connections with nodes $2$ and $4$. Suppose again that the initial state is ${\bf x}_{0}=[1,1,1,1]^{T}$, so that the asymptotic state is ${\bf x}_{\infty}=[6.855,6.800,-1.418,6.800,]^{T}$. If the attribute of node $1$ increases, for instance with ${\bf x}_{0}=[18,1,1,1]^{T}$, the final asymptotic state becomes ${\bf x}_{\infty}=[52.599,41.961,-0.952, 41.961]^{T}$, and the state of node $3$ increases by an amount of $0.466$.
	If, on the other hand, along with that of node $1$, the state of nodes $2$ and $4$ also increases slightly, for example starting from an initial state of ${\bf x}_{0}=[20,1.2,1,1.2]^{T}$, we obtain ${\bf x}_{\infty}=[58.808,47.457,-1.724,47.457]^{T}$. The state of node $3$ decreases. In this case, the local balance of the nodes are  $\kappa_{1}(G)=\kappa_{3}(G)=0.508$ and $\kappa_{2}(G)=\kappa_{4}(G)=0.677$. Thus nodes $1$ and $3$ contribute to unbalance the network more than nodes $2$ and $4$.
\end{example}

\begin{example}
	Graph $K_4$, (k) is a complete graph with all negative edges and so it is the most unbalanced in the set, with $\kappa(G)=0.387$. As expected, if node $1$ increases its attribute, all the other nodes have the opposite behavior. For instance, starting from ${\bf x}_{0}=[1.2,1,1,1]^{T}$, the final state becomes ${\bf x}_{\infty}=[0.460,-0.084,-0.084,-0.084]^{T}$.
	However, we can make node $3$ increase its state while increasing the state value of node $1$ even more by acting on the state values of nodes $2$ and $4$. For example, by reducing them as in the initial state ${\bf x}_{0}=[1.4,0.7,1,0.7]^{T}$, we obtain ${\bf x}_{\infty}=[1.271,-0.632,0.183,-0.632]^{T}$.
\end{example}

These examples illustrate two important points. First, even within a fully deterministic diffusive model, the response of a node to state perturbations of at least two other nodes exhibits non-monotonic behaviors in an unbalanced network. Second, in a balanced network, there is a more pronounced dependence on the fluctuations in the attributes of some nodes than in an unbalanced network.

\subsection{Global balance and condition number}
\label{Global balance and condition number}
We show how the global balance of the network is related to the condition number of the matrix $e^{\bf A}$. Specifically, we prove that the global balance can be used as an indicator of the conditioning of the linear system ${\bf x}_{\infty}=e^{\bf A} {\bf x}_{0}$.

Let us start with a preliminary observation. As in the toy examples in Sec. \ref{predictability}, we perturb the initial state ${\bf x}_{0}$ at node $i$ as follows:
\begin{equation}
	{\bf x}_{0}^{\star}={\bf x}_{0}+\varepsilon {\bf e}_{i}, \ \varepsilon \in \mathbb{R}.
\end{equation}
The asymptotic state of the diffusive model becomes
\begin{equation}
	{\bf x}_{\infty}^{\star}={\bf x}_{\infty}+\varepsilon [e^{\bf A}]_{i}.
\end{equation}
In particular, the $i-$th component of ${\bf x}_{\infty}^{\star}$ is:
\begin{equation}
	{x}_{\infty,i}^{\star}={x}_{\infty,i}+\varepsilon [e^{\bf A}]_{ii}.
\end{equation}
The diagonal element $[e^{\bf A}]_{ii}$ measures how a perturbation of the initial state at node $i$ is amplified in the final state of the \textit{same} node. Therefore, the local balance $\kappa_{i}(G)$ of node $i$, defined in Eq. \eqref{Localbalance}, quantifies how less effective the propagation of the perturbation starting from node $i$ is in the signed network, compared to the same perturbation from the same node propagating in the unsigned network. 

In general, in the linear problem ${\bf x}_{\infty}=e^{\bf A} {\bf x}_{0}$, let us perturb the initial state ${\bf x}_{0}$ and evaluate the effect on the asymptotic state ${\bf x}_{\infty}$. This is equivalent to considering the inverse problem $e^{-{\bf A}}{\bf x}_{\infty}={\bf x}_{0}$ and quantifying the variation in the solution ${\bf x}_{\infty}$ due to a perturbation on ${\bf x}_{0}$. Let ${\bf e}\in {\mathbb R}^{n}$ be the general error or perturbation on ${\bf x}_{0}$. Then the condition number $\mathscr{K}$ of the linear problem is the ratio between the relative errors in ${\bf x}_{\infty}$ and in ${\bf x}_{0}$ (see, for instance, \cite{Golub2013}):
\begin{equation}
	\mathscr{K}(e^{-{\bf A}})=\lim_{\varepsilon \to 0} \sup_{||{\bf e}||<\varepsilon}\frac{||e^{\bf A}{\bf e}||}{||{\bf e}||}\frac{||{\bf x}_{0}||}{||e^{\bf A}{{\bf x}_{0}}||}=||e^{-{\bf A}}|| \cdot || e^{{\bf A}}||,
\end{equation}
where  $||\cdot||$ is any matrix norm. In particular, using the \textit{trace} norm (Schatten $p-$norm with $p=1$), it can be expressed as
\begin{equation}
	\mathscr{K}(e^{-{\bf A}})=\left( \sum_{j=1}^{N}e^{-\lambda_{j}} \right)\left( \sum_{j=1}^{N}e^{\lambda_{j}} \right).
	\label{conditionnumber}
\end{equation}
We are interested in studying the ratio between the condition numbers of the exponential matrices $e^{\bf A}$ and $e^{\bf |A|}$, that is
\begin{equation}
	{\mathscr R}({\bf A}) = \frac{{\mathscr{K}(e^{-{\bf A}})}}{{\mathscr{K}(e^{-|{\bf A}|})}}.
	\label{theratio}
\end{equation}
In particular, we aim to discriminate when ${\mathscr R}({\bf A})$ is greater or less than $1$. In fact, if ${\mathscr R}({\bf A})<1$, the linear system in Eq. \eqref{asymptoticstate} on the signed network is better conditioned than the same system on the unsigned network, and vice versa if ${\mathscr R}({\bf A})>1$. This determines the different response of the nodes to perturbations propagating within the network in the two versions, signed and unsigned, of the network. We also show that ${\mathscr R}({\bf A})$ is tightly related to the global balance $\kappa(G)$ in Eq. \eqref{Globalbalance} and, under some further assumptions, the time series of their values in a time-evolving network are highly correlated (see Sec. \ref{Empirical analysis}). We start observing that the ratio ${\mathscr R}({\bf A})$ between the condition numbers of the matrices $e^{-{\bf A}}$ and $e^{-|{\bf A}|}$, in the Schatten $p-$norm with $p=1$, can be expressed as a function of the global balances of the networks $G$, $-G$ and $-|G|$ as
	\begin{equation}
		{\mathscr R}({\bf A})= \frac{\kappa(G)\cdot \kappa(-G)}{
			\kappa(-|G|)}
		\label{ratioandbalances}
	\end{equation}
	where $-G$ is obtained by flipping the signs of all the edges in $G$, and $-|G|$ is the fully negative network.
In fact, from Eqs. \eqref{Globalbalance} and \eqref{conditionnumber}, we have
	\begin{equation*}
		{\mathscr R}({\bf A})=
		\frac{{\rm tr}[e^{{\bf A}}]\cdot {\rm tr}[e^{-{\bf A}}]}
		{{\rm tr}[e^{|{\bf A}|}]\cdot {\rm tr}[e^{-|{\bf A}|}]}=
		\frac{\frac{{\rm tr}[e^{{\bf A}}]}{{\rm tr}[e^{|{\bf A}|}]} \cdot \frac{{\rm tr}[e^{-{\bf A}}]}{{\rm tr}[e^{|{\bf A}|}]}}{\frac{{\rm tr}[e^{-|{\bf A}|}]}{{\rm tr}[e^{|{\bf A}|}]}}=\frac{\kappa(G)\cdot \kappa(-G)}{\kappa(-|G|)}.
	\end{equation*}

Formula \ref{ratioandbalances} shows that the ratio of the two condition numbers is closely related to the global balance of the network $G$, also through the global balances of the networks associated with it, namely $-G$ and $-|G|$.

We now prove a necessary condition for ${\mathscr R}({\bf A})=1$. First, observe that if $G$ is balanced, that is $\kappa(G)=1$, both $-G$ and  $-|G|$ are antibalanced. Actually, we do not know whether they are balanced or not. We can only say that $\kappa(-G)\leq1$ and $\kappa(-|G|)\leq1$. Nevertheless, we can state \added{the following proposition}:
\begin{proposition}
	\label{proposition1}
If $G$ is a balanced or antibalanced signed network, then ${\mathscr R}({\bf A})=1$.
\end{proposition}
 The proof of this proposition can be found in the Supplementary Material, Sec. 2. We conjecture that this necessary condition can also be extended to a sufficient condition, that is: If ${\mathscr R}({\bf A})=1$ then $G$ is a balanced or antibalanced signed network.\footnote{We have numerically checked that all complete signed networks that can be generated up to a certain order $N$ satisfy this condition as both a necessary and a sufficient condition.}

Observe that, \added{in the proof of} the previous proposition, the ratio ${\mathscr R}({\bf A})$ has been expressed as a function of the properties of the network $G$, namely through the penalized number or weight of even/odd and positive/negative closed walks in $G$.
\added{Specifically, we prove that}
\begin{equation}
{\mathscr R}({\bf A})=\frac{
	\left[
	W^{+}_{\rm even}-W^{-}_{\rm even}
	\right]^{2}
	-
	\left[
	W^{+}_{\rm odd}-W^{-}_{\rm odd}
	\right]^{2}
}
{
	\left[
	W^{+}_{\rm even}+W^{-}_{\rm even}
	\right]^{2}
	-
	\left[
	W^{+}_{\rm odd}+W^{-}_{\rm odd}
	\right]^{2}
}
\label{ratioandwalks}
\end{equation}
\added{where $W^{\pm}_{\rm even}$ and $W^{\pm}_{\rm odd}$ denote the total penalized weight of positive and negative closed walks of even and odd length, respectively.}

In the next Remark we highlight the explicit expression of these four types of closed walks.

\begin{remark}
	The penalized weight of even/odd and positive/negative closed walks in a signed network $G$ can be expressed as
	\begin{equation}
		\begin{split}
			W^{\pm}_{\rm even}&={\rm tr}\left( \cosh |{\bf A}| \pm \cosh {\bf A}\right)=\sum_{i=1}^{N}\left( \cosh \overline{\lambda}_{i} \pm \cosh {\lambda}_{i}\right)\\
			W^{\pm}_{\rm odd}&={\rm tr}\left( \sinh |{\bf A}| \pm \sinh {\bf A}\right)=\sum_{i=1}^{N}\left( \sinh \overline{\lambda}_{i} \pm \sinh {\lambda}_{i}\right)
		\end{split}
		\label{numberofwalks}
	\end{equation}
\added{where ${\lambda}_{N}\leq{\lambda}_{N-1}\leq\dots \leq{\lambda}_{1}$ and $\overline{\lambda}_{N}\leq\overline{\lambda}_{N-1}\leq\dots \leq \overline{\lambda}_{1}$ are the eigenvalues of $\bf A$ and $\bf |A|$, respectively.}
\end{remark}
\begin{remark}
	By Eq. \eqref{numberofwalks}, it follows immediately that $ W^{+}_{\rm even}> W^{-}_{\rm even}$ and $ W^{+}_{\rm even}> W^{+}_{\rm odd}$ for any signed network. In words, weighted positive even closed walks always prevail on the even negative and odd positive ones. 
\end{remark}

\begin{definition}
	A signed graph such that $W^{+}_{\rm even} \cdot W^{-}_{\rm even}> W^{+}_{\rm odd}\cdot W^{-}_{\rm odd}$ is said \textit{even-dominant} and a signed graph such that $W^{+}_{\rm even} \cdot W^{-}_{\rm even}< W^{+}_{\rm odd}\cdot W^{-}_{\rm odd}$ is said \textit{odd-dominant}.
\end{definition}

We can now show that ${\mathscr R}({\bf A})$ is greater or less than $1$ according to the fact that the signed graph is even-dominant or odd-dominant. Let ${\mathscr R}({\bf A})$ be the ratio between the condition numbers of the two matrices $e^{-{\bf A}}$ and $e^{-|{\bf A}|}$, in the Schatten $p-$norm with $p=1$. ${\mathscr R}({\bf A}) <1$ 	if and only if the matrix ${\bf A}$ is the adjacency matrix of a strictly unbalanced even-dominant signed network $G$. ${\mathscr R}({\bf A}) >1$  	if and only if the matrix ${\bf A}$ is the adjacency matrix of a strictly unbalanced odd-dominant signed network $G$. In fact, if we refer for instance to an even-dominant signed network, by expanding the ratio as in Eq. \ref{ratioandwalks}, we have that ${\mathscr R}({\bf A})<1$ if and only if $W^{+}_{\rm even} \cdot W^{-}_{\rm even}- W^{+}_{\rm odd}\cdot W^{-}_{\rm odd}>0$.

Formula \ref{ratioandbalances} shows the close relation between the ratio of the condition numbers in the signed and unsigned networks and the balance of the network $G$ and its related networks $-G$, $|G|$ and $-|G|$. If the network $G$ is balanced or antibalanced, the conditioning of the linear problem in Eq. \eqref{asymptoticstate} is the same in the signed network and in its underlying unsigned network. The ratio ${\mathscr R}({\bf A})$ can be both less or greater than $1$ depending on the topological properties of the network, in particular on the dominance of even closed walks or odd closed walks. The previous considerations provide an almost complete classification of the signed graphs based on the value of the ratio in formula \eqref{ratioandwalks}. If ${\mathscr R}({\bf A})< 1$ the graph is even-dominant strictly unbalanced; if ${\mathscr R}({\bf A})>1$ the graph is odd-dominant strictly unbalanced; if the graph is balanced or antibalanced, ${\mathscr R}({\bf A})=1$.

\subsection{Approximating Global Balance and Condition Numbers  in Correlation Networks}

We show now that the ratio ${\mathscr R}({\bf A})$, for a large set of significant matrices, in particular for the empirical correlation matrices we are interested in, is less than $1$.  We assume from now on that the matrix ${\bf A}$ is positive semidefinite so that all its eigenvalues $\lambda_{i}$ are nonnegative: $0\leq{\lambda}_{N}\leq{\lambda}_{N-1}\leq\dots \leq{\lambda}_{2}\leq{\lambda}_{1}$. We refer to ${\lambda}_{1}$ as the spectral radius and to ${\lambda}_{1}-{\lambda}_{2}$ as the spectral gap.

We first express ${\mathscr R}({\bf A})$ as a function of the intervals between pairs of eigenvalues. The ratio ${\mathscr R}({\bf A})$ between the condition numbers of the matrices $e^{-{\bf A}}$ and $e^{-|{\bf A}|}$, in the Schatten $p-$norm with $p=1$, can be expressed as
	\begin{equation}
		{\mathscr R}({\bf A})= \frac{\left( \sum_{i=1}^{N}e^{-\lambda_{i}} \right)\left( \sum_{i=1}^{N}e^{\lambda_{i}} \right)}{\left( \sum_{i=1}^{N}e^{-\overline{\lambda}_{i}} \right)\left( \sum_{i=1}^{N}e^{\overline{\lambda}_{i}} \right)}=
		\frac{1+\frac{2}{N}\sum_{j<i}\cosh(\lambda_{i}-\lambda_{j})}{1+\frac{2}{N}\sum_{j<i}\cosh(\overline{\lambda}_{i}-\overline{\lambda}_{j})}
		\label{ratio}
	\end{equation}
	where again ${\lambda}_{i}$ and $\overline{\lambda}_{i}, \ i=1,\dots,N$, are the eigenvalues of $\bf A$ and $\bf |A|$, and $\cosh(\cdot)$ denotes the  hyperbolic cosine function.
	The statement follows directly from the identity $\left( \sum_{i=1}^{N}e^{-\lambda_{i}} \right)\left( \sum_{i=1}^{N}e^{\lambda_{i}} \right)=N+2\sum_{j<i}\cosh(\lambda_{i}-\lambda_{j})$.

Formula \ref{ratio} emphasizes the dependence of ${\mathscr R}({\bf A})$ on the intervals $\Delta\lambda_{ij}=\lambda_{i}-\lambda_{j}$ between pairs of eigenvalues. It highlights the relevant role of the spectral gap in large correlation networks based on the correlation matrix ${\bf C}$.

\begin{remark}
	\label{remark4}
	In the presence of a large spectral gap for both the signed and unsigned networks, that is if $\Delta \lambda_{1j}\gg \Delta \lambda_{ij}, 2<i,j\leq N$, and similarly for the $\overline{\lambda}_{i}$'s, the ratio ${\mathscr R}({\bf A})$ can be expressed in the approximated form
	\begin{equation}
		{\mathscr R}({\bf A})\approx
		\frac{1+\frac{2}{N}\sum_{j=2}^{N}\cosh \Delta \lambda_{1j}}{1+\frac{2}{N}\sum_{j=2}^{N}\cosh \Delta \overline{\lambda}_{1j}}\approx
		\frac{1+\frac{2(N-1)}{N}\cosh \lambda_{1}}{1+\frac{2(N-1)}{N}\cosh \overline{\lambda}_{1}}\approx \frac{e^{\lambda_{1}}}{e^{\overline{\lambda}_{1}}}
		\label{approximateratio}
	\end{equation}
	In general, for any signed network, the inequality $\lambda_{1}\leq \overline{\lambda}_{1}$ holds. The spectral radius of the matrix ${\bf A}$ is strictly smaller than that of the matrix $|{\bf A}|$ if and only if the signed network $G$ is strictly unbalanced, that is neither balanced nor antibalanced, i.e. ${\lambda}_{1}<\overline{\lambda}_{1}$ if and only if  $G$ is strictly unbalanced (see \cite{Lambiotte2024}, Theorem 3.9, and \cite{Belardo2019}). This implies that the approximated value in Eq. \eqref{approximateratio} is less than $1$ for strictly unbalanced networks.
	
	Let us observe that, since $\sum_{j=2}^{N}\lambda_{j}=N-\lambda_{1}$ and $\sum_{j=2}^{N}\overline{\lambda}_{j}=N-\overline{\lambda}_{1}$, we have $\sum_{j=2}^{N}\Delta\lambda_{1j}=N(\lambda_{1}-1)<N(\overline{\lambda}_{1}-1)=\sum_{j=2}^{N}\Delta\overline{\lambda}_{1j}$, but this does not guarantees that $\Delta \lambda_{1j}<\Delta \overline{\lambda}_{1j}, \ \forall j$. However, if ${\lambda}_{1}\gg {\lambda}_{j}$ and $\overline{\lambda}_{1}\gg \overline{\lambda}_{j}$, for $j=2,\dots, N$, we can also assume $\Delta {\lambda}_{1j}\approx {\lambda}_{1}$ and $\Delta \overline{\lambda}_{1j}\approx \overline{\lambda}_{1}$. This justifies the second approximation in Eq. \eqref{approximateratio}.
	As a further consequence of the same equation, we can conclude that for correlation networks with large spectral gap the value of the ratio ${\mathscr R}({\bf A})$ can be approximated with the value of the global balance $\kappa(G)$, that is ${\mathscr R}({\bf A})\approx \kappa(G)$. This effectiveness of the approximation is confirmed by the results in the empirical analysis.
\end{remark}

Remark \ref{remark4} shows that, in general, when dealing with correlation networks, the conditioning of the linear system in Eq. \eqref{asymptoticstate} is better on the signed correlation network than on the corresponding unsigned correlation network. The correlation matrices, when interpreted as adjacency matrices of a signed graph, are usually even-dominant.

Our interpretation of this result is that the presence of negative edges, that is negative correlations, mitigates the propagation of a perturbation shock within the network. Therefore, an unbalanced graph can mitigate and absorb the perturbations from one node better than a balanced graph. In a strictly unbalanced graph, we expect that the oscillations produced by a single node cannot be amplified to a large extent due to the lower condition number of the adjacency matrix. Conversely, in a balanced graph, where there are two disjoint communities and all closed walks are positive, a perturbation originating from one node is amplified traveling through the network. This result shows that a balanced graph is more unstable and sensitive to perturbations than an unbalanced one, and
supports the interpretation of network global balance in terms of predictability in transmitting information, as discussed in subsection \ref{predictability}. We conclude this section with a numerical example.

\begin{example}
	Consider the following randomly generated correlation matrix with $N=4$:
	\begin{equation}
		A=
		\left[ \begin{array}{cccc} 
			1 & 0.764 & 0.335 & 0.236 \\
			0.764 &  1 & 0.051 & -0.349  \\
			0.335 & 0.051 &  1 & 0.735  \\
			0.236 & -0.349 & 0.735  &  1  \\
		\end{array} \right].
	\end{equation}
	The network associated with this correlation matrix is a weighted version of the graph $K_4$ (d) in Fig. \ref{fig1}, third line. It is a strictly unbalanced network and its global balance is $\kappa(G)=0.962$. The local balances of the four nodes are collected in the vector $[0.960, 0.952, 0.983, 0.953]$.
	
	Let us first note that by increasing the initial state ${\bf x}_{0}=[1,1,1,1]^{T}$, for example, by the error $\varepsilon{\bf e}_{1}$ with $\varepsilon=2$, the increase in the asymptotic state of node $1$, that is ${\bf x}_{\infty,1}^{\star}-{\bf x}_{\infty,1}=\varepsilon [e^{\bf A}]_{11}$, is equal to $7.641$ for the signed network and to $7.963$ for the unsigned network. The difference for a unitary error, equal to $[e^{|{\bf A}|}]_{11}-[e^{\bf A}]_{11}=0.161$, increases linearly with $\varepsilon$ and the ratio $[e^{\bf A}]_{11}/[e^{|{\bf A}|}]_{11}=0.960$ is exactly the local balance of node $1$.
	The eigenvalues of ${\bf A}$ are $\lambda_1=1.956$ $\lambda_2=1.705$. $\lambda_3= 0.317$ and $\lambda_4=0.022$, while the eigenvalues of $|{\bf A}|$ are $\overline{\lambda}_{1}=2.239$ $\overline{\lambda}_{2}=1.274$. $\overline{\lambda}_{3}=0.439$ and $\overline{\lambda}_{4}=0.047$. 
	The two exact condition numbers are $\mathscr{K}(e^{-{\bf A}})=30.378$ and
	$\mathscr{K}(e^{-{{|\bf A}|}})=30.887$, and their exact ratio is therefore
	${\mathscr R}({\bf A})=0.984$.
	For such a low-dimensional correlation matrix, the difference between the two condition numbers is small. Nevertheless it evidences the fact that the balanced network shows a greater dependence on the initial data than the unbalanced network.
\end{example}
This numerical example can be extended to a statistical analysis of the relationship between the global balance $\kappa(G)$ and the ratio of the two condition numbers ${\mathscr R}({\bf A})$ in the case of randomly generated correlation matrices. We refer to the Supplementary Material, Sec. 3, for further details.

\subsection{Interpretation in terms of systemic risk measure}

The close relationship between global balance and condition number proved in the previous sections prompts further the comparison toward indicators that generalize the condition number of the correlation matrix in the context of systemic risk.

The information contained in the eigenvalues of the correlation matrix can be used to define systemic risk measures associated with the underlying dataset. Two classes of indicators have been identified in the literature, the first exploiting the information contained in the larger eigenvalues and, the second, the information content of the smaller eigenvalues. The two classes provide alternative views on systemic risk.

The first looks at the greatest eigenvalues to study the process of synchronization among the $N$ variables of interest. The eigenvalues of the covariance (or correlation) matrix $ \bf C$ can be viewed as the amount of variance that is explained by the corresponding eigenvectors. 
In fact, if ${\bf \phi}_i$ represents the normalized eigenvector associated with the eigenvalue $\lambda_{i}$, then the following equalities hold:
\begin{equation}
	\lambda_{i}
	=\lambda_{i}{\bf \phi}_{i}^{T}{\bf \phi}_{i}
	= {\bf \phi}_{i}^{T}{\bf C}{\bf \phi}_{i}
	=\frac{1}{T}{\bf \phi}_{i}^{T}{\bf \tilde X}{\bf \tilde X}^{T}{\bf \phi}_{i}
	=\frac{1}{T}\left({\bf \phi}_{i}^{T}{\bf \tilde X} \right)\left({\bf \phi}_{i}^{T}{\bf \tilde X} \right)^{T}
	={\rm Var}\left({\bf \phi}_{i}^{T}{\bf \tilde X} \right).
	\label{eigenvalue_variance}
\end{equation}
Eq. \eqref{eigenvalue_variance} shows that $\lambda_{i}$ represents the portion of the total variance associated with the $i-{th}$ principal portfolio ${\bf \phi}_{i}$, derived from PCA analysis, as defined in \cite{Meucci2010}. Each component $\lambda_{i} {\bf \phi}_{i}{\bf \phi}_{i}^{T}$ represents a new portfolio out of the original stocks, which is orthogonal to the others. Then the corresponding eigenvalue measures its contribution to the total risk.

In particular, if all the empirical eigenvalues are lower than the upper bound defined by the Marchenko-Pastur law (see \cite{Jonnson1982}), then the common components tracking the empirical data are not statistically different from those generated by a random process, and we can conclude that there is no structural correlation and no synchronization among the variables. Conversely, when some of the largest eigenvalues exceed the upper bound, they represent common drivers that convey information about the covariance structure of the empirical data better than what a null random matrix model could do. 

Specifically, if the total risk of the system is represented by $\sum_{i=1}^{N}\lambda_{i}={\rm tr}\, {\bf C}=N$, the fraction
of the total risk associated with the first $M$ principal components, the so-called \textit{Cumulative Risk Fraction} (CRF) or \textit{Absorption Rate}, is defined as
$CRF={\sum_{i=1}^{M}\lambda_{i}}/{\sum_{i=1}^{N}\lambda_{i}}=\frac{1}{N}\sum_{i=1}^{M}\lambda_{i}$ (see \cite{Billio2012}). As a consequence, when the system is highly interconnected, $M$ out of the $N$ principal components can explain most of the volatility in the system.

The second approach focuses on the eigenvalues $\lambda_{i}, i=1,\dots, N$, as they represent the approximation of the algebraic distances of the empirical observations matrix $\bf \tilde{X}$ from the subspace of rank deficient matrices. In particular, $\lambda_{N}\leq \ldots \leq \lambda_{2} \leq \lambda_{1}$ give the distances between $\bf \tilde{X}$ and the closest rank deficient matrices with ranks $N-1, \ldots, 1,0$, respectively. Therefore, if the eigenvalues $\lambda_{i}$ are small, and so are the singular values of $\bf \tilde{X}$, then the distances of the $\bf \tilde{X}$ matrix from the rank deficient matrices are small. A short distance from the rank deficient matrices corresponds to a high \textit{numerical} linear dependence between the assets. A high numerical linear dependence is equivalent to a reduced number of diversification opportunities for the investor and, hence, a higher risk. The class of indicators that best captures this aspect in terms of minimum distances and thus minimum eigenvalues is that of the \textit{Market Rank Indicators} (MRI) (see \cite{Uberti2020a}). We refer here to a possible realization of these indicators, specifically to the arithmetic version  of the Market Rank Indicator (AMRI) defined in terms of power mean as
\begin{equation}
	AMRI=\frac{\lambda_{1}}{\left(\frac{1}{H}\sum_{i=N-H+1}^{N}\lambda_{i}^{p}\right)^{1/p}}
	\label{AMRI}
\end{equation}
where $1\leq H \leq N$ and $p\in {\mathbb N}$. Note that the AMRI reduces to the condition number in norm $2$, $\lambda_{1}/\lambda_{N}$, for $H=1$. This measure has been shown to belong to the broad class of so-called \textit{proper measures of connectedness} (\cite{Uberti2020b}) and have been used to analyze financial datasets in \cite{Uberti2023}.\footnote{Let us observe that, in graphs theory, the quantity $\frac{1}{N}{\rm Tr}\, \bf {C}^{p}=\frac{1}{N}\sum_{i=1}^{N}\lambda_{i}^{p}$ represents the $p-$moment of the adjacency matrix $\bf C$ and, in a positive unweighted network, represents the mean of the number of $p-$cycles around the nodes, that is the average number of $p-$cycles to which a node in the network belongs.}

The local and global balance indices introduced in Eqs. \eqref{Localbalance} and \eqref{Globalbalance} are defined in a way completely independent of both the class of systemic risk measures based on the larger eigenvalues and the class based on the smaller eigenvalues, since they imply a comparison between the correlation matrix and its entry-wise absolute value. However, the interpretation proposed in the previous sections suggests that a bridge can be established between the domain of balance in signed networks and the domain of systemic risk measures for asset correlation networks. Indeed, the relationship of the balance index with the condition numbers of the diffusion problem and with the ratio ${\mathscr R}({\bf A})$ defined in Eq. \eqref{theratio} indicates that they are candidates to be indicators of systemic risk as a relative measure of the effect of correlation signs on the spread of information in the underlying network.\footnote{It is also possible to show that the global balance trivially satisfies the four properties required in \cite{Uberti2020b} for a proper measure of connectedness.}
As a consequence,  they can be used to detect the greater or lesser instability/sensitivity of the portfolio behavior with respect to the maximum instability/sensitivity associated with the balanced one, which typically has a higher condition number and, consequently, is less diversified and more exposed to systemic risk. At this point, we want to explore the ability of this new indicator to reveal systemic events on empirical correlation networks and compare its capabilities with that of other known measures in the literature. 

\section{Empirical analysis}
\label{Empirical analysis}

We test the results in Sections 3 and 4 on a real database of financial returns. We focus the analysis on the weighted version of the correlation-based signed network. However, we compare the results with the binary version of the same network. The binarization is obtained by applying a standard threshold chosen in the interval $(0,1)$; above the threshold, the correlations are set equal to $1$, below the corresponding opposite value, the correlations are set equal to $-1$, and are 0 in the remaining cases. The binary network loses the values of the edge weights, making the behavior of the correlations sharper and more extreme. The chosen threshold is $0.25$. \added{The choice of this symmetric threshold, although simplifying and inherently arbitrary, is inspired by what was done in Harary's seminal work \cite{Harary2002} and allows a first comparison of the results obtained on the real weighted network and its binarization.}

The weighted and the binary signed correlation networks are constructed over $\Delta T$-days windows, rolling at $\Delta t$-days step, that is, the returns are split by using $\Delta t$-days stepped windows of $\Delta T$-days width and the data from each window are used to construct the correlation matrix of that window. The appropriate choice of $\Delta T$ is a critical point when using eigenvalue-based indicators, since a full rank correlation matrix is required.

\subsection{Comparison between the Global Balance and the Market Rank Indicator}

The database contains the daily returns of the stocks of the S\&P500 index, from January 4, 2005 to March 18, 2020. The number of assets is $385$, smaller than the nominal one, since we limited our analysis to the stocks with a complete time series on the whole period. The entire dataset spans $T=4109$ days.
We start with a direct comparison between the behavior of the global balance and that of some systemic risk indicators. The time windows width is $\Delta T=400$ days and a step of $\Delta t=30$ days is applied. We note that $\Delta T$ is larger than the number of assets; this permits to calculate the eigenvalue-based indicators. First, we consider the time series of the global balance and the average correlation in each time window. The correlation between these two time series is generally high. The Pearson correlation is equal to $\rho_{\rm Pearson}=0.5902093$, while rank correlation is equal to $\rho_{\rm Spearman}=0.8771375$. In the Supplementary Material, Sec. 4, we show the time series plots of the global balance and of the average correlation over the period 2005-2020. We now focus on the comparison between the behavior of the global balance and that of the market rank indicator. Their evolution in time are shown in Fig. \ref{fig2}.

Fig. \ref{fig2}, panel (a), represents the temporal evolution of the global balance (black line) and of the market rank indicator (blue, red, and green line) over the interval January 2005-October 2016. The three plots for the market rank indicator refer to different choices of the number of smaller eigenvalues used for its computation: the blue line refers to 10 minimum eigenvalues, the red line to 50 and the green line to 100.
Similarly, in panel (b), for the period November 2016-January 2020. In the definition of MRI, we have assumed $p=3$ in the power mean used in the denominator (see Eq. \eqref{AMRI}). The correlations between the global balance and the MRI for the two intervals and for the different parameters are listed in Table \ref{Table1}. The last row of the Table ($H=N=385$) returns the special case of the cumulative risk fraction indicator (see \cite{Billio2012}).

As expected, these correlations do not change dramatically if we use a different value of $p$ in the definition of MRI: for instance, for the weighted version, for $p=4$ and $H=10$ we get $\rho_{\rm Pearson}=0.5096713$ and $\rho_{\rm Spearman}=0.7500659$. Similarly for other values of $p$ and $k$.

\begin{figure}[H]
	\centering
	\subfloat[]{\includegraphics[width=1.00\textwidth]{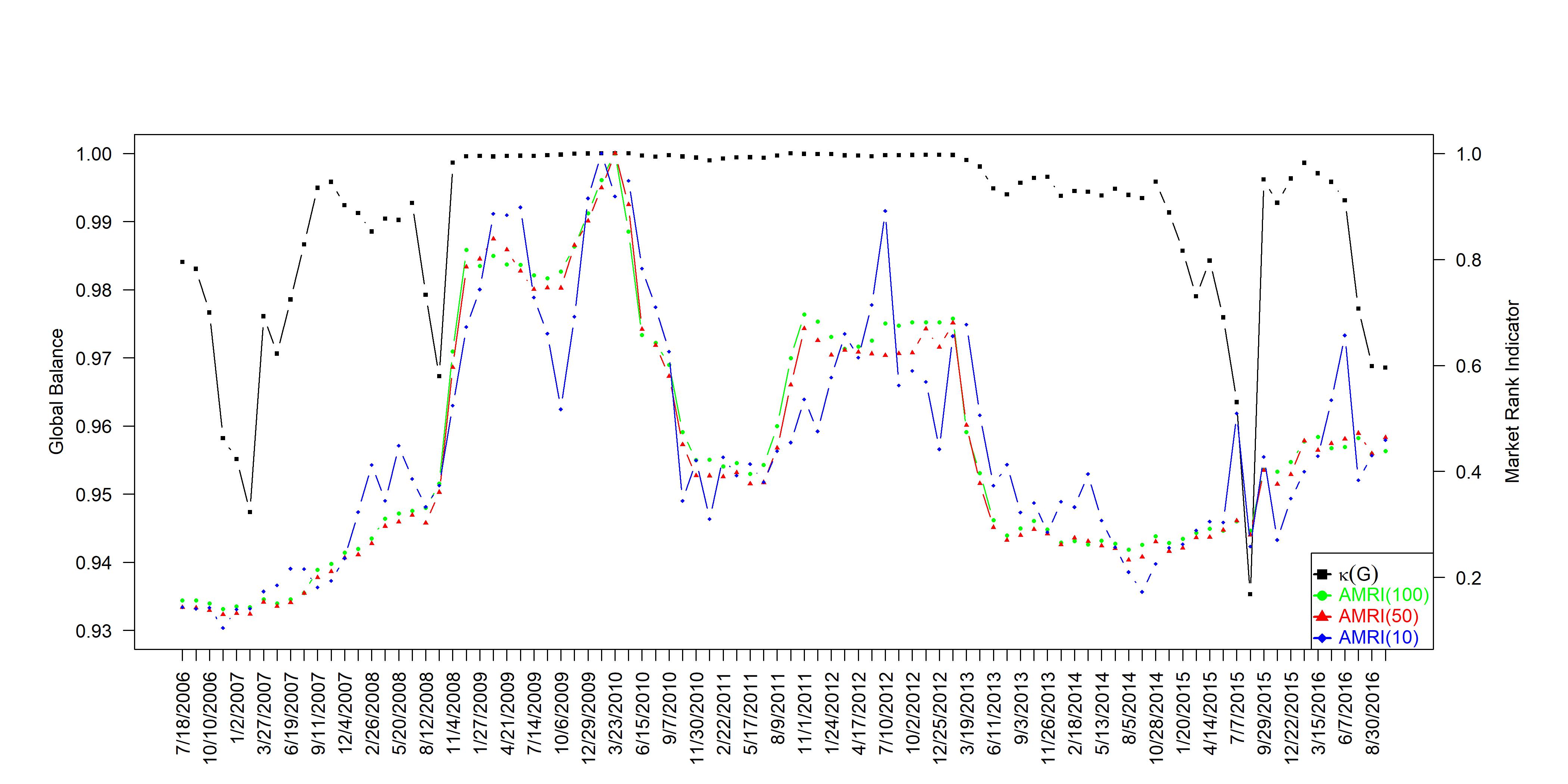}}\\
	\subfloat[]{\includegraphics[width=1.00\textwidth]{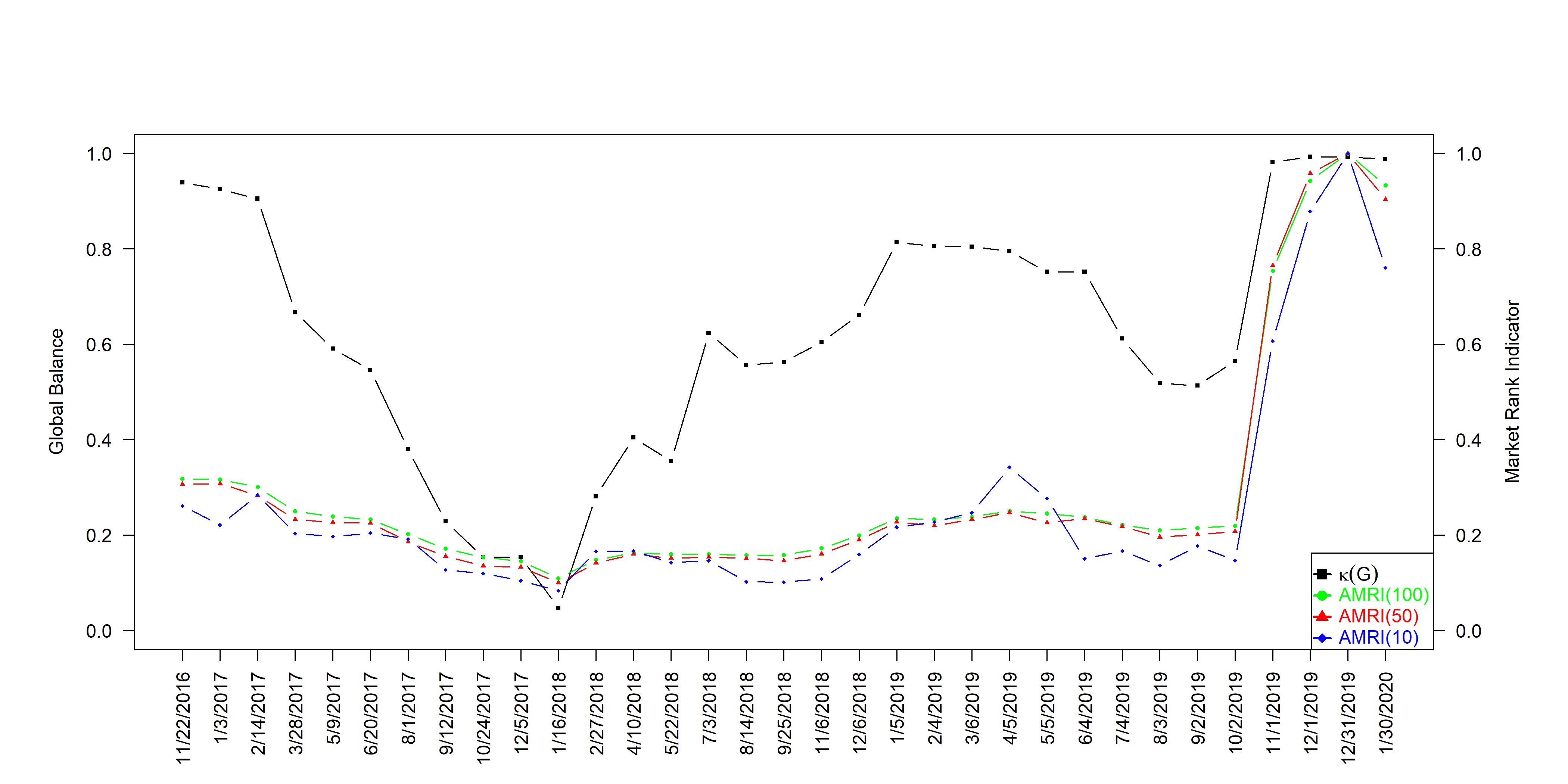}}
	\caption{Time evolution of the global balance (black square dot line) and of the market rank indicator for $p=3$ and $H=10$ (blue kite dot line), $H=50$ (red triangle dot line) and $H=100$ (green circle dot line) for the period (a) 2005-2016 and (b) 2016-2020.}
	\label{fig2} 
\end{figure}

\begin{table}[H]
	\footnotesize
	\begin{center}
		\begin{tabular}{c|c|c|c|c|}
			\cline{2-5}
			& \multicolumn{4}{c|}{\bf Correlations between Global balance and MRI} \tabularnewline
			\cline{2-5}
			& \multicolumn{2}{c|}{\bf 2005-2016} &\multicolumn{2}{c|}{\bf 2016-2020} \tabularnewline
			\cline{1-5}
			\multicolumn{1}{|c|}{\bf H} & Pearson & Spearman & Pearson & Spearman \tabularnewline
			\cline{1-5}
			\multicolumn{1}{|c|}{\bf 10}
			& $0.51043890$
			& $0.75268360$
			& $0.65198390$
			& $0.81116310$
			\tabularnewline \hline
			\multicolumn{1}{|c|}{\bf 50}
			& $0.53938464$
			& $0.81110470$
			& $0.65998890$
			& $0.90441180$
			\tabularnewline \hline
			\multicolumn{1}{|c|}{\bf 100} 
			& $0.55415830$
			& $0.82765060$
			& $0.65877970$
			& $0.89204550$
			\tabularnewline \hline
			\multicolumn{1}{|c|}{\bf 385} 	
			& $0.87786900$
			& $0.89459170$
			& $0.88014410$
			& $0.98061500$
			\tabularnewline \hline
			\hline
		\end{tabular}
	\end{center}
	\caption{Correlations between Global balance and MRI in the two intervals 2005-2016 and 2016-2020, for different values of $H$ in the definition of the MRI.}
		\label{Table1}
\end{table}

We can extend the analysis to the correlations between the AMRI and the ratio ${\mathscr R}(G)$ in Eq. \eqref{ratioandbalances}. These correlations are very close to the previous ones. Indeed, both $\kappa(G)$ and ${\mathscr R}(G)$ produce time series with very similar behaviors, making the two measures almost equivalent for application purposes. This fact can be explained by the presence of a large spectral gap in eigenvalues of the correlation matrices for each time window. In this case, it is possible to apply the approximate formula for the ratio in Eq. \eqref{approximateratio}. This explains the almost-1 correlation between $\kappa(G)$ and ${\mathscr R}(G)$. We observe that the error between the exact value of the ratio in Eq. \eqref{ratio} and its approximate version in Eq. \eqref{approximateratio} is very small. Until October 2016, the mean relative error is $0.033\%$, with a maximum equal to $0.201\%$ in August, 2015. In the period from November 2016 to January 2020, the approximation error is on average $1.22\%$ with a max of $3.91\%$ in October 2017. The maximum is achieved in correspondence with a severe reduction of the spectral gap in 2017, see Fig. \ref{fig3}. We also observe that the network is perfectly balanced in the windows ending in February to May 2010, and starting at the end of 2008.
\begin{figure}[H]
	\centering
	\includegraphics[width=0.80\textwidth]{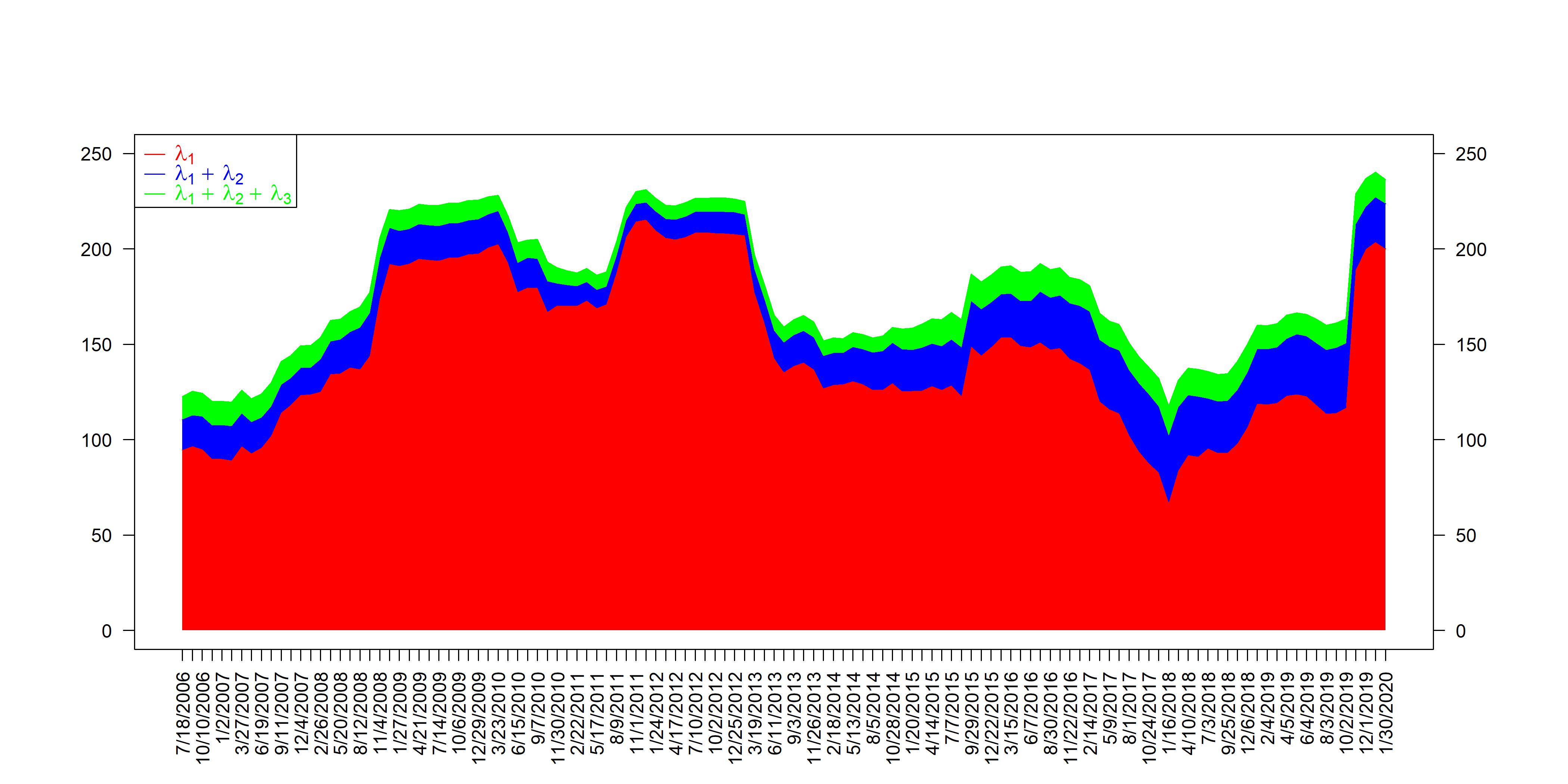}
	\caption{Time series of the largest eigenvalues of the correlation matrix: $\lambda_1$ in red, $\lambda_1+\lambda_2$ in blue and $\lambda_1+\lambda_2+\lambda_3$ in green.}
	\label{fig3}
\end{figure}
The previous analysis required rather long time windows due to the large number of assets in the entire dataset and the condition $T\geq N$. This means that correlations are computed over long time intervals that do not allow local variations to be captured, implying a loss of resolution.
Therefore, we repeat the previous analysis on a subset of assets from the same dataset, consisting of 50 randomly chosen assets, using 100-days time windows, always shifted by 30 days-length window. The MRI indicator has been computed using $H=2$, $H=5$ and $H=10$ smallest eigenvalues.

For the selected subset of assets, we confirm the strong correlation between the global balance and the average correlation: $\rho_{\rm Pearson}=0.6297643$ and $\rho_{\rm Spearman}=0.8912235$.

Similarly, the correlations between the global balance and, for instance, the AMRI for $H=2$ are $\rho_{\rm Pearson}=0.5307341$ and $\rho_{\rm Spearman}=0.7685456$, on the period 2005-2016 and $\rho_{\rm Pearson}=0.4994806$ and $\rho_{\rm Spearman}=0.9441659$, on the period 2016-2020. The correlation between the global balance and the Cumulative Risk Fraction are $\rho_{\rm Pearson}=0.9096708$ and $\rho_{\rm Spearman}=0.9203074 $, on the period 2005-2016 and $\rho_{\rm Pearson}=0.9297902$ and $\rho_{\rm Spearman}=0.991976$, on the period 2016-2020. To test the robustness of our findings, we repeated the exercise over many different subsets of 50 assets randomly chosen from the original index, obtaining similar results (see Fig. \ref{fig4}).
\begin{figure}[H]
	\centering
	\subfloat[]{\includegraphics[width=0.50\textwidth]{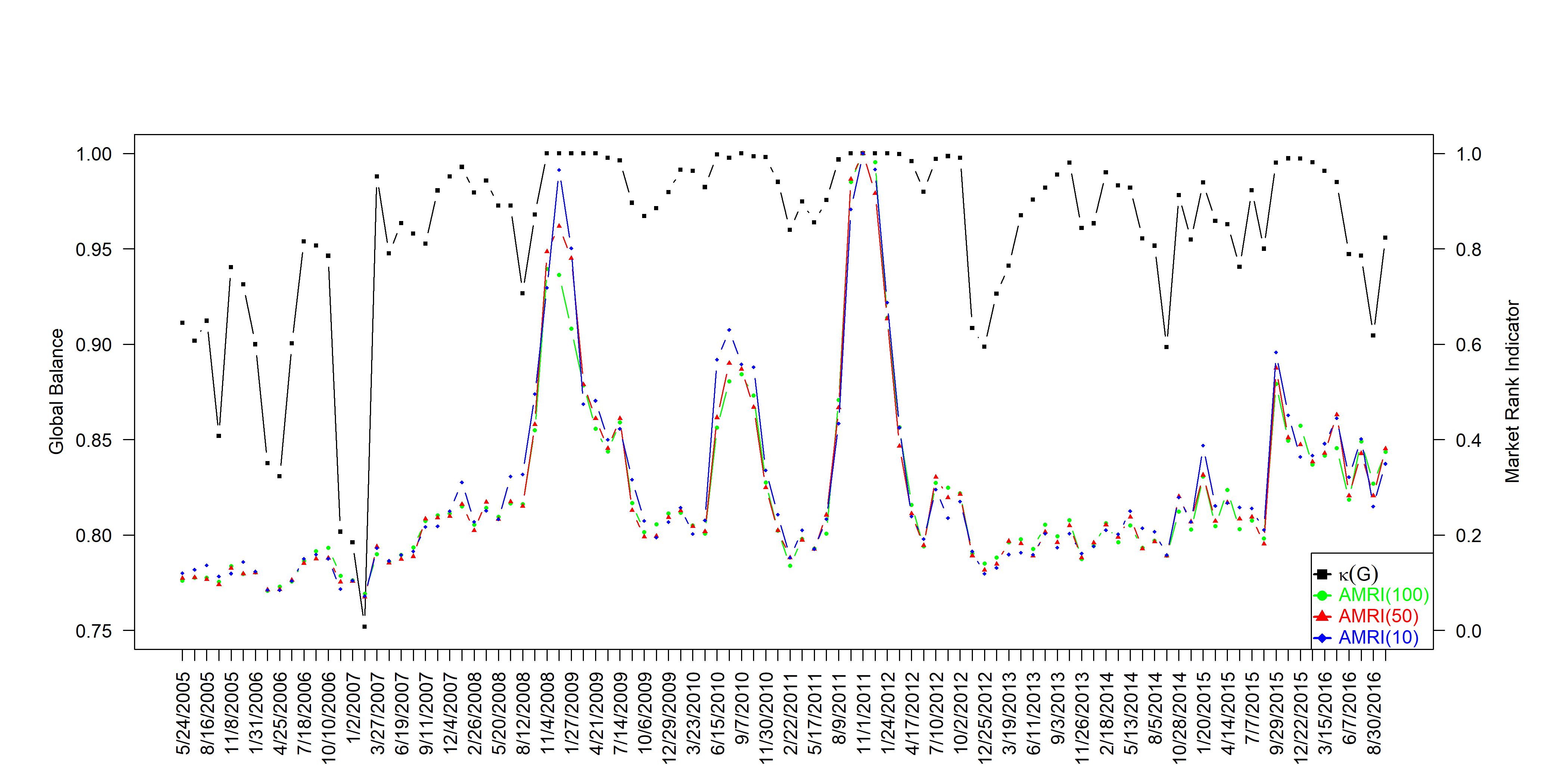}}
	\subfloat[]{\includegraphics[width=0.50\textwidth]{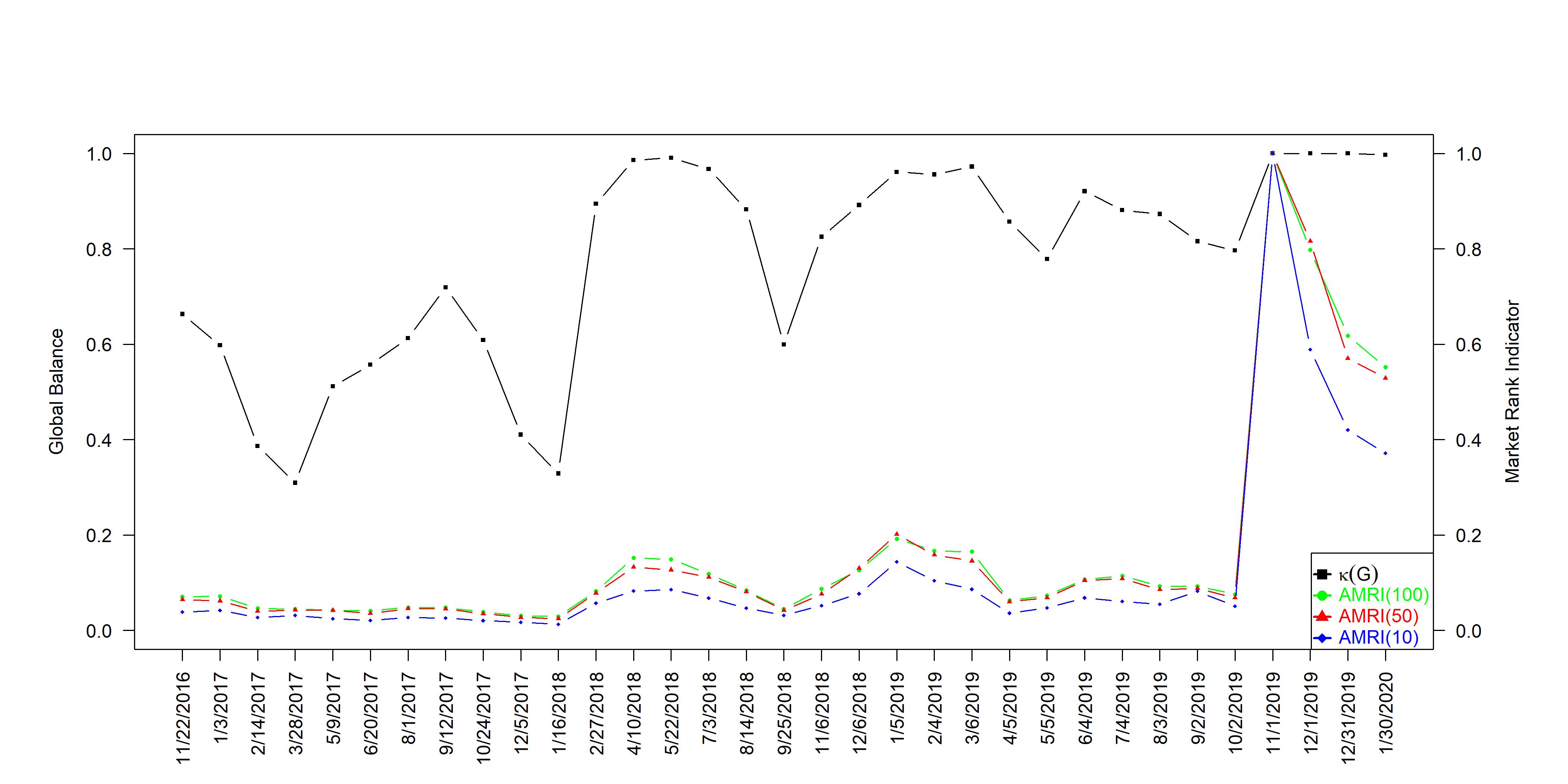}}
	\caption{Time evolution of the global balance (black square dot line) and of the market rank indicator for $p=3$ and $H=2$ (blue kite dot line), $H=5$ (red triangle dot line) and $H=10$ (green circle dot line) for the period (a) 2005-2016 and (b) 2016-2020. The figures refer to a subset of $50$ randomly selected assets.}
	\label{fig4} 
\end{figure}
We can finally conclude that there is a high correlation between the global balance and the considered systemic risk measures.

\subsection{Effectiveness of the Global Balance as a Systemic Risk Measure}

We perform now an empirical analysis to assess the effectiveness of the global balance in detecting systemic events.
In order to identify systemic events, we refer to the definition in \cite{Uberti2023}. A systemic risk event occurs if the average return of all the $N$ securities falls below a threshold $\tau$ for $\Delta T$ consecutive days. We adopt a threshold $\tau=-0.01$ or $\tau=-0.005$ and sliding windows of $\Delta T=20$ days. \vspace{-0.5cm}
\begin{figure}[H]
	\centering
	\subfloat[]{\includegraphics[width=0.50\textwidth]{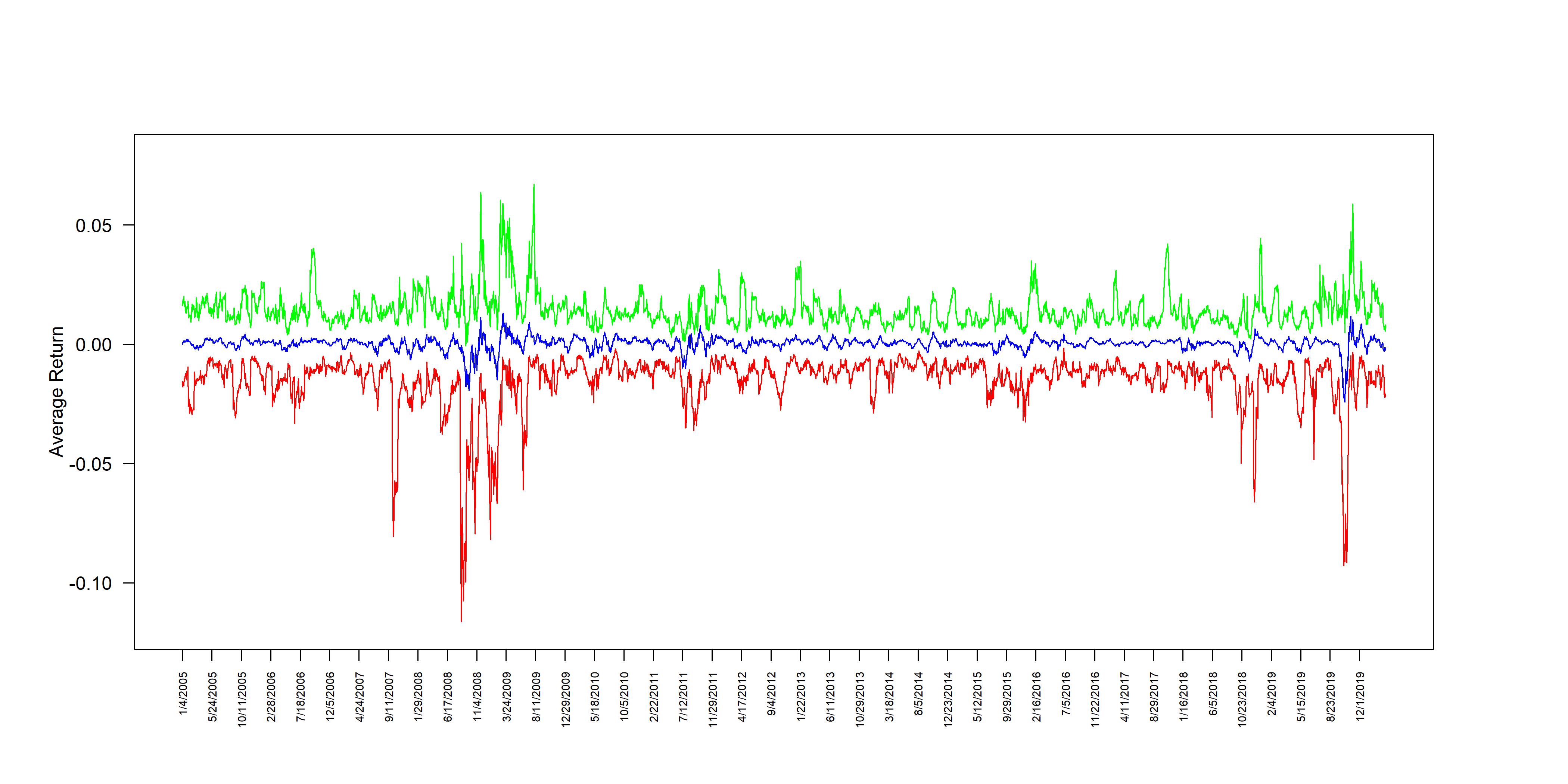}}
	\subfloat[]{\includegraphics[width=0.50\textwidth]{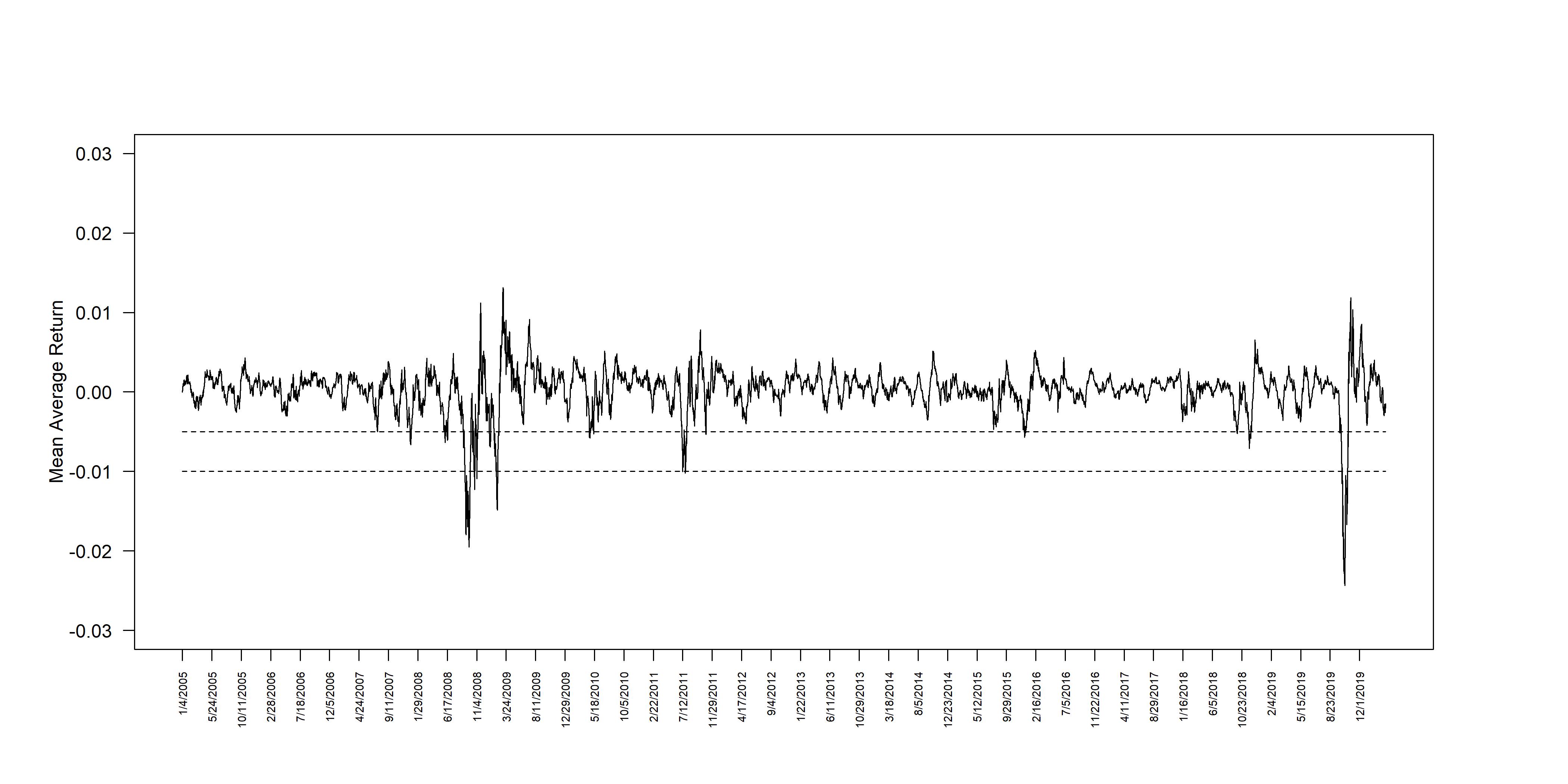}}
	\caption{Panel (a): Average returns of the S\&P500 database in the period 2005-2020. The average returns of each asset have been computed over a $20$-days-wide sliding window with $1$-day step. The blue line represents the mean over all the assets in the same window of the average returns. The green and the red lines represent the maximum and the minimum, respectively, of the average returns in the same window. Panel (b) focuses on the mean of the average returns and highlights crisis events when the line goes below the two possible threshold represented by the two horizontal dashed lines.}
	\label{fig5} 
\end{figure}
In Fig, \ref{fig5}, panel (a), we plot the average returns of the S\&P500 database from January 4, 2005 to March 18, 2020. The average returns of each asset have been computed over a $20$-days-wide sliding window with $1$-day step. The blue line represents the mean of the average returns over all the $385$ assets in the same window. The green and the red lines represent the maximum and the minimum value, respectively, of the average returns in the same window. Panel (b) focuses on the mean of the average returns and highlights crisis events when the line goes below the two possible thresholds represented by the two horizontal dashed lines.

In this way, we can clearly identify systemic crisis events when the mean of the average returns over all the assets goes below $\tau$, for both $\tau=-0.005$ and $\tau=-0.010$. In particular, the crisis events below $\tau=-0.010$ ranges in the time intervals September-November 2008, February 2009, July 2011, October 2019.

We study the joint distribution of the average returns and of the global balance in the sliding windows. The range of the indicator $\kappa(G)$ has been divided into $n$ bins $B_1, B_2 , \dots , B_{n}$ and the conditional distributions of the average returns $P\left( \langle X_{it}\rangle |\, \kappa(G)\in B_{k}\right)$ have been computed on each bin. Specifically, in a first experiment the interval $[0,1]$ is divided into $5$ bins of equal width while, in a second scenario, the five intervals are $(0,0.5]$, $(0.5,0.8]$, $(0.8,0.9]$, $(0.9,0.99]$ and $(0.99,1.00]$. The necessity of testing what happens with intervals of different lengths, very small when close to 1, is a peculiarity in the framework of systemic risk. Since systemic risk events are, hopefully, rare compared to the periods in which the market behaves normally, a systemic risk indicator should discriminate a few events within a multitude.

\begin{figure}[h]
	\centering
	\subfloat[]{\includegraphics[width=0.50\textwidth]{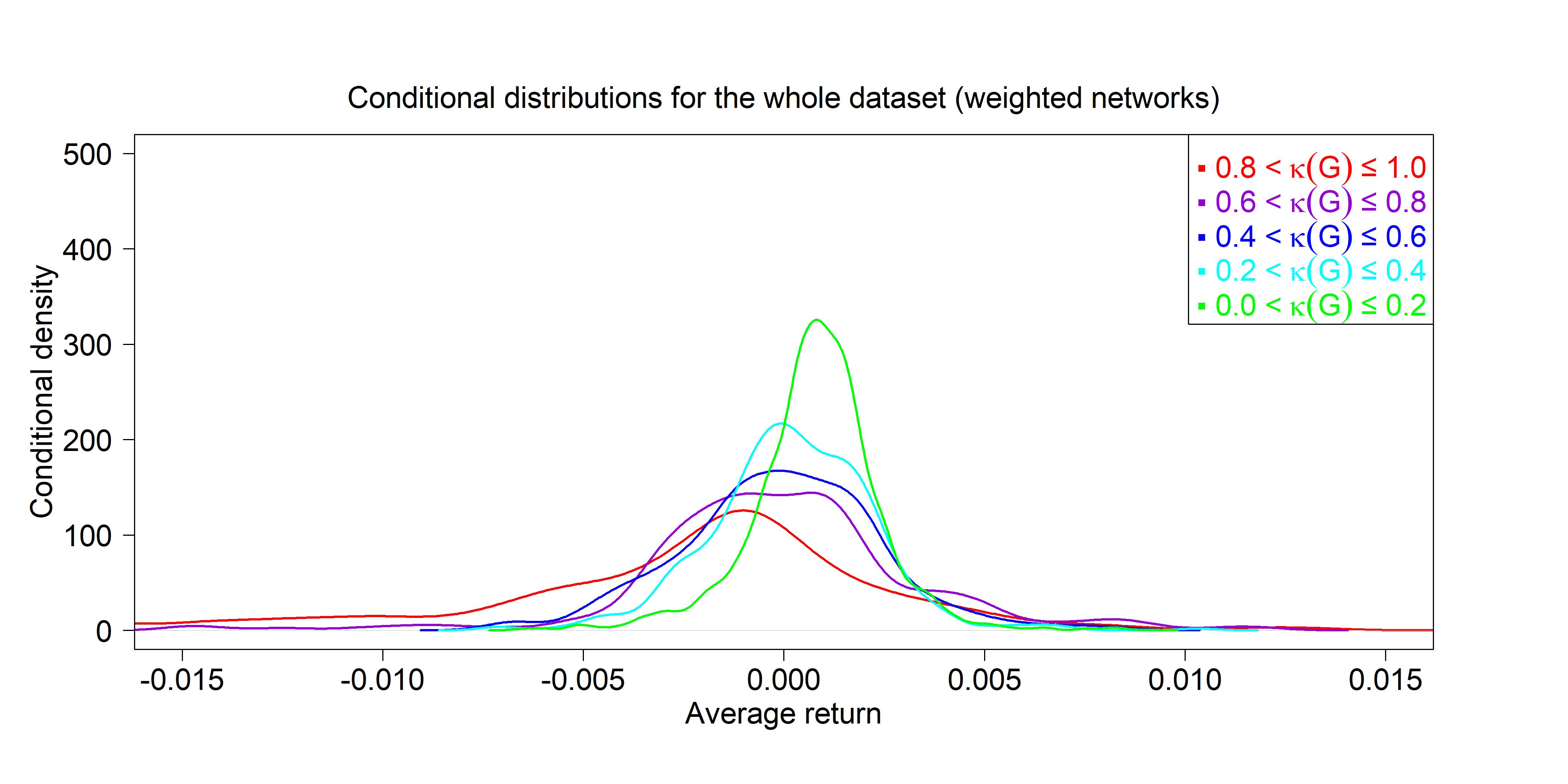}}
	\subfloat[]{\includegraphics[width=0.50\textwidth]{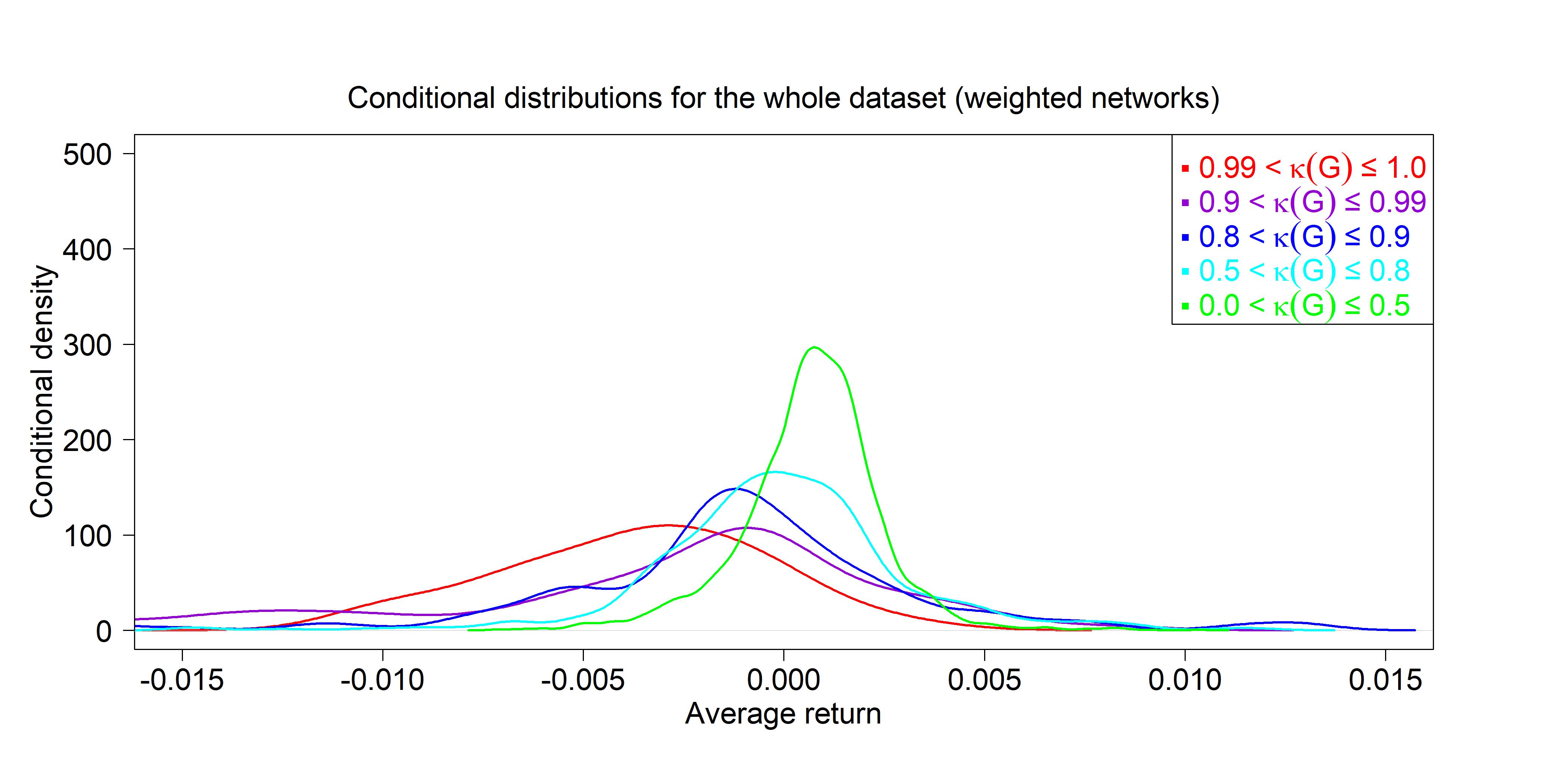}}
	\\
	\subfloat[]{\includegraphics[width=0.50\textwidth]{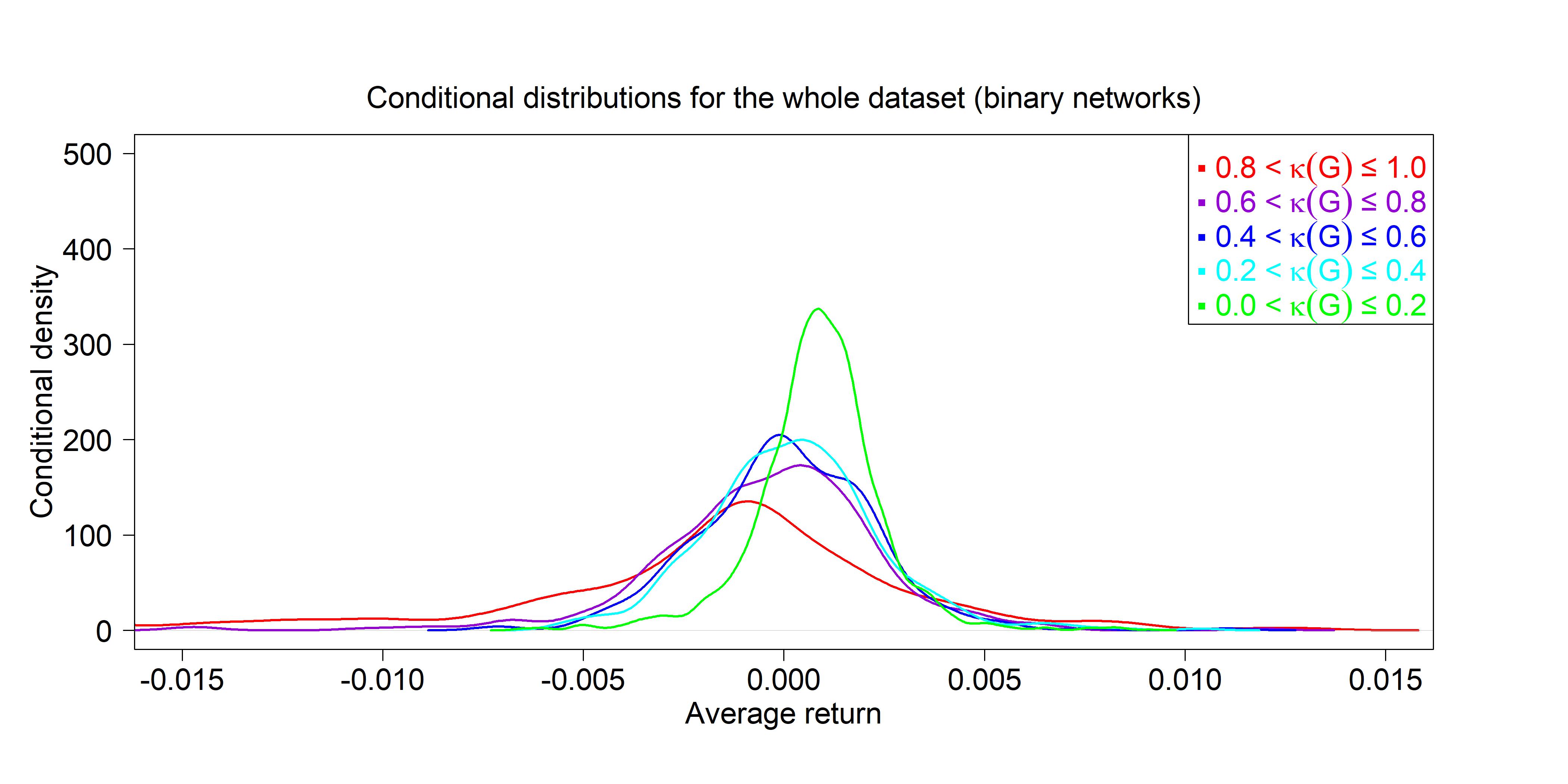}}
	\subfloat[]{\includegraphics[width=0.50\textwidth]{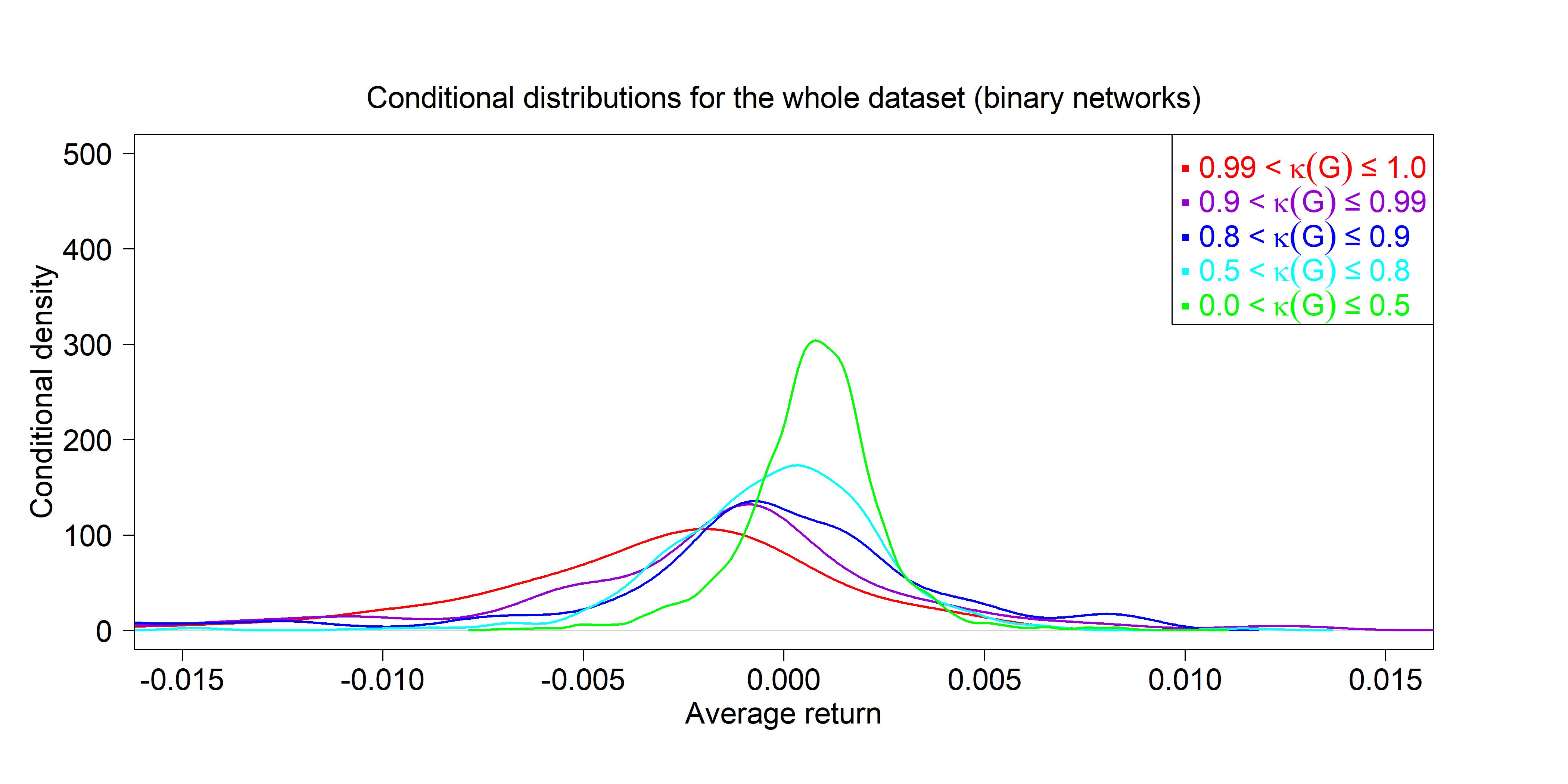}}
	\caption{Conditional distributions of the $20-$days average returns with respect to the value of $\kappa(G)$, for the whole network of $385$ assets. In panels (a) and (b), the densities are computed on the weighted version of the correlations networks, while, in panels (c) and (d), the densities are computed on the binary version of the correlation networks. Panels (a) and (c) refer to equal width intervals and panels (b) and (d) to the partition described in the text.}
	\label{fig6}
\end{figure}

Figure \ref{fig6} displays the densities of the average returns conditioned to the values of $\kappa(G)$, obtained by using a Gaussian kernel smoothing for the whole dataset of $385$ assets (panels (a) and (b) for the weighted version of the correlation networks and panels (c) and (d) for the binary version). When $\kappa(G)$ increases, the densities move leftwards and flatten. More precisely, three effects can be recognized in relation to the increase of $\kappa(G)$: the average of the conditional distribution decreases; the standard deviation increases; the VaR\footnote{The Value at Risk, $p$ VaR, measures the loss that will not be exceeded with a given confidence probability $p$ over a certain time horizon (see, for instance, \cite{Manganelli2001}). We adopt $p=0.05$.} of the average returns increases. Fig. \ref{fig7} shows the boxplots of the same distributions conditioned on the two different partitions of the $[0,1]$ interval for the global balance in the weighted and binary version. The mean of the average returns within each interval decreases as $\kappa(G)$ increases, while the standard deviation increases. This is particularly noticeable by considering the extreme intervals in Fig. \ref{fig7}, panels (b) and (d), which refer to the events in the interval $\kappa(G) \in (0.99,1.00]$. The values of the mean of the conditional distributions in Fig. \ref{fig7} are collected in Table \ref{Table2}. The means of the smoothed distributions in Fig. \ref{fig6} are reported in the Supplementary Material, Sec. 5. Both are characterized by a strictly negative monotonic behavior as $\kappa(G)$ increases. Finally, we computed the VaR within each band. For instance, for the whole network in the weighted version and with intervals of equal width the VaR's are: $0.02384863$ for $0<\kappa(G)\leq 0.2$;  $0.0281525$, for  $0.2<\kappa(G)\leq 0.4$;  $0.03010301$, for $0.4<\kappa(G)\leq 0.6$; $0.03521984$, for $0.6<\kappa(G)\leq 0.8$; and $0.05100622$ for $0.8<\kappa(G)\leq 1.0$. \vspace{-0.5cm}
\begin{figure}[H]
	\centering
	\subfloat[]{\includegraphics[width=0.45\textwidth]{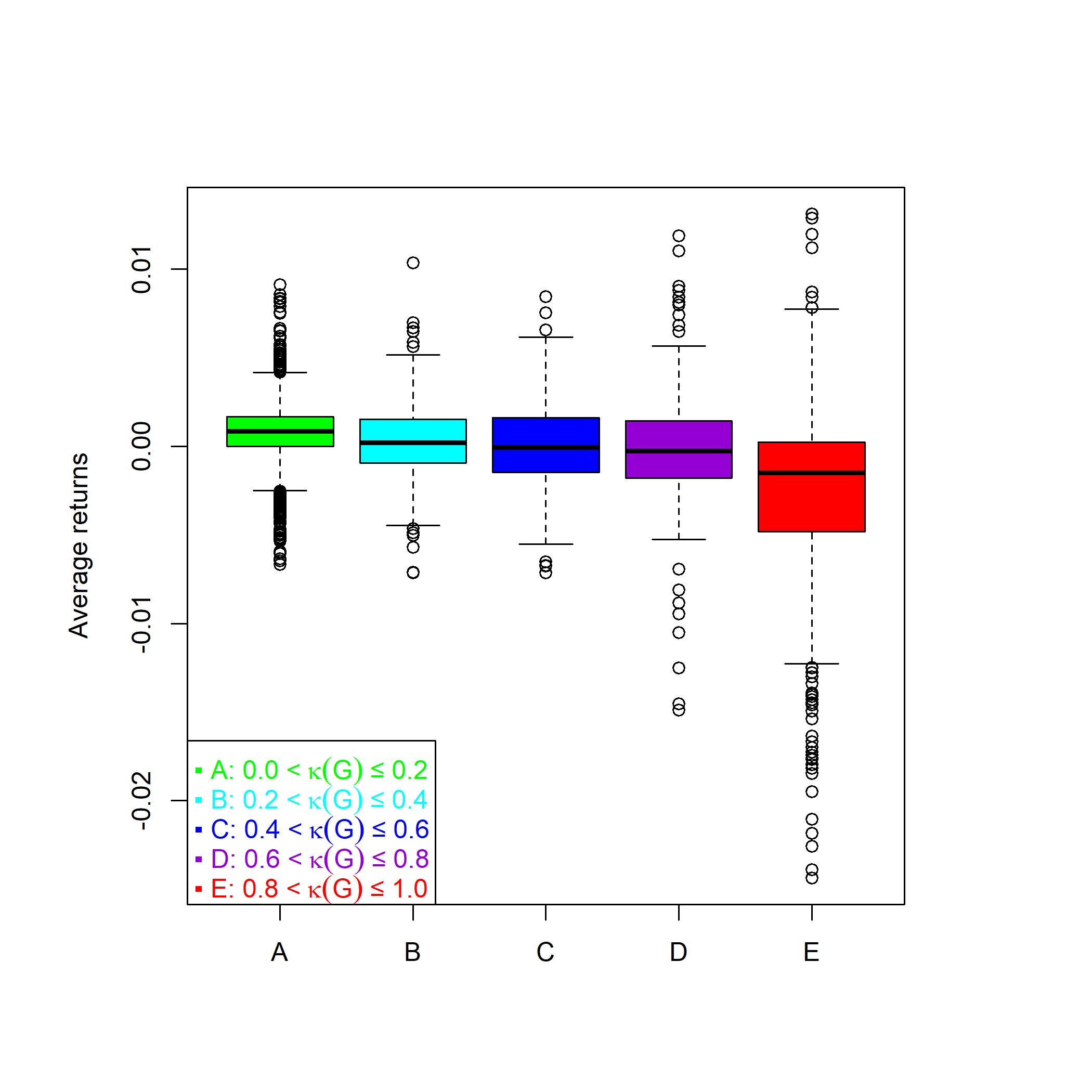}}
	\subfloat[]{\includegraphics[width=0.45\textwidth]{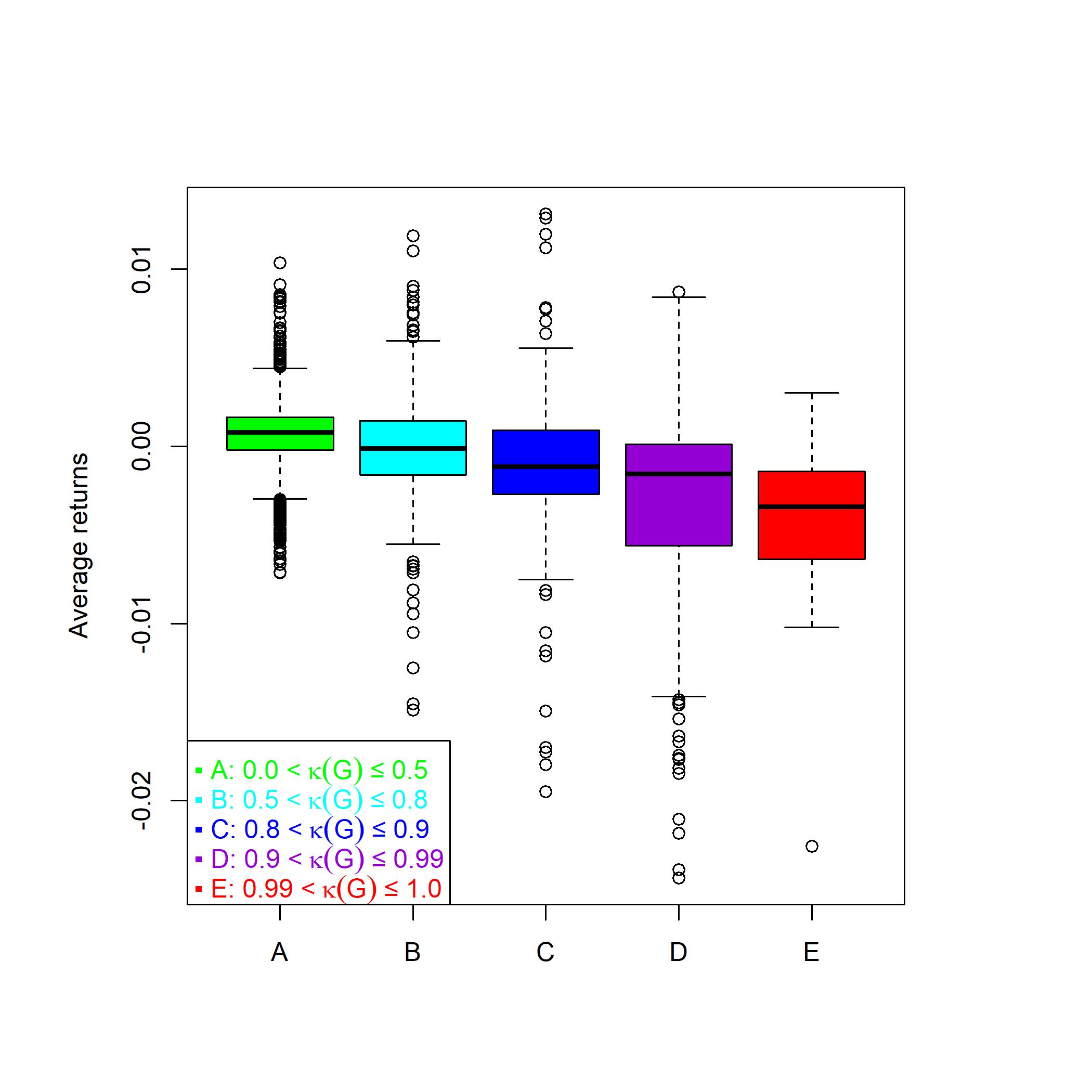}} \vspace{-0.4cm}
	\\
	\subfloat[]{\includegraphics[width=0.45\textwidth]{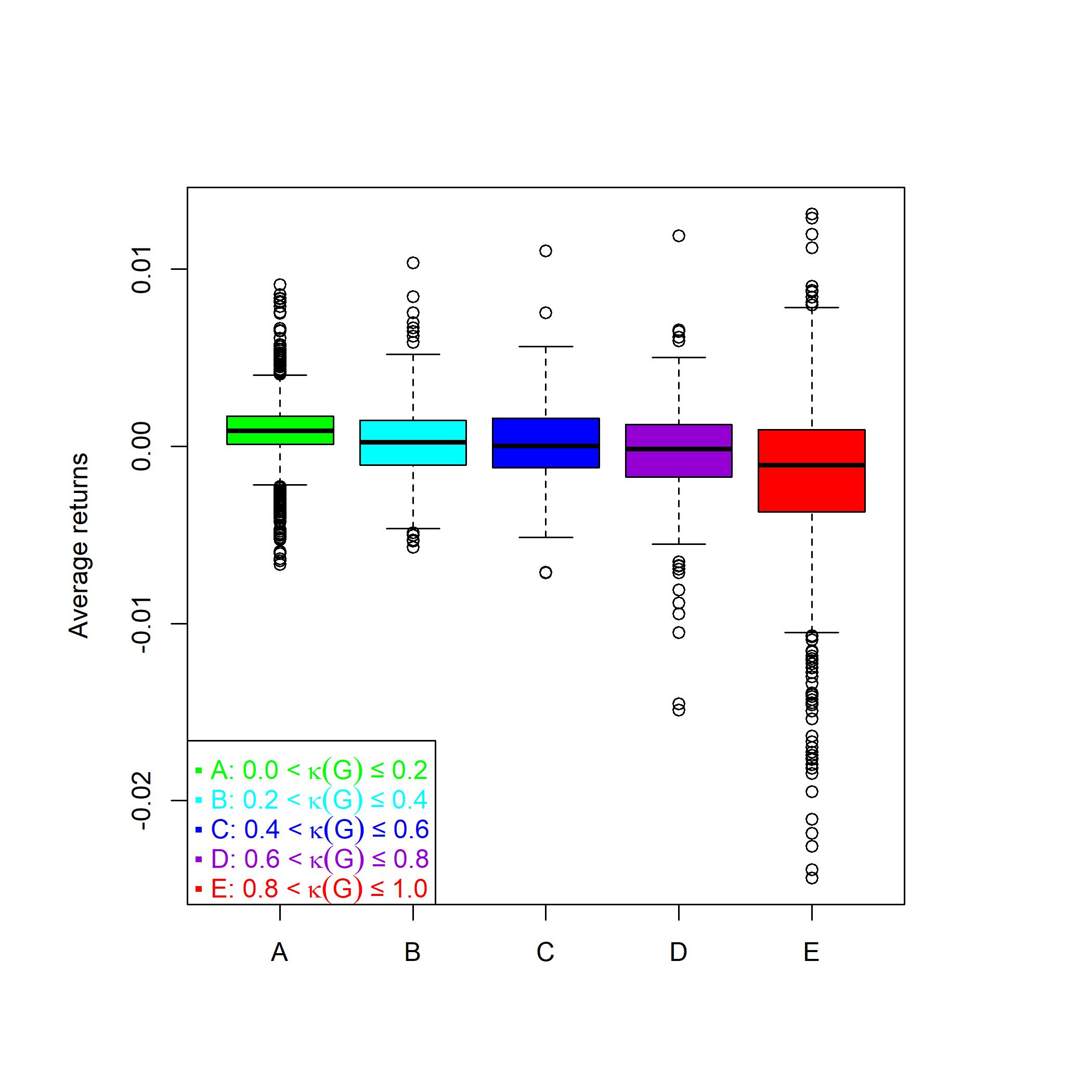}}
	\subfloat[]{\includegraphics[width=0.45\textwidth]{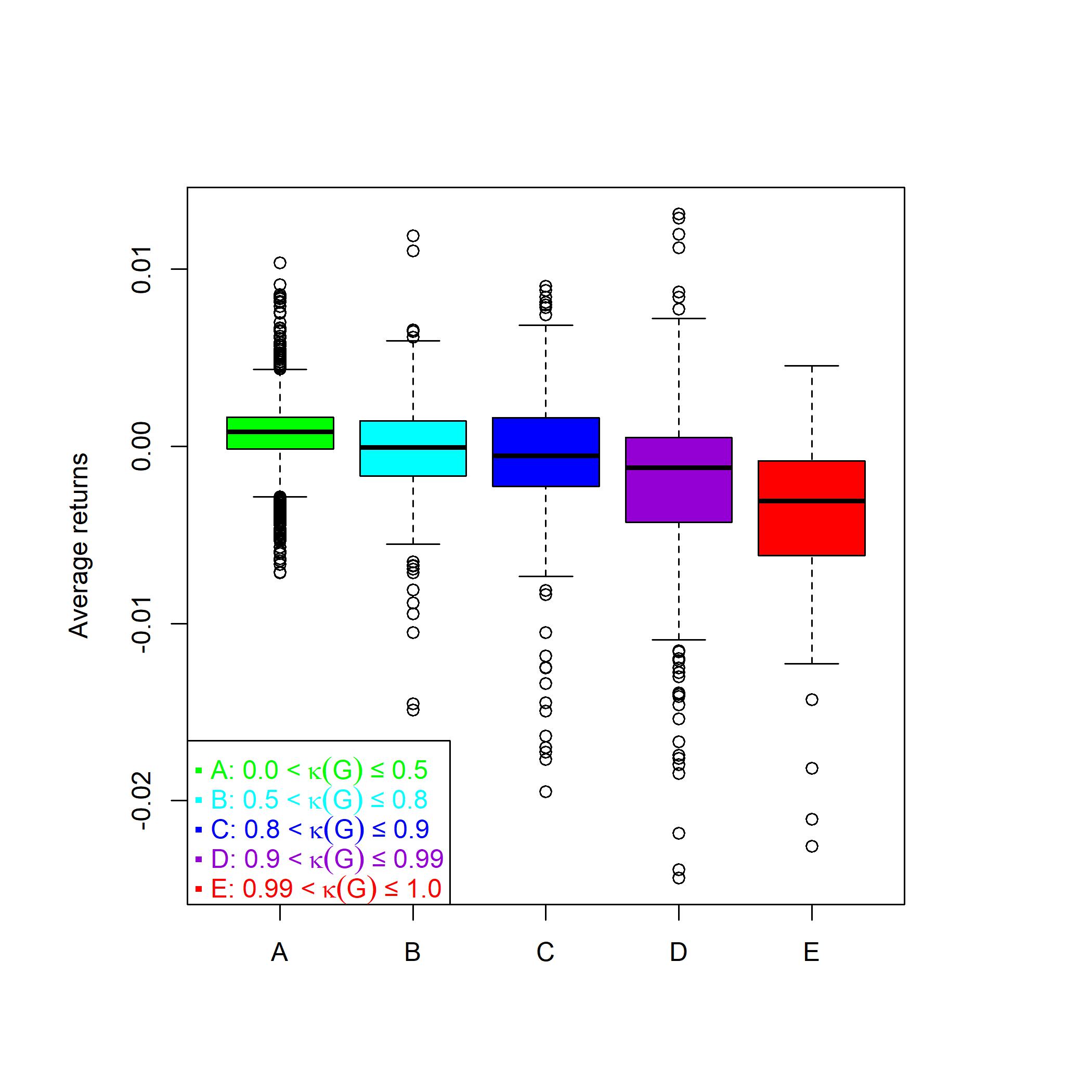}}
	\caption{Boxplots for the $20-$days average returns conditioned on the value of $\kappa(G)$, for the whole network of $385$ assets. Panels (a) and (b) refer to the weighted version of the correlations networks, while panels (c) and (d) to the binary version of the correlations networks. Panels (a) and (c) refer to equal width intervals and panels (b) and (d) to the partition described in the text.}
	\label{fig7} 
\end{figure}
\vspace{-1cm}

\begin{table}[H]
	\footnotesize
	\begin{center}
		\begin{tabular}{c|c|c|c|c|c|}
			\cline{2-6}
			& \multicolumn{5}{c|}{Mean of the average returns} \tabularnewline \cline{1-6}
			\multicolumn{1}{|l|}{\bf Interval} & $0<\kappa(G)\leq 0.2$ & $0.2<\kappa(G)\leq 0.4$ & $0.4<\kappa(G)\leq 0.6$ & $0.6<\kappa(G)\leq 0.8$ & $0.8<\kappa(G)\leq 1.0$  \tabularnewline \hline
			\multicolumn{1}{|c|}{\bf Weighted}
			& $0.0007890292$
			& $0.0002017107$
			& $-0.00008924373$
			& $-0.00009901411$
			& $-0.002686201$
			\tabularnewline \hline
			\multicolumn{1}{|c|}{\bf Binary}
			& $0.0008529778$
			& $0.0002181775$
			& $0.00006867374$
			& $-0.0004070862$
			& $-0.00202211$
			\tabularnewline \hline
			\hline
			\multicolumn{1}{|l|}{\bf Interval} & $0<\kappa(G)\leq 0.5$ & $0.5<\kappa(G)\leq 0.8$ & $0.8<\kappa(G)\leq 0.9$ & $0.9<\kappa(G)\leq 0.99$ & $0.99<\kappa(G)\leq 1.0$  \tabularnewline \hline
			\multicolumn{1}{|c|}{\bf Weighted} 
			& $0.0006679086$
			& $-0.00009905363$
			& $-0.001337312$
			& $-0.003448823$
			& $-0.004218843$
			\tabularnewline \hline
			\multicolumn{1}{|c|}{\bf Binary} 
			& $0.0007200273$
			& $-0.0002525449$
			& $-0.001012558$
			& $-0.002203236$
			& $-0.003917621$
			\tabularnewline \hline
		\end{tabular}
	\end{center}
	\caption{Mean of the average returns in the different balance bands. Weighted and binary refer to the two versions of the correlation networks.}
	\label{Table2}
\end{table}
This evidence shows that the global balance is a meaningful systemic risk measure. Indeed, a large value of $\kappa(G)$ signals that the market is experiencing worse performance and that the probability of big losses is high. We test the robustness of these numerical results on a database (ESX50) independent of the one analyzed here. We show the relative results in the Supplementary Material, Sec. 6.

\section{Conclusion}
\label{Sect. Conclusion}
This paper provides a bridge between the realm of balance in signed networks and that of systemic risk measures for asset correlation networks.
The link between the two frameworks is built through the definition of a discrete diffusive process that drives the spread of the information on the network. On the one side, the process permits to derive the global balance indicator in an alternative way with respect to the standard combinatoric approach. On the other side, the steady state of the diffusive process is the solution of a linear system where the coefficients are the elements of the exponential matrix. Then the relation between the condition number of the exponential matrix and the global balance of the network follows in a straightforward way. 
The structural predictability of a network refers to the monotone behavior of the nodes in response to a perturbation at one or more nodes, and is equivalent to the numerical stability of a linear system.
We prove interesting properties that shed light on the relationship between the global balance and the condition numbers of the signed and unsigned problems. We also introduce an almost complete characterization of signed networks that consider the number of possible odd/even positive/negative closed walks.
By replacing the adjacency matrix with the correlation matrix of financial returns, we apply the global balance index to systemic risk detection.
This step is natural considering that there exists a whole class of systemic risk measures proposed in the literature that are generalizations of the condition number of the correlation matrix.   
The application to two different databases of real financial data supports and confirms our idea. The global balance of the correlation network can be used as a systemic risk indicator. The main advantage of this approach is that it does not require a full rank correlation matrix. Consequently, the global balance of the correlation network can be calculated also when the number of observations is smaller than the number of the financial assets. 
This scenario is of great interest in the framework of systemic risk for two main reasons. 
First, ideally the number of securities constituting a good approximation of the whole financial index is supposed to be very large, and our methodology allows to deal with a large number of assets. Second, our approach allows to work with short estimation windows containing only more recent data, which are necessary to obtain a systemic risk measure with a good descriptive content and a potential predictive power.
There is one further advantage that characterizes our proposal. The global balance indicator is naturally obtained as the aggregation of the local balance indices. Hence, these local indices measure the contribution of the single node in the correlation network to the global balance. While standard approaches to systemic risk usually measures the overall risk of the economic system, the approach based on the balance is basically local. Our future research will focus on the local balance to investigate the role of each asset to the general stability of the market.

\textbf{Data availability}
Data will be made available on request. The code used to compute the indicators used in this article can be found in the following repository:

\noindent 
\href{https://github.com/fernandodiazdiaz/Global\_balance-and\_Systemic\_Risk/tree/main}{https://github.com/fernandodiazdiaz/Global\_balance-and\_Systemic\_Risk/tree/main}

\textbf{Acknowledgments}
We thank the editor and anonymous reviewers for helpful comments and suggestions. FDD acknowledges funding from Ministerio de Ciencia e Innovacion, Agencia Estatal de Investigacion Program for Units of Excellences Maria de Maeztu (CEX2021$-$001164-M$/$10.13039$/$501100011033), from MDM$-2017-0711-20-2$ funded by MCIN$/$AEI$/$10.13039$/$\\501100011033 and by FSE invierte en tu futuro, as well as project APASOS (No. PID2021 $-$122256NB$-$C22).
RG and PU acknowledge financial support from the European Union – NextGenerationEU. Project PRIN 2022 “Networks:decomposition, clustering and community detection'' code: 2022NAZ0365 - CUP H53D23002510006. PB and RG are members of GNAMPA-INdAM.

\bibliographystyle{elsarticle-num}
\bibliography{References}

\begin{thebibliography}{10}
\expandafter\ifx\csname url\endcsname\relax
  \def\url#1{\texttt{#1}}\fi
\expandafter\ifx\csname urlprefix\endcsname\relax\def\urlprefix{URL }\fi
\expandafter\ifx\csname href\endcsname\relax
  \def\href#1#2{#2} \def\path#1{#1}\fi

\bibitem{31.giglio}
S.~Giglio, B.~Kelly, S.~Pruitt, Systemic risk and the macroeconomy: An
  empirical evaluation, Journal of Financial Economics 119~(3) (2016) 457--471.

\bibitem{26.diebold2014network}
F.~X. Diebold, K.~Y{\i}lmaz, On the network topology of variance
  decompositions: Measuring the connectedness of financial firms, Journal of
  Econometrics 182~(1) (2014) 119--134.

\bibitem{23.demirer}
M.~Demirer, F.~X. Diebold, L.~Liu, K.~Yilmaz, Estimating global bank network
  connectedness, Journal of Applied Econometrics 33~(1) (2018) 1--15.

\bibitem{adrian2011covar}
T.~Adrian, M.~K. Brunnermeier, Covar, Tech. rep., National Bureau of Economic
  Research (2011).

\bibitem{36.laeven2016bank}
L.~Laeven, L.~Ratnovski, H.~Tong, Bank size, capital, and systemic risk: Some
  international evidence, Journal of Banking \& Finance 69 (2016) S25--S34.

\bibitem{9.Bernal}
O.~Bernal, J.-Y. Gnabo, G.~Guilmin, Assessing the contribution of banks,
  insurance and other financial services to systemic risk, Journal of Banking
  \& Finance 47 (2014) 270--287.

\bibitem{37.lopez2012CoVaR}
G.~L{\'o}pez-Espinosa, A.~Moreno, A.~Rubia, L.~Valderrama, Short-term wholesale
  funding and systemic risk: A global covar approach, Journal of Banking \&
  Finance 36~(12) (2012) 3150--3162.

\bibitem{aldasoro2018multiplex}
I.~Aldasoro, I.~Alves, Multiplex interbank networks and systemic importance: An
  application to european data, Journal of Financial Stability 35 (2018)
  17--37.

\bibitem{14.Bisias}
D.~Bisias, M.~Flood, A.~W. Lo, S.~Valavanis, A survey of systemic risk
  analytics, Annual Reviews of Financial Economics 4~(1) (2012) 255--296.

\bibitem{45.rodriguez2013systemic}
M.~Rodr{\'\i}guez-Moreno, J.~I. Pe{\~n}a, Systemic risk measures: The simpler
  the better?, Journal of Banking \& Finance 37~(6) (2013) 1817--1831.

\bibitem{47.silva2017analysis}
W.~Silva, H.~Kimura, V.~A. Sobreiro, An analysis of the literature on systemic
  financial risk: A survey, Journal of Financial Stability 28 (2017) 91--114.

\bibitem{Billio2012}
M.~Billio, M.~Getmansky, A.~W. Lo, L.~Pelizzon, Econometric measures of
  connectedness and systemic risk in the finance and insurance sectors, Journal
  of Financial Economics 104~(3) (2012) 535--559, market Institutions,
  Financial Market Risks and Financial Crisis.

\bibitem{53.zheng2012changes}
Z.~Zheng, B.~Podobnik, L.~Feng, B.~Li, Changes in cross-correlations as an
  indicator for systemic risk, Scientific Reports 2~(1) (2012) 1--8.

\bibitem{52.zhang2020global}
D.~Zhang, D.~C. Broadstock, Global financial crisis and rising connectedness in
  the international commodity markets, International Review of Financial
  Analysis 68 (2020) 101239.

\bibitem{Mantegna1999}
R.~N. Mantegna, Hierarchical structure in financial markets, The European
  Physical Journal B - Condensed Matter and Complex Systems 11 (1999) 193--197.

\bibitem{Hasse2020}
J.-B. Hasse, Systemic risk: a network approach, AMSE Working Papers 63~(2025)
  (2020) 313--344.

\bibitem{Harary1953}
F.~Harary, On the notion of balance of a signed graph., Michigan Mathematical
  Journal 2~(2) (1953) 143--146.

\bibitem{Cartwright1956}
D.~Cartwright, F.~Harary, Structural balance: a generalization of heider's
  theory., Psychological review 63~(5) (1956) 277.

\bibitem{Kirkley2019}
A.~Kirkley, G.~T. Cantwell, M.~E.~J. Newman, Balance in signed networks, Phys.
  Rev. E 99 (2019) 012320.

\bibitem{Facchetti2011}
G.~Facchetti, G.~Iacono, C.~Altafini, Computing global structural balance in
  large-scale signed social networks, Proceedings of the National Academy of
  Sciences 108~(52) (2011) 20953--20958.

\bibitem{Gallo2023}
A.~Gallo, D.~Garlaschelli, R.~Lambiotte, F.~Saracco, T.~Squartini, Testing
  structural balance theories in heterogeneous signed networks, Commun Phys
  7~(154) (2024).

\bibitem{Kunegis2010}
J.~Kunegis, S.~Schmidt, A.~Lommatzsch, J.~Lerner, E.~W.~D. Luca, S.~Albayrak,
  Spectral Analysis of Signed Graphs for Clustering, Prediction and
  Visualization, SIAM, 2010.

\bibitem{Estrada2014}
E.~Estrada, M.~Benzi, Walk-based measure of balance in signed networks:
  Detecting lack of balance in social networks, Physical review. E,
  Statistical, nonlinear, and soft matter physics 90 (2014) 042802.

\bibitem{Aref2018}
S.~Aref, M.~C. Wilson, {Measuring partial balance in signed networks}, Journal
  of Complex Networks 6~(4) (2017) 566--595.

\bibitem{Talaga2023}
S.~Talaga, M.~Stella, T.~J. Swanson, A.~S. Teixeira, Polarization and
  multiscale structural balance in signed networks, Commun Phys 6~(349) (2023).

\bibitem{Altafini2013}
C.~Altafini, Consensus problems on networks with antagonistic interactions,
  Automatic Control, IEEE Transactions on 58 (2013) 935--946.

\bibitem{Li2015}
Y.~W. Yanhua~Li, Wei~Chen, Z.-L. Zhang, Voter model on signed social networks,
  Internet Mathematics 11~(2) (2015) 93--133.

\bibitem{Schaub2016}
M.~T. Schaub, N.~O'Clery, Y.~N. Billeh, J.-C. Delvenne, R.~Lambiotte,
  M.~Barahona, {Graph partitions and cluster synchronization in networks of
  oscillators}, Chaos: An Interdisciplinary Journal of Nonlinear Science 26~(9)
  (2016) 094821.

\bibitem{Fan2012}
P.~Fan, H.~Wang, P.~Li, W.~Li, Z.~Jiang, Analysis of opinion spreading in
  homogeneous networks with signed relationships, Journal of Statistical
  Mechanics: Theory and Experiment 2012~(08) (2012) P08003.

\bibitem{Lee2023}
K.-M. Lee, S.~Lee, B.~Min, K.-I. Goh, Threshold cascade dynamics on signed
  random networks, Chaos, Solitons \& Fractals 168 (2023) 113118.

\bibitem{DiazDiaz2025}
A.~Vendeville, F.~Diaz-Diaz, Modeling echo chamber effects in signed networks,
  Phys. Rev. E 111 (2025) 024302.

\bibitem{Neal2020}
Z.~P. Neal, A sign of the times? weak and strong polarization in the u.s.
  congress, 1973–2016, Social Networks 60 (2020) 103--112, social Network
  Research on Negative Ties and Signed Graphs.

\bibitem{Fraxanet2023}
F.~Emma, P.~Max, S.~Simon, G.~Vicenç, G.~David, Unpacking polarization:
  Antagonism and alignment in signed networks of online interaction, ArXiv (07
  2023).

\bibitem{Estrada2025}
F.~Diaz-Diaz, E.~Estrada, Signed graphs in data sciences via communicability
  geometry, Information Sciences 710 (2025) 122096.

\bibitem{Ruiz-Garcia2023}
M.~Ruiz-García, J.~Ozaita, M.~Pereda, A.~Alfonso, P.~Brañas-Garza, J.~A.
  Cuesta, A.~Sánchez, Triadic influence as a proxy for compatibility in social
  relationships, Proceedings of the National Academy of Sciences 120~(13)
  (2023) e2215041120.

\bibitem{bartesaghi2024d}
F.~Diaz-Diaz, P.~Bartesaghi, E.~Estrada, Mathematical modeling of local balance
  in signed networks and its applications to global international analysis,
  Journal of Applied Mathematics and Computing 70~(6) (2024) 6195--6218.

\bibitem{Zaslavsky1982}
T.~Zaslavsky, Signed graphs, Discrete Applied Mathematics 4~(1) (1982) 47--74.

\bibitem{Harary2002}
F.~Harary, M.~Lim, D.~C. Wunsch, Signed graphs for portfolio analysis in risk
  management, IMA Journal of Management Mathematics 13~(3) (2002) 201--210.

\bibitem{Uberti2020a}
S.~Figini, M.~Maggi, P.~Uberti, The market rank indicator to detect financial
  distress, Econometrics and Statistics 14 (2020) 63--73.

\bibitem{McDonald2008}
D.~McDonald, L.~Waterbury, R.~Knight, M.~D. Betterton, Activating and
  inhibiting connections in biological network dynamics, Biology Direct 3~(49)
  (2008).

\bibitem{Galam1996}
S.~Galam, Fragmentation versus stability in bimodal coalitions, Physica A:
  Statistical Mechanics and its Applications 230~(1) (1996) 174--188.

\bibitem{Pandey2021}
S.~G. Pandey, A model of signed network formation with heterogeneous players,
  Research in Economics 75~(1) (2021) 119--128.

\bibitem{Estrada2021}
E.~Ferreira, S.~Orbe, J.~Ascorbebeitia, B.~Alvarez~Pereira, E.~Estrada, Loss of
  structural balance in stock markets, Scientific Reports 11 (2021) 1--10.

\bibitem{Harary1957}
F.~Harary, Structural duality, Behavioral Science 2~(4) (1957) 255--265.

\bibitem{Zaslavsky2013}
T.~Zaslavsky, Matrices in the theory of signed simple graphs, Advances in
  Discrete Mathematics and Applications: Mysore, 2008 (ICDM-2008, Mysore,
  India) 13 (03 2013).

\bibitem{Lambiotte2024}
Y.~Tian, R.~Lambiotte, Spreading and structural balance on signed networks,
  SIAM Journal on Applied Dynamical Systems 23~(1) (2024) 50--80.

\bibitem{Lerman2024}
R.~Ghosh, K.~Lerman, T.~Surachawala, K.~Voevodski, S.~Teng, {Non-conservative
  diffusion and its application to social network analysis}, Journal of Complex
  Networks 12~(1) (2024) cnae006.

\bibitem{Acharya1980}
B.~D. Acharya, Spectral criterion for cycle balance in networks, Journal of
  Graph Theory 4~(1) (1980) 1--11.

\bibitem{Golub2013}
G.~H. Golub, C.~F. Van~Loan, Matrix computations, JHU press, 2013.

\bibitem{Belardo2019}
F.~Belardo, S.~M. Cioab{\u{a}}, J.~H. Koolen, J.~Wang, Open problems in the
  spectral theory of signed graphs, arXiv preprint arXiv:1907.04349 (2019).

\bibitem{Meucci2010}
A.~Meucci, Managing diversification, Risk 22 (2010) 74--79.

\bibitem{Jonnson1982}
D.~Jonsson, Some limit theorems for the eigenvalues of a sample covariance
  matrix, Journal of Multivariate Analysis 12~(1) (1982) 1--38.

\bibitem{Uberti2020b}
M.~Maggi, M.~Torrente, P.~Uberti, Proper measures of connectedness, Annals of
  Finance 16 (2020) 547--571.

\bibitem{Uberti2023}
C.~Pastorino, P.~Uberti, An empirical comparison of correlation-based systemic
  risk measures, {Quality \& Quantity} 2023/09/23 (2023).

\bibitem{Manganelli2001}
S.~Manganelli, R.~Engle, Value at risk models in finance, SSRN Electronic
  Journal (09 2001).

\end{thebibliography}

\section*{\Large Supplementary material}

\setcounter{section}{0}

\section{Another possible choice of the coefficient $\alpha(k)$ in Eqs. \eqref{diffusion1} and \eqref{diffusion2}}

We show here another analytically tractable choice for the coefficient $\alpha(k)$ in Eq. (2) of the main text. For simplicity, we only consider the limit $t\to\infty$. We assume throughout the derivation that the spectral radius of $\bf A$ is smaller than one. We begin with
\begin{equation*}
	{\bf G} = {\bf I}+\sum_{j=1}^\infty \left(\prod_{k=0}^{j-1}\alpha(k)\right){\bf A}^j,
\end{equation*}
with
\begin{equation*}
	\alpha(k)=\frac{(k+2)^2}{(k+2)^2-1}=\frac{(k+2)^2}{(k+1)(k+3)}.
\end{equation*}
We first compute the product. Define
\begin{equation*}
	P(j)=\prod_{k=0}^{j-1}\alpha(k)
	=\prod_{k=0}^{j-1}\frac{(k+2)^2}{(k+1)(k+3)}.
\end{equation*}
Trivial algebraic manipulations show that
\begin{equation*}
	P(j)  = \frac{2\,((j+1)!)^2}{\,j!(j+2)!} = \frac{2(j+1)}{j+2}.
\end{equation*}
Substituting \(P(j)\) into $\bf G$ and doing some algebra, we get to:
\begin{align*}
	{\bf G} = \textbf{I} + 2\sum_{j=1}^\infty \left(1 - \frac{1}{j+2} \right){\bf A}^j= \textbf{I} + 2{\bf A}(\textbf{I}-{\bf A})^{-1} - 2\sum_{j=1}^\infty \frac{{\bf A}^j}{j+2}.
\end{align*}
We evaluate the remaining sum:
\begin{equation*}
	\sum_{j=1}^\infty \frac{{\bf A}^j}{j+2} = {\bf A}^{-2}\sum_{j=3}^\infty \frac{{\bf A}^j}{j} = {\bf A}^{-2}\Bigl(-\ln(\textbf{I}-{\bf A})-{\bf A}-\frac{{\bf A}^2}{2}\Bigr),
\end{equation*}
where we have used the Taylor expansion of the logarithm. 
Substituting back into the expression for \({\bf G}\):
\begin{equation*}
	{\bf G} = 2\textbf{I} + 2{\bf A}(\textbf{I}-{\bf A})^{-1}+2{\bf A}^{-1} +2{\bf A}^{-2}\ln(\textbf{I}-{\bf A}).
\end{equation*}

\section{Proof of Proposition \ref{proposition1}}
\setcounter{proposition}{0}
\begin{proposition}
	If $G$ is a balanced or antibalanced signed network, then ${\mathscr R}({\bf A})=1$. 
\end{proposition}
\begin{proof}
	First recall that, in general, the $ij$- element of the $k$ power of the adjacency matrix returns the number of signed walks between nodes $i$ and $j$, that is the difference between the number of positive and negative walks of length $k$. Moreover, a closed walk is said to be positive (negative) if the product of the edge signs is positive (negative). Then observe that ${\rm tr}[e^{{\bf A}}]=\sum_{i=1}^{N}[\sum_{k}\frac{1}{k!}{\bf A}^{k}]_{ii}=\sum_{i=1}^{N}[\sum_{k}\frac{1}{k!}(W^{+}_{k,i}-W^{-}_{k,i})]$, where $W^{\pm}_{k,i}$ is the number (weight) of positive and negative closed walks, respectively, of length $k$ starting from node $i\in V$ and ending at the same node. 
	Then ${\rm tr}[e^{{\bf A}}]=\sum_{k}\frac{1}{k!}(\sum_{i=1}^{N}W^{+}_{k,i}-\sum_{i=1}^{N}W^{-}_{k,i})=\sum_{k}\frac{1}{k!}(W^{+}_{k}-W^{-}_{k})$, where $W^{\pm}_{k}$ is the \textit{total} number (weight) of positive and negative walks of length $k$ starting from any node and ending at the same node, summed over all the nodes in the network $G$. Similarly, for the network $|G|$ we have ${\rm tr}[e^{{\bf |A|}}]=\sum_{k}\frac{1}{k!}(W^{+}_{k}+W^{-}_{k})$. Therefore:

	\begin{align*}
		{\mathscr R}({\bf A})
		&=\frac{{\rm tr}[e^{{\bf A}}]\cdot {\rm tr}[e^{-{\bf A}}]}
		{{\rm tr}[e^{|{\bf A}|}]\cdot {\rm tr}[e^{-|{\bf A}|}]}
		=\frac{{\rm tr}[\sum_{k}\frac{1}{k!}{\bf A}^{k}]\cdot {\rm tr}[\sum_{k}\frac{1}{k!}(-{\bf A})^{k}]}
		{{\rm tr}[\sum_{k}\frac{1}{k!}|{\bf A}|^{k}]\cdot {\rm tr}[\sum_{k}\frac{1}{k!}(-|{\bf A}|)^{k}]}\\
		&=\frac{\sum_{k}\frac{1}{k!}\left[W^{+}_{k}-W^{-}_{k}\right]
			\cdot \sum_{k}\frac{(-1)^{k}}{k!}\left[W^{+}_{k}-W^{-}_{k}\right]}
		{\sum_{k}\frac{1}{k!}\left[W^{+}_{k}+W^{-}_{k}\right]\cdot \sum_{k}\frac{(-1)^{k}}{k!}\left[W^{+}_{k}+W^{-}_{k}\right]}\\
		&=\frac{
			\left[
			\sum_{k}\frac{1}{2k!}\left(W^{+}_{2k}-W^{-}_{2k}\right)
			+
			\sum_{k}\frac{1}{(2k+1)!}\left(W^{+}_{2k+1}-W^{-}_{2k+1}\right)
			\right]
		}
		{
			\left[
			\sum_{k}\frac{1}{2k!}\left(W^{+}_{2k}+W^{-}_{2k}\right)
			+
			\sum_{k}\frac{1}{(2k+1)!}\left(W^{+}_{2k+1}+W^{-}_{2k+1}\right)
			\right]
		}\cdot \\
		&\quad \ \frac{
			\left[
			\sum_{k}\frac{1}{2k!}\left(W^{+}_{2k}-W^{-}_{2k}\right)
			-
			\sum_{k}\frac{1}{(2k+1)!}\left(W^{+}_{2k+1}-W^{-}_{2k+1}\right)
			\right]
		}
		{
			\left[
			\sum_{k}\frac{1}{2k!}\left(W^{+}_{2k}+W^{-}_{2k}\right)
			-
			\sum_{k}\frac{1}{(2k+1)!}\left(W^{+}_{2k+1}+W^{-}_{2k+1}\right)
			\right]
		}\\
		&=\frac{
			\left[
			\sum_{k}\frac{1}{2k!}\left(W^{+}_{2k}-W^{-}_{2k}\right)
			\right]^{2}
			-
			\left[
			\sum_{k}\frac{1}{(2k+1)!}\left(W^{+}_{2k+1}-W^{-}_{2k+1}\right)
			\right]^{2}
		}
		{
			\left[
			\sum_{k}\frac{1}{2k!}\left(W^{+}_{2k}+W^{-}_{2k}\right)
			\right]^{2}
			-
			\left[
			\sum_{k}\frac{1}{(2k+1)!}\left(W^{+}_{2k+1}+W^{-}_{2k+1}\right)
			\right]^{2}
		}\\
		&=\frac{
			\left[
			W^{+}_{\rm even}-W^{-}_{\rm even}
			\right]^{2}
			-
			\left[
			W^{+}_{\rm odd}-W^{-}_{\rm odd}
			\right]^{2}
		}
		{
			\left[
			W^{+}_{\rm even}+W^{-}_{\rm even}
			\right]^{2}
			-
			\left[
			W^{+}_{\rm odd}+W^{-}_{\rm odd}
			\right]^{2}
		}\\
	\end{align*}
	where $W^{\pm}_{\rm even}$ and $W^{\pm}_{\rm odd}$ denote the total penalized weight of positive and negative closed walks of even and odd length, respectively. For a balanced graph, by definition, there are no negative closed walks and $W^{-}_{\rm even}=W^{-}_{\rm odd}=0$. For an antibalanced graph, by definition, every even closed walk is positive and every odd closed walk is negative, so $W^{-}_{\rm even}=W^{+}_{\rm odd}=0$. The claim follows.
\end{proof}	

\section{Results for random correlation matrices}\label{secA1}

We provide here some exact results on the asymptotic behavior of the sums in Eq. \eqref{ratio} in the case of purely random correlation matrices.
Specifically, we prove two analytical results about the two sums involved in the condition number on the signed network $G$, that is $\mathscr{K}(e^{-{\bf A}})=\left( \sum_{j=1}^{N}e^{-\lambda_{j}} \right)\left( \sum_{j=1}^{N}e^{\lambda_{j}} \right)$ for independent normal variables $\tilde{X}_{it}$ with zero mean and unit variance.

\begin{proposition}
	\label{propositionA1}
	Let $\tilde{X}_{it}$ be independent normal variables with zero mean and unit variance and let $\lambda_{1}>\dots >\lambda_{n}>0$ the eigenvalues of their correlation matrix $\bf C$. Then, for $N\to +\infty$, $T\to +\infty$, $N/T\to 1$
	\begin{equation*}
		\sum_{i=1}^{N}e^{\lambda_{i}} \sim I_{2}(2) e^2 N \sim 5.090679 N
	\end{equation*}
	and
	\begin{equation*}
		\sum_{i=1}^{N}e^{-\lambda_{i}} \sim \frac{I_{0}(2)+I_{1}(2)}{e^2}N \sim 0.523778 N
	\end{equation*}
	where $I_{n}(z)=\left( \frac{z}{2} \right)^{n}\sum_{k=0}^{\infty}\frac{\left(\frac{z^2}{4}\right)^k}{k! \Gamma(n+k+1)}$ are the modified Bessel functions of the first kind.
\end{proposition}
\begin{proof}
	The $k-$moment $M_{k}$ of the graph $G$ shows the following asymptotic behavior (see \cite{Jonnson1982})
	\begin{equation*}
		M_{k}=\frac{1}{N}\sum_{i=1}\lambda_{i}^{k} \sim \binom{2k}{k}\frac{1}{k+1}=\frac{(2k)!}{k!(k+1)!}
	\end{equation*}
	Then we have
	\begin{equation*}
		\begin{split}
			\sum_{i=1}^{N}e^{-\lambda_{i}}&=
			\sum_{i=1}^{N}\sum_{k=0}^{\infty} (-1)^k\frac{1}{k!}\lambda_{i}^{k}
			=\sum_{k=0}^{\infty}(-1)^k\frac{1}{k!}\sum_{i=1}^{N}\lambda_{i}^{k}\\
			&=N\sum_{k=0}^{\infty}(-1)^k\frac{1}{k!}M_{k}
			=N\sum_{k=0}^{\infty}(-1)^k\frac{1}{k!}\frac{(2k)!}{k!(k+1)!}\\
			&=N\sum_{k=0}^{\infty}(-1)^k\frac{(2k)!}{(k!)^3(k+1)}
			=N \frac{I_{0}(2)+I_{1}(2)}{e^2}
		\end{split}
	\end{equation*}
\end{proof}
As a consequence, the condition number grows, under these assumptions, as $2.67 N^2$, that is  $\mathscr{K}(e^{-{\bf A}})\sim O(N^2)$. A similar result can hardly be replicated for the absolute value matrix, i.e., for the network $|G|$. Nonetheless, we can investigate the statistical correlation between the ratio of the two condition numbers ${\mathscr R}({\bf A})=\mathscr{K}(e^{-{\bf A}})/\mathscr{K}(e^{-|{\bf A}|})$ and the global balance $\kappa(G)$. In particular, we generate random correlation matrices\footnote{Random matrices are generated by \textsc{randcorr} package in R.} of increasing size $N=5$, $N=10$, $N=20$, $N=50$, and we consider for each dimension 100 repetitions.
In Fig. \ref{fig8}, we report the scatterplot of the values of  ${\mathscr R}({\bf A})$ and $\kappa(G)$.   
As can be seen, the values of the two indicators correlate very strongly and increasingly as the network size, i.e. the number of assets involved, increases.
\begin{figure}[H]
	\centering
	\subfloat[]{\includegraphics[width=0.45\textwidth]{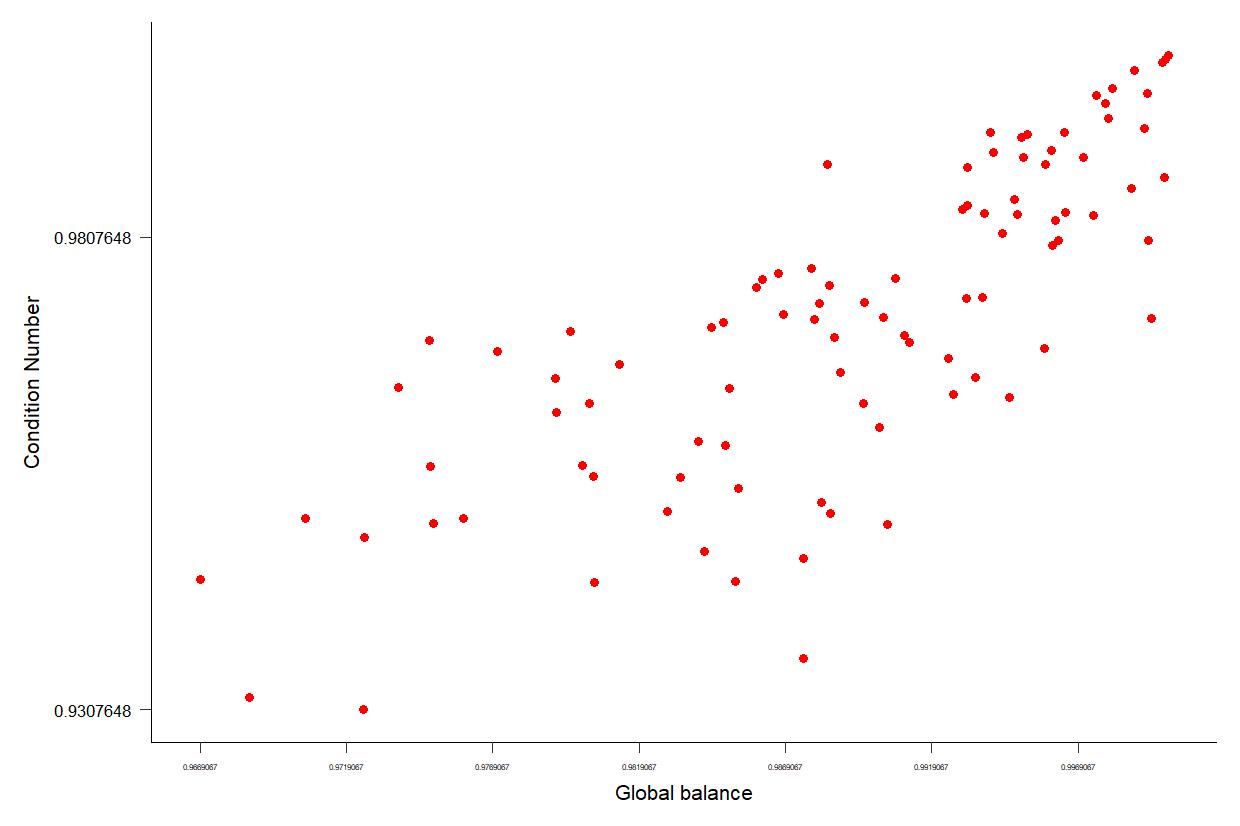}}
	\subfloat[]{\includegraphics[width=0.45\textwidth]{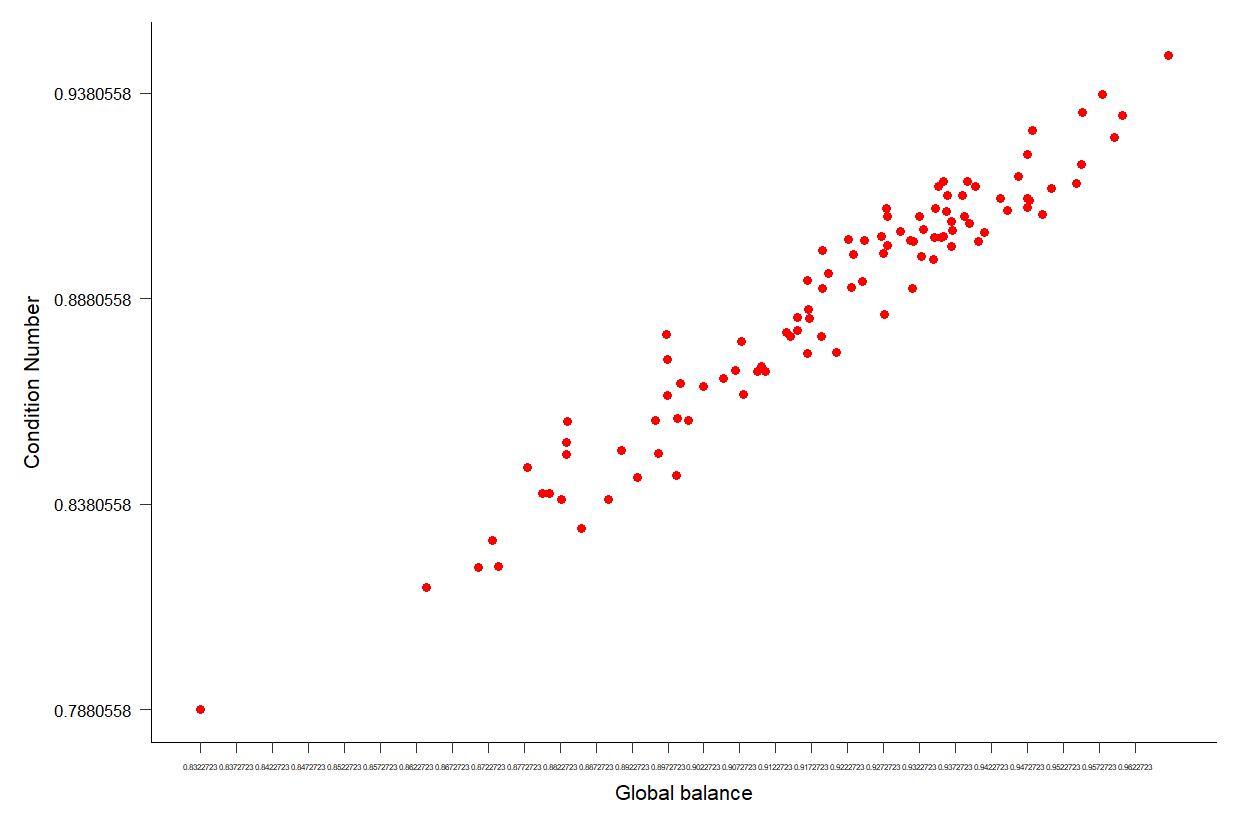}}\\
	\subfloat[]{\includegraphics[width=0.45\textwidth]{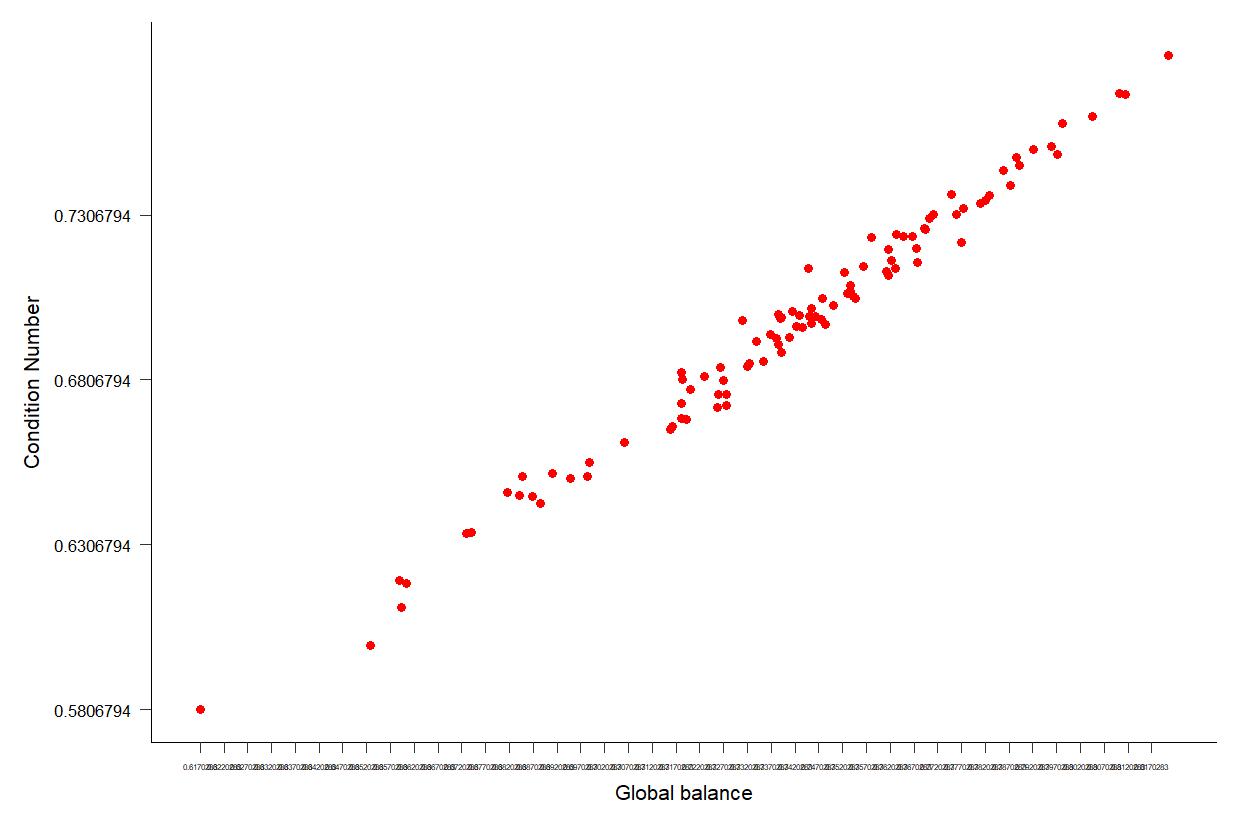}}
	\subfloat[]{\includegraphics[width=0.45\textwidth]{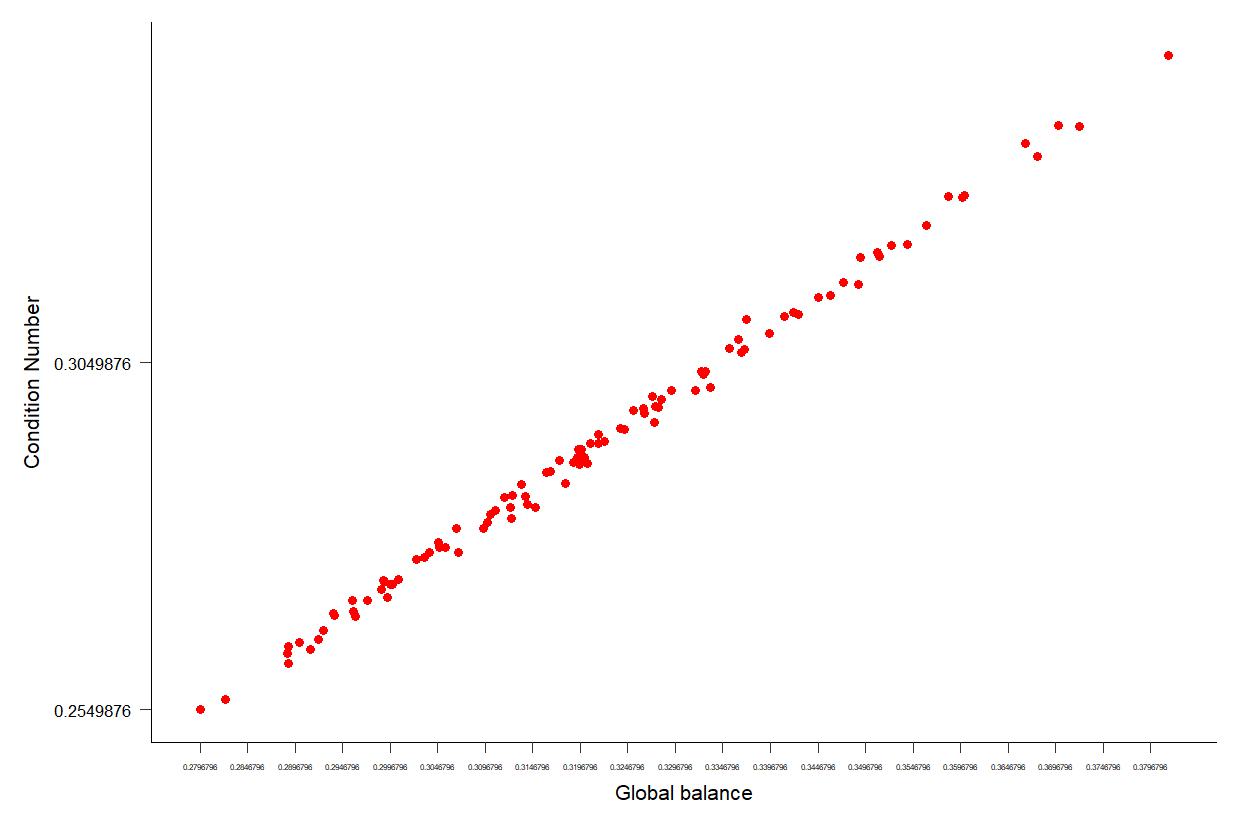}}
	\caption{Increasing correlation between the ratio ${\mathscr R}({\bf A})$ and the global balance $\kappa(G)$ for (a) $N=5$, $\rho=0.8546633$,  (b) $N=10$, $\rho=0.9693817$,  (c) $N=20$, $\rho=0.9918888$,  (d) $N=50$, $\rho=0.9984123$}
	\label{fig8} 
\end{figure}

\section{Global balance and average correlation}\label{secA2}

In Fig. \ref{fig9}, we show the temporal evolution of the global balance $\kappa(G)$ and of the average correlation for the S\&P500 dataset. In Fig. \ref{fig9} we have split the entire interval into two sub-intervals, 2005-2016 in panel (a) and 2016-2020 in panel (b), in order to highlight with higher resolution the information carried by the indicators in the two different periods. In Fig. \ref{fig10}, we zoom in on the period 2008-2014, when the global balance undergoes small relative changes.
\begin{figure}[H]
	\centering
	\subfloat[]{\includegraphics[width=0.80\textwidth] 
		{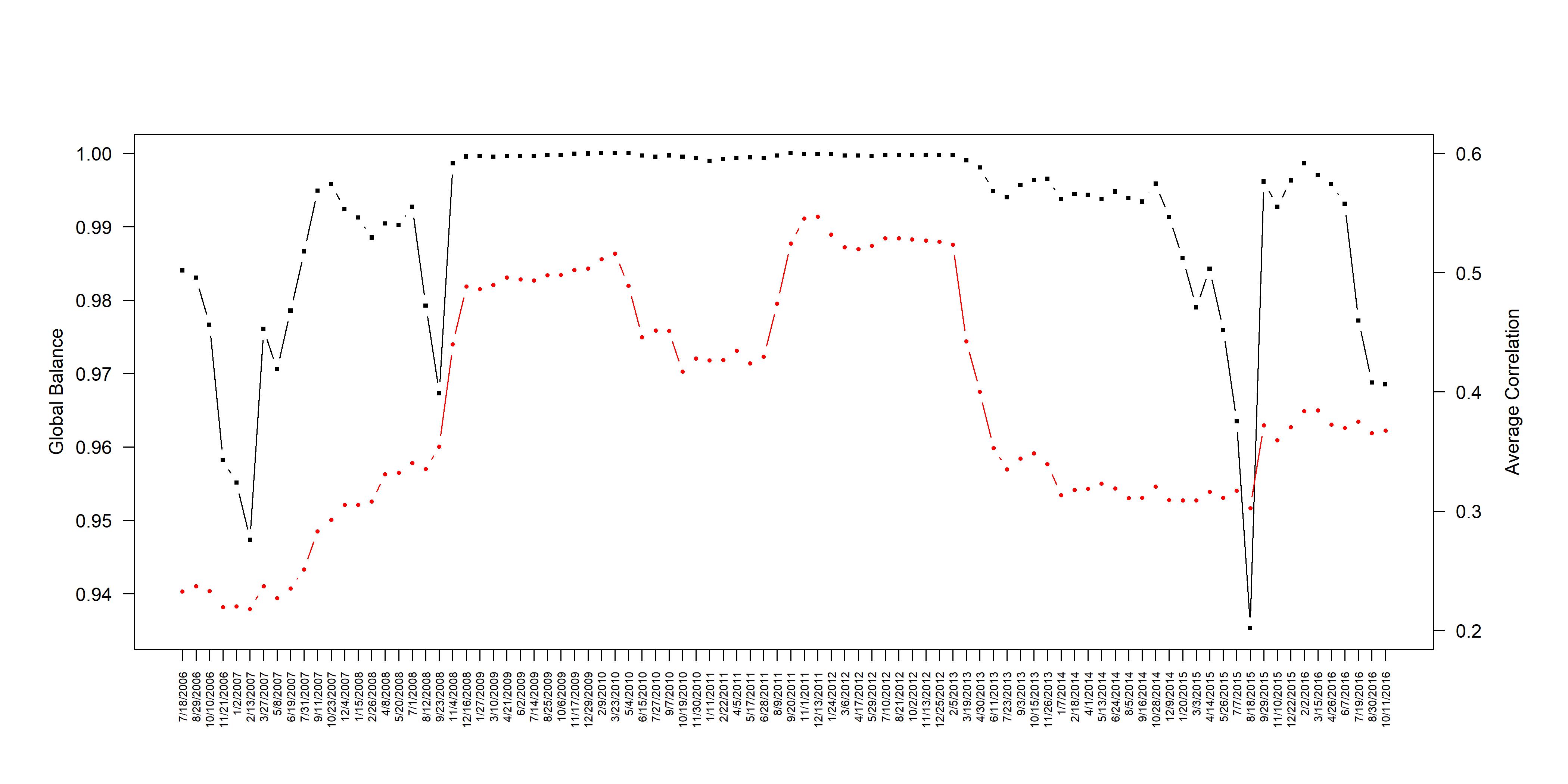}}\\
	\subfloat[]{\includegraphics[width=0.80\textwidth]{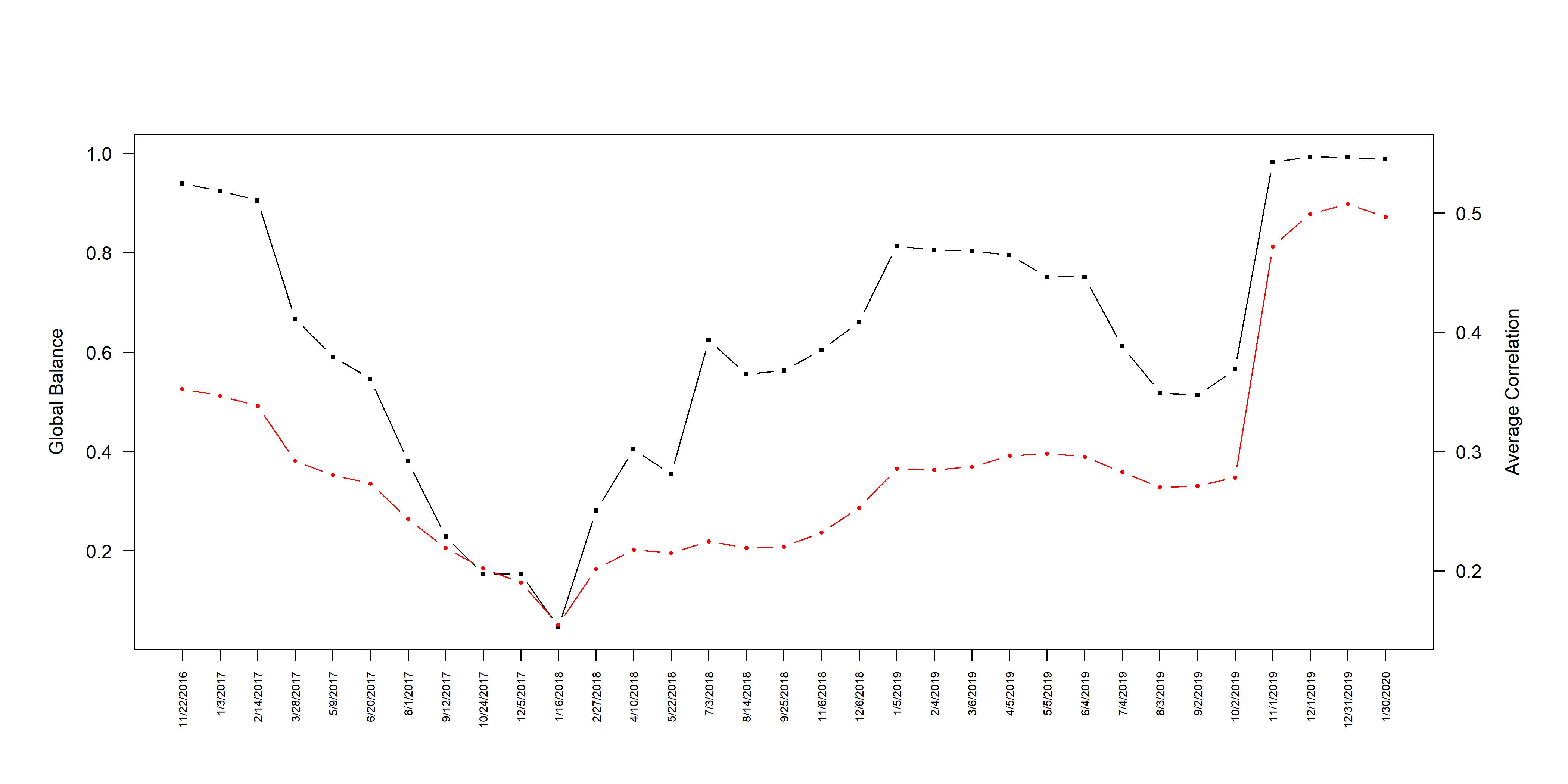}}\\
	\caption{Global Balance (square dots black line) and Average Correlation (circular dots red lines) for the S\&P500 daily returns with $\Delta T=400$ and $\Delta t=30$, for the period (a) 2005-2016 and (c) 2016-2020.}
	\label{fig9} 
\end{figure}

\begin{figure}[H]
	\centering
	\includegraphics[width=0.80\textwidth]{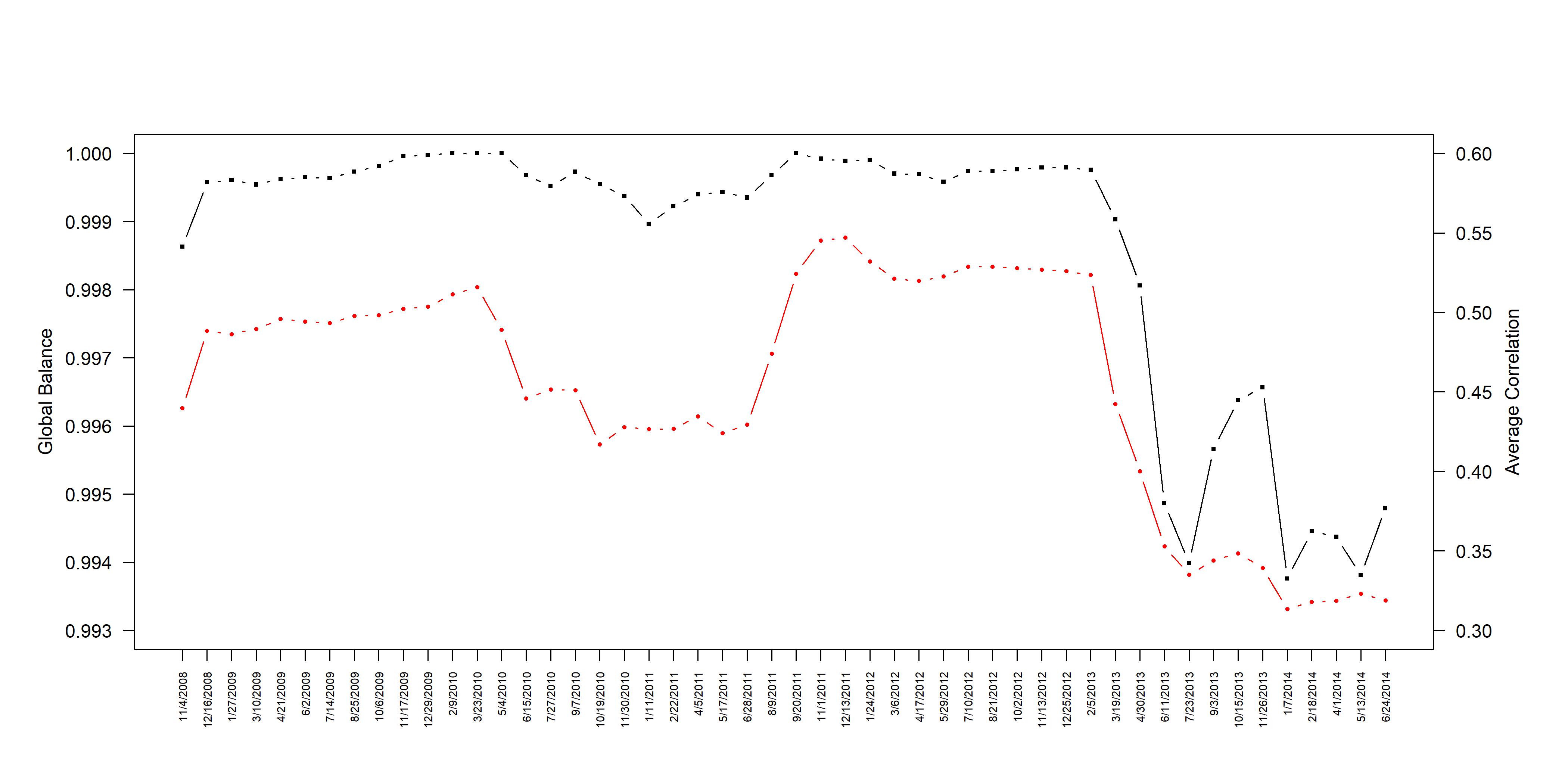}
	\caption{Global Balance (square dots black line) and Average Correlation (circular dots red lines) for the S\&P500 daily returns with $\Delta T=400$ and $\Delta t=30$, zoom on the period 2008-2014.}
	\label{fig10} 
\end{figure}

In Figs. \ref{fig11} and \ref{fig12}, we show the temporal evolution of the global balance $\kappa(G)$ and of the average correlation for the subset of $50$ assets in the S\&P500 dataset, considered in the main text, again over the whole time interval and the partitioned time interval, respectively.
\begin{figure}[H]
	\centering
	\includegraphics[width=0.80\textwidth]{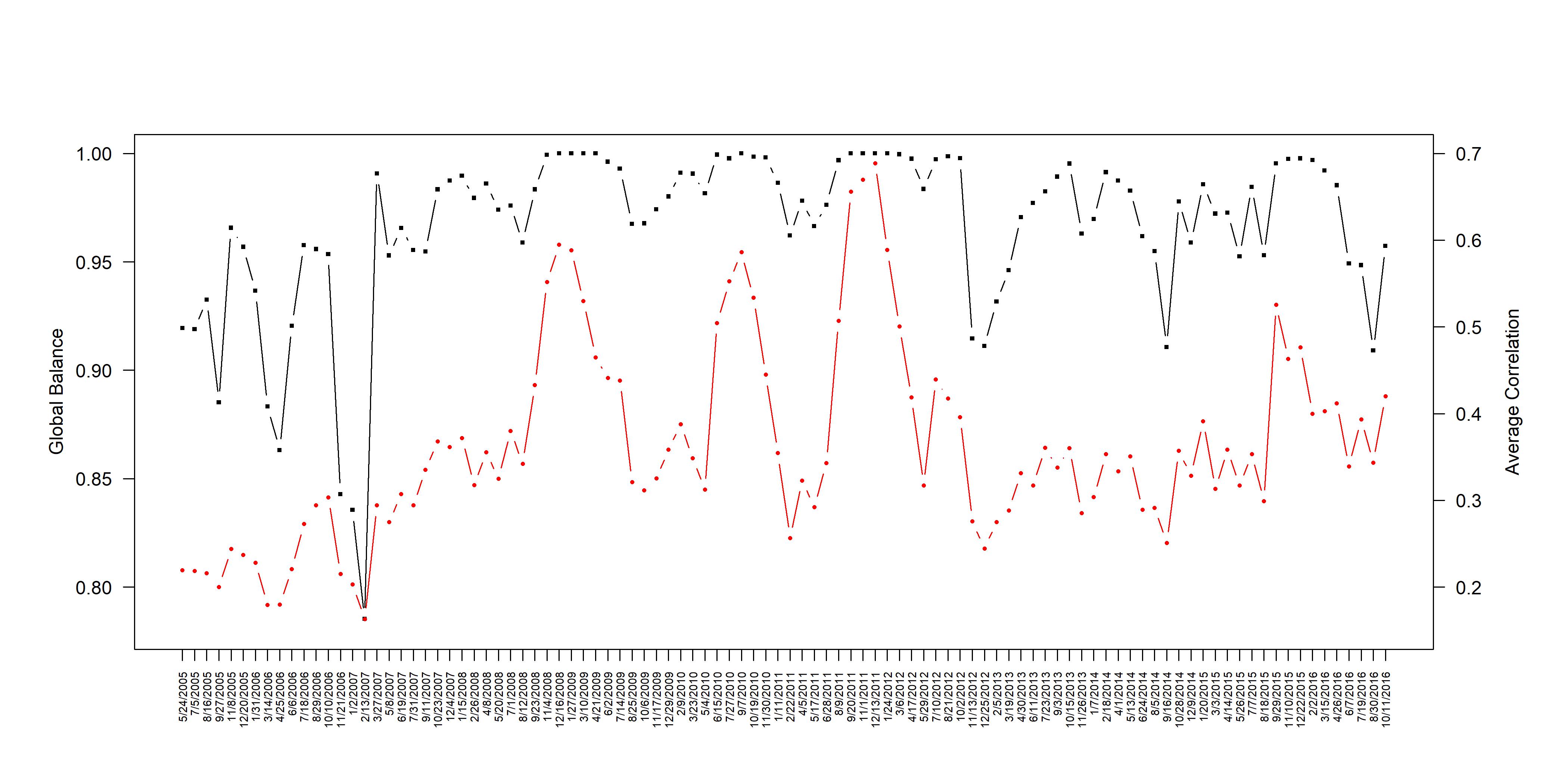}
	\caption{Global Balance (square dots black line) and Average Correlation (circular dots red lines), a sample of $50$ assets from S\&P500 daily returns, $\Delta T=100$ and $\Delta t=30$, for the period 2005-2016.}
	\label{fig11} 
\end{figure}
\begin{figure}[H]
	\centering
	\includegraphics[width=0.80\textwidth]{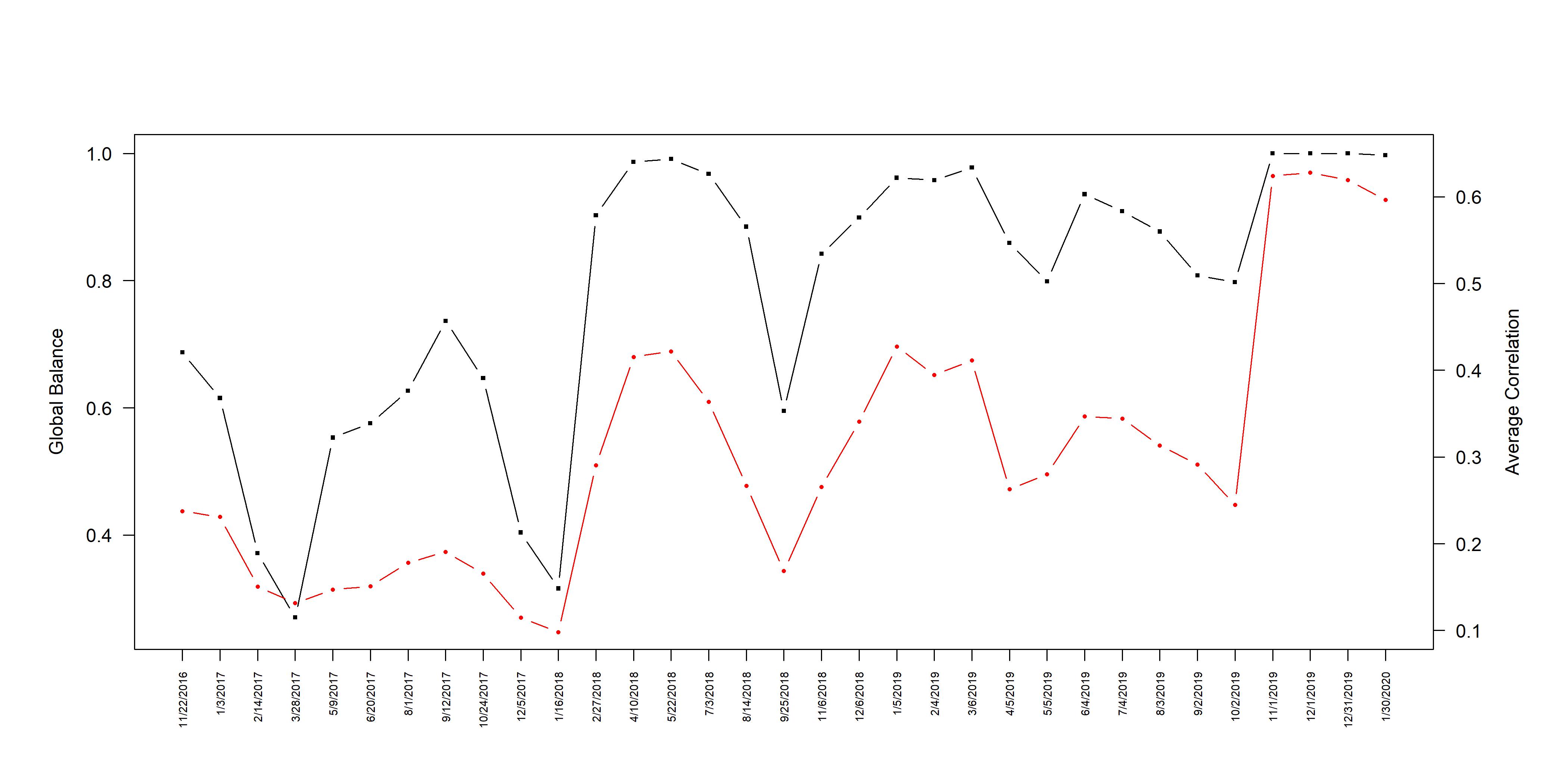}
	\caption{Global Balance (square dots black line) and Average Correlation (circular dots red lines), a sample of $50$ assets from S\&P500 daily returns, $\Delta T=100$ and $\Delta t=30$, for the period 2016-2020.}
	\label{fig12} 
\end{figure}

\section{Mean of the smoothed conditional distributions of the average returns.}\label{secA3}
\begin{table}[H]
	\footnotesize
	\begin{center}
		\begin{tabular}{c|c|c|c|c|c|}
			\cline{2-6}
			& \multicolumn{5}{c|}{Mean of the smoothed conditional distributions} \tabularnewline \cline{1-6}
			\multicolumn{1}{|l|}{\bf Interval} & $0<\kappa(G)\leq 0.2$ & $0.2<\kappa(G)\leq 0.4$ & $0.4<\kappa(G)\leq 0.6$ & $0.6<\kappa(G)\leq 0.8$ & $0.8<\kappa(G)\leq 1.0$  \tabularnewline \hline
			\multicolumn{1}{|c|}{\bf Weighted}
			& $0.001228899$
			& $0.001604856$
			& $0.0006486389$
			& $-0.001496206$
			& $-0.005628693$
			\tabularnewline \hline
			\multicolumn{1}{|c|}{\bf Binary}
			& $0.001228899$
			& $0.00232083$
			& $0.001939107$
			& $-0.001496206$
			& $-0.005628693$
			\tabularnewline \hline
			\hline
			\multicolumn{1}{|l|}{\bf Interval} & $0<\kappa(G)\leq 0.5$ & $0.5<\kappa(G)\leq 0.8$ & $0.8<\kappa(G)\leq 0.9$ & $0.9<\kappa(G)\leq 0.99$ & $0.99<\kappa(G)\leq 1.0$  \tabularnewline \hline
			\multicolumn{1}{|c|}{\bf Weighted} 
			& $0.001604856$
			& $-0.001496206$
			& $-0.003188417$
			& $-0.00783892$
			& $-0.00977485$
			\tabularnewline \hline
			\multicolumn{1}{|c|}{\bf Binary} 
			& $0.001604856$
			& $-0.001496206$
			& $-0.005236051$
			& $-0.005628693$
			& $-0.009021566$
			\tabularnewline \hline
		\end{tabular}
	\end{center}
	\caption{Mean of the average returns in the different balance bands. Weighted and binary refer to the two versions of the correlation networks.}
	\label{Table3}
\end{table}

\section{Eurostoxx 50 index analysis}\label{secA4}
We check the robustness of our results using a different database than the one used in the main text. We refer here to the constituents of the eurostoxx 50 index (ESX50). The cleaned dataset contains the daily returns of $42$ stocks from January 5, 2005 to September 17, 2020. We perform the test over windows of $45$ days. In Fig. \ref{fig13} we plot the mean, maximum and minimum average returns of the constituents, panel (a), and a zoom on the mean, panel (b). Due to the larger size of the time windows, we detect less but more severe events. In this new European dataset, we specifically identify the systemic crisis associated with COVID-19 in January 2020 with a threshold $\tau=-0.010$.
Fig. \ref{fig14} and Table \ref{Table4} illustrate the conditional smoothed distributions, the boxplots and the corresponding mean values for the weighted version of the correlation networks. We do not show the results for the binary network for the present example. The numerical results confirm the previous findings, reinforcing the interpretation of the global balance as an effective systemic risk indicator.
\begin{figure}[H]
	\centering
	\subfloat[]{\includegraphics[width=0.80\textwidth]{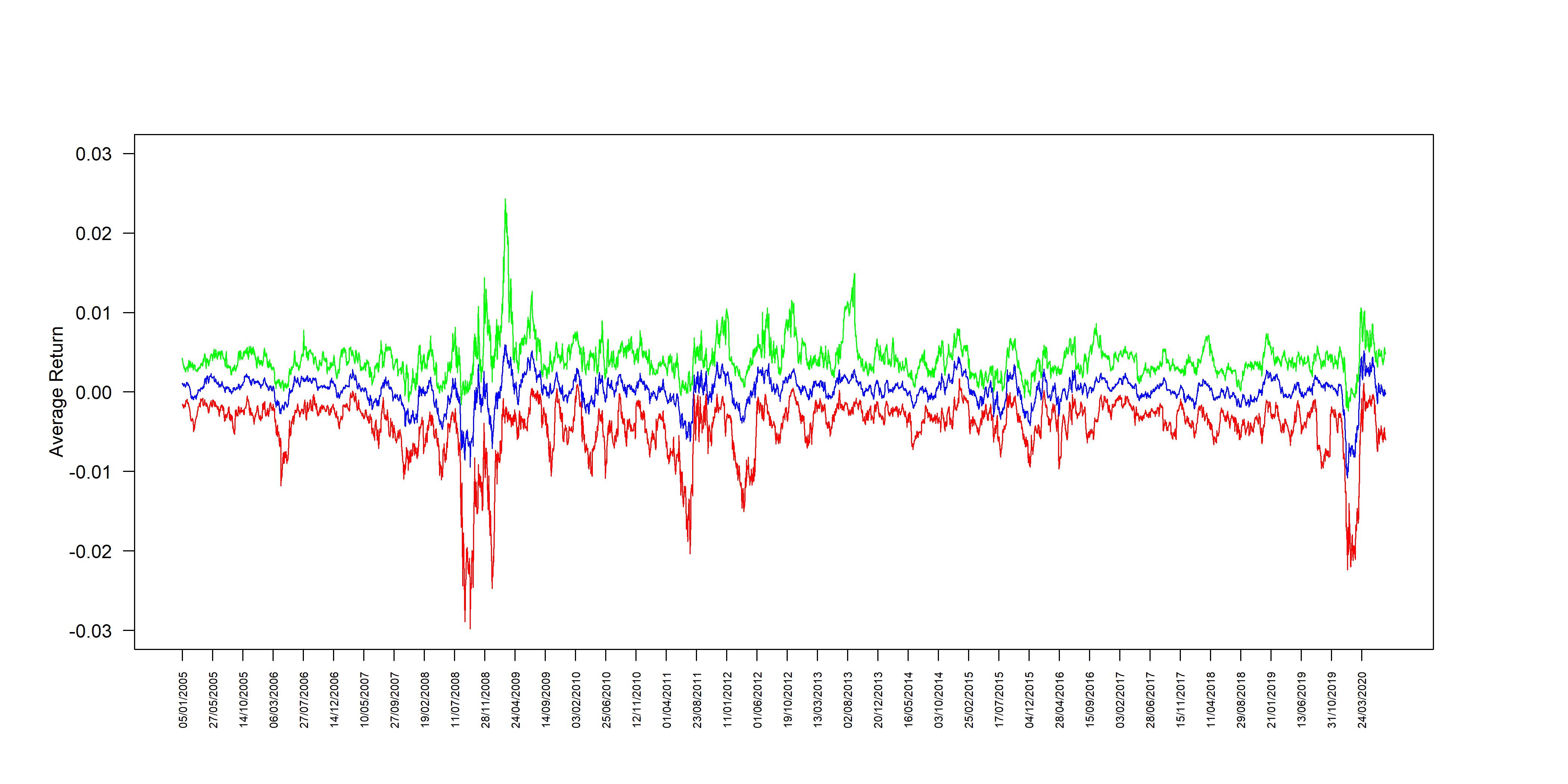}}\\
	\subfloat[]{\includegraphics[width=0.80\textwidth]{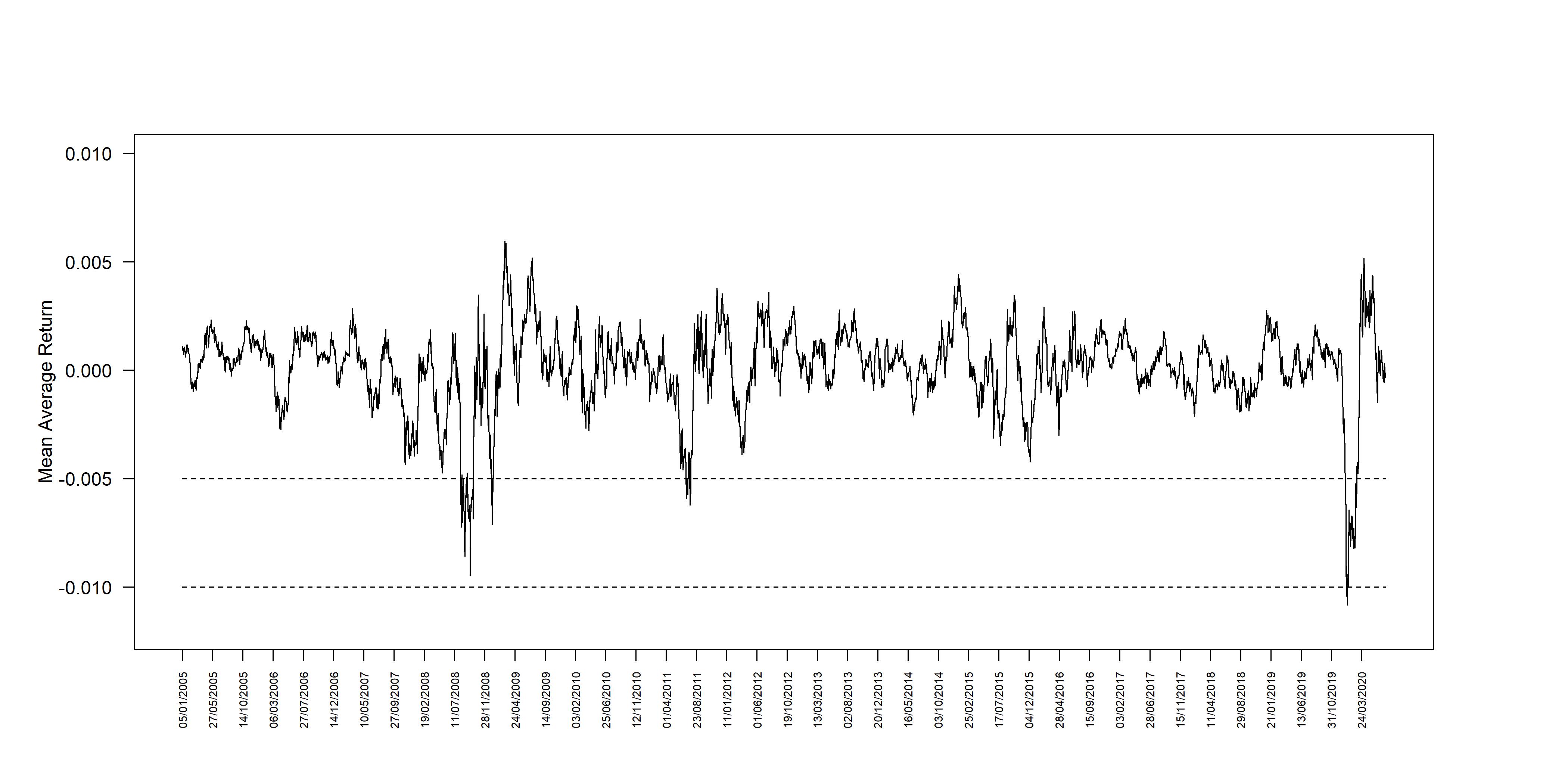}}
	\caption{Panel (a): Average returns of the eurostoxx 50 index (ESX50) in the period 2005-2020. The average returns of each asset have been computed over a $20$-days-wide sliding window with step $1$ day. The blue line represents the mean over all the assets in the same window of the average returns. The green and the red lines represent the max and the min, respectively, of the average returns in the same window. Panel (b) focuses on the mean of the average returns and highlight crisis events when the line goes below the two possible threshold represented by the two horizontal dashed lines.}
	\label{fig13} 
\end{figure}

\begin{figure}[H]
	\centering
	\subfloat[]{\includegraphics[width=0.80\textwidth]{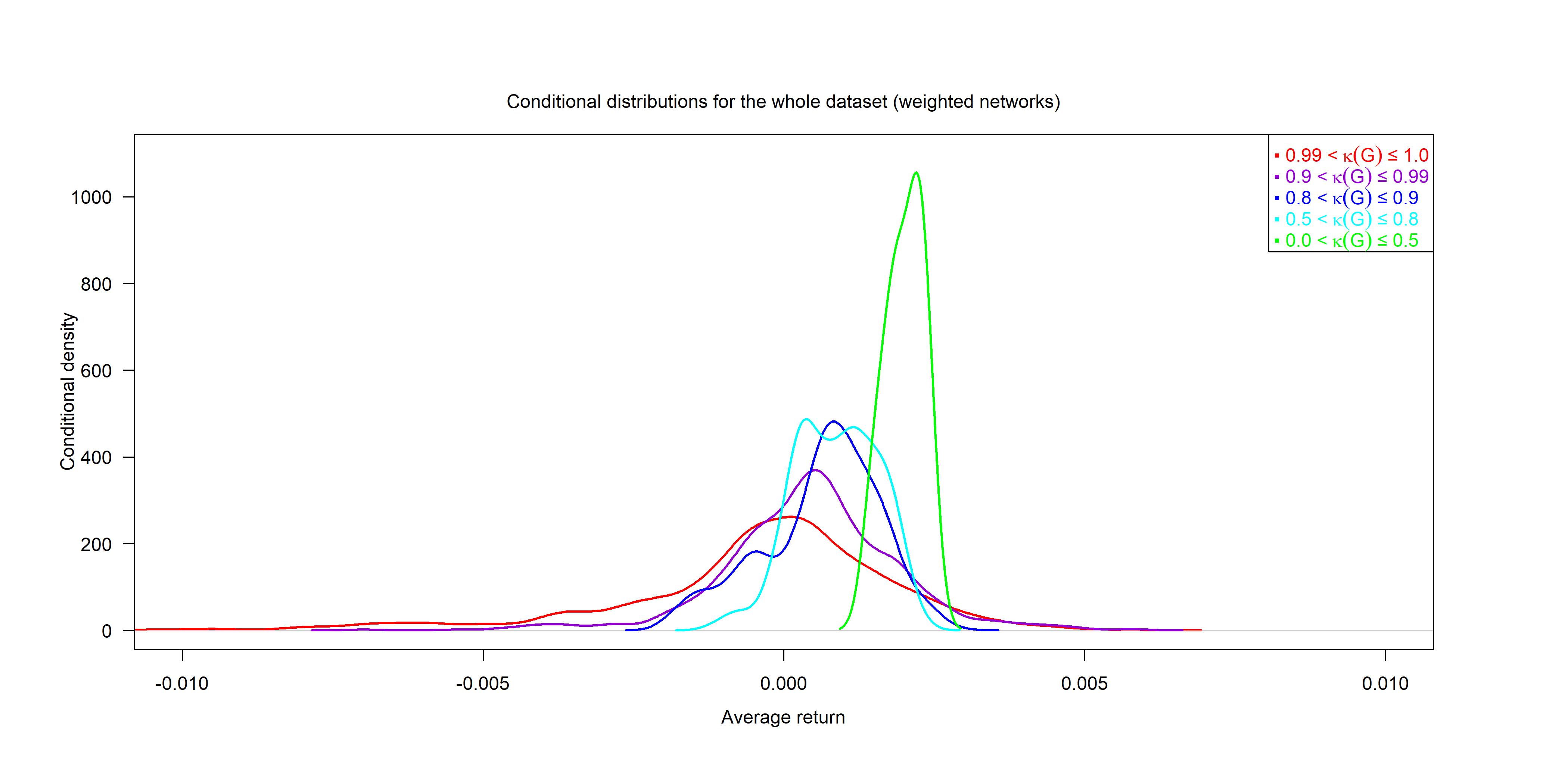}}\\
	\subfloat[]{\includegraphics[width=0.42\textwidth]{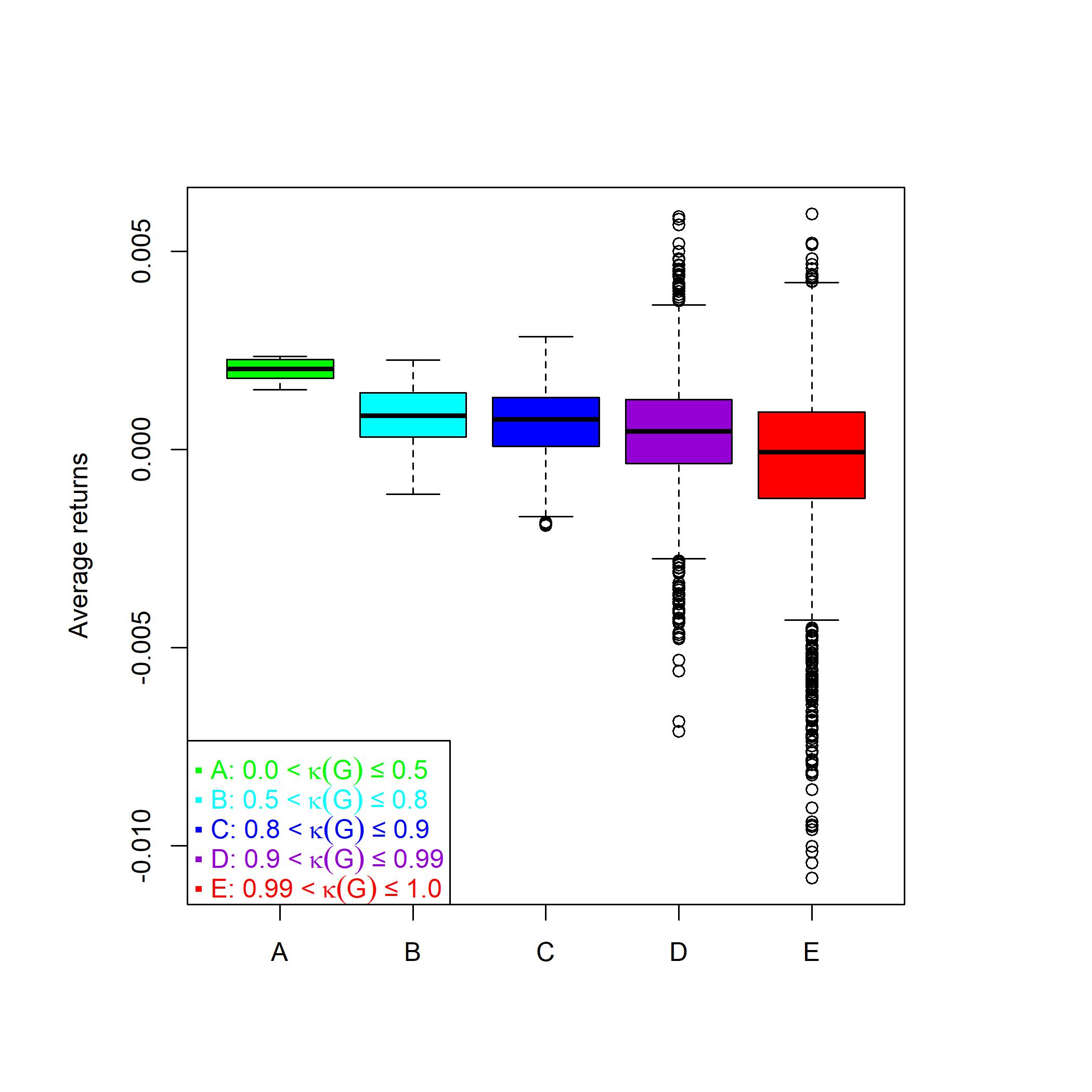}}
	\caption{Densities, panel (a), and boxplot, panel (b), of the conditional distributions of the mean average returns of the 42 constituents of the ESX dataset.}
	\label{fig14} 
\end{figure}

\begin{table}[H]
	\footnotesize
	\begin{center}
		\begin{tabular}{c|c|c|c|c|c|}
			\hline
			\multicolumn{1}{|l|}{\bf Interval} & $0<\kappa(G)\leq 0.5$ & $0.5<\kappa(G)\leq 0.8$ & $0.8<\kappa(G)\leq 0.9$ & $0.9<\kappa(G)\leq 0.99$ & $0.99<\kappa(G)\leq 1.0$  \tabularnewline \hline
			\multicolumn{1}{|c|}{\bf Mean density} 
			& $0.001932035$
			& $0.000561916$
			& $0.0004691143$
			& $-0.0006158585$
			& $-0.002433815$
			\tabularnewline \hline
			\multicolumn{1}{|c|}{\bf Mean values} 
			& $0.001999696$
			& $0.0008673605$
			& $0.000616936$
			& $0.0004399794$
			& $-0.0003903168$
			\tabularnewline \hline
		\end{tabular}
	\end{center}
	\caption{Mean of the real average returns (mean values) and mean of the smoothed densities (mean density) in the different balance bands for the ESX dataset in the weighted version of the correlation networks.}
	\label{Table4}
\end{table}

\end{document}